%% file: bare_jrnl_compsoc.tex
\documentclass[10pt,journal,compsoc]{IEEEtran}
\usepackage{cite}
\usepackage{amsmath,amssymb,amsfonts}
\usepackage{algorithm}
\usepackage{graphicx}
\usepackage{textcomp}
\usepackage{xcolor}

\ifCLASSINFOpdf
\else
\fi
\newcommand{\stitle}[1]{\vspace{1ex} \noindent{\bf #1}}
\newcommand{\gr}{\emph{EGR}}
\newcommand{\ls}{\emph{ALS}}
\newcommand{\qd}{\emph{query drift }}
\newcommand{\qds}{\emph{query-drifted vertices }}
\usepackage{subfigure}
\usepackage[shortlabels]{enumitem}
\usepackage{multirow}
\usepackage[noend]{algpseudocode}
\usepackage{algorithmicx,algorithm}
\usepackage{balance}
\usepackage{marginnote}
\usepackage{multicol}
\usepackage{times}
\usepackage{booktabs} % For formal tables
\usepackage{bbm}
\usepackage{epstopdf}
\usepackage{url}
\usepackage{multirow}
\usepackage{subfigure}

\usepackage{amssymb}  

\newtheorem{definition}{Definition}

\newenvironment{proof}{\quad{\it Proof:}}{\hfill $\square$\par}  

\newtheorem{theorem}{Theorem}[section]
\newtheorem{lemma}{Lemma}[section]

\newtheorem{proposition}{Proposition}
\renewcommand{\baselinestretch}{0.899}

\def\BibTeX{{\rm B\kern-.05em{\sc i\kern-.025em b}\kern-.08em
		T\kern-.1667em\lower.7ex\hbox{E}\kern-.125emX}}
\begin{document}
	\title{Query-Centered Temporal Community Search via Time-Constrained Personalized PageRank}

	\author{Longlong Lin,
		Pingpeng Yuan$^*$, Member, IEEE,
	Rong-Hua Li$^*$,  Member, IEEE, Chunxue Zhu, Hongchao Qin,  Hai Jin,~\IEEEmembership{Fellow, IEEE, Life Member, ACM}, and Tao Jia% <-this % stops a space
		\IEEEcompsocitemizethanks{\IEEEcompsocthanksitem Longlong Lin and Tao Jia are with the College of
			Computer and Information Science,
			Southwest University, Chongqing 400715, China. \protect\\ Email: \{longlonglin, tjia\}@swu.edu.cn 
			\IEEEcompsocthanksitem  Pingpeng Yuan, Chunxue Zhu and Hai Jin are with National Engineering Research Center for Big Data Technology and System / Service Computing Technology and System Laboratory /
			Cluster and Grid Computing Laboratory / School of Computer Science and Technology, Huazhong University of Science and Technology, Wuhan 430074, China.\protect\\
			E-mail: \{ppyuan, cxzhu, hjin\}@hust.edu.cn
			\IEEEcompsocthanksitem  Rong-Hua Li and Hongchao Qin are with the School of
			Computer Science and Technololgy,
			Beijing Institute of Technology, Beijing 100081, China. \protect\\ Email: lironghuabit@126.com, qhc.neu@gmail.com \protect\\
			Pingpeng Yuan and Rong-Hua Li are Corresponding Authors}
		% <-this % stops an unwanted space
		\thanks{Manuscript received XXX, XXX; revised XXX, XXXX.}}

\IEEEtitleabstractindextext{%

	\begin{abstract}
Existing temporal community search suffers from two defects: (i) they ignore the temporal proximity between the query vertex $q$ and other vertices but simply require the result to include $q$. Thus, they find many temporal irrelevant vertices (these vertices are called \emph{query-drifted vertices}) to $q$ for satisfying their cohesiveness, resulting in $q$ being marginalized; (ii) their methods are NP-hard, incurring high costs for exact solutions or compromised qualities for approximate/heuristic algorithms. Inspired by these, we propose a novel problem named \emph{query-centered} temporal community search to circumvent \emph{query-drifted vertices}. Specifically, we first present a novel concept of Time-Constrained Personalized PageRank to characterize the temporal proximity between  $q$ and other vertices. Then, we introduce a  model called $\beta$-temporal proximity core, which can combine temporal proximity and structural cohesiveness. Subsequently, our problem is formulated as an optimization task that finds a $\beta$-temporal proximity core with the largest $\beta$. To solve our problem, we first devise an exact and near-linear time greedy removing algorithm that iteratively removes unpromising vertices.  To improve efficiency, we then design an approximate two-stage local search algorithm with bound-based pruning techniques. Finally, extensive experiments on eight real-life datasets and nine competitors show the superiority of the proposed solutions.
\end{abstract}
}

\maketitle

\input{introduction}

%
\input{problem_formulation}
%
\input{algorithm}

%
%
\input{experiment}
\input{related_work}
\vspace{-0.15cm}
\section{CONCLUSION}\label{sec:conclude}

In this work, we are the first to introduce and address the \emph{query-centered} temporal community search problem. We first develop the Time-Constrained Personalized PageRank to capture the temporal proximity between query vertex and other vertices. Then, we introduce $\beta$-temporal proximity core to combine seamlessly structural cohesiveness and temporal proximity. Subsequently, we formulate our problem as an optimization task, which returns a $\beta$-temporal proximity core with the largest $\beta$. To query quickly, we first devise an exact and near-linear time greedy removing algorithm $\gr$.  To further boost efficiency, we then propose an approximate two-stage local search algorithm $\ls$. Finally, extensive experiments on eight real-life temporal networks and nine competitors show the superiority of the proposed solutions.

\bibliographystyle{IEEEtran}
\bibliography{sample}



%\bibliographystyle{IEEEtran}

\begin{IEEEbiography}[{\includegraphics[width=1in,height=1.25in,clip,keepaspectratio]{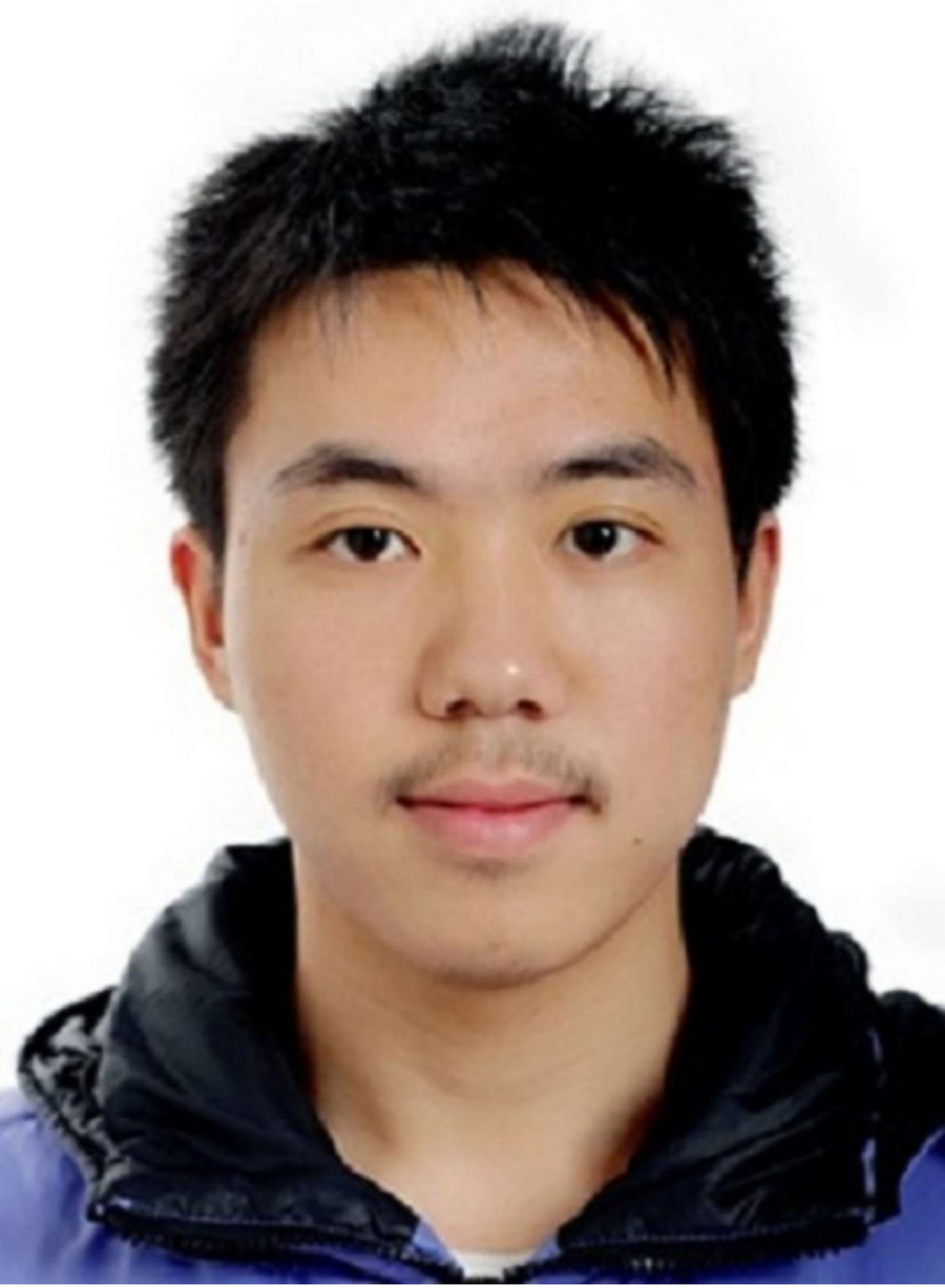}}]{Longlong Lin} received his Ph.D. degree in the theory of computer software from Huazhong University of Science and Technology (HUST), Wuhan, in 2022. He is currently an associate professor in the College of
	Computer and Information Science, Southwest University, Chongqing. His current research interests include social network analysis and temporal network mining.
\end{IEEEbiography}
\vspace{-10mm}

\begin{IEEEbiography}[{\includegraphics[width=1in,height=1.25in,clip,keepaspectratio]{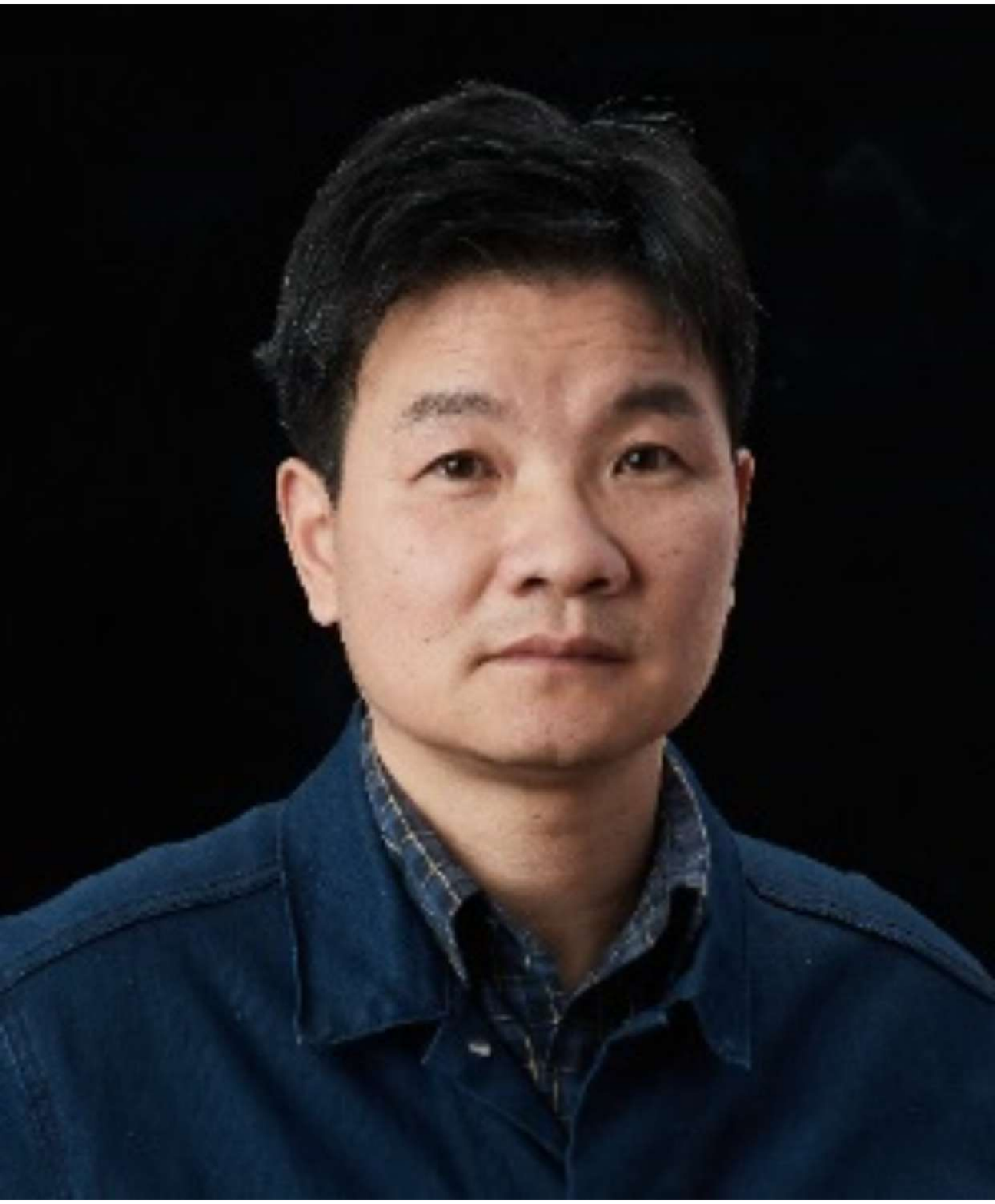}}]{Pingpeng Yuan}
	received his Ph.D. degree in computer science from Zhejiang University, Hangzhou, in 2002. He is now a professor in the School of Computer Science and Technology at Huazhong University of Science
	and Technology (HUST), Wuhan. His research interests include databases, knowledge representation and reasoning, and natural language processing, with a focus on high performance computing. He is the principle developer in multiple system prototypes, including TripleBit, PathGraph and SemreX.
\end{IEEEbiography}
\vspace{-10mm}

\begin{IEEEbiography}[{\includegraphics[width=1in,height=1.25in,clip,keepaspectratio]{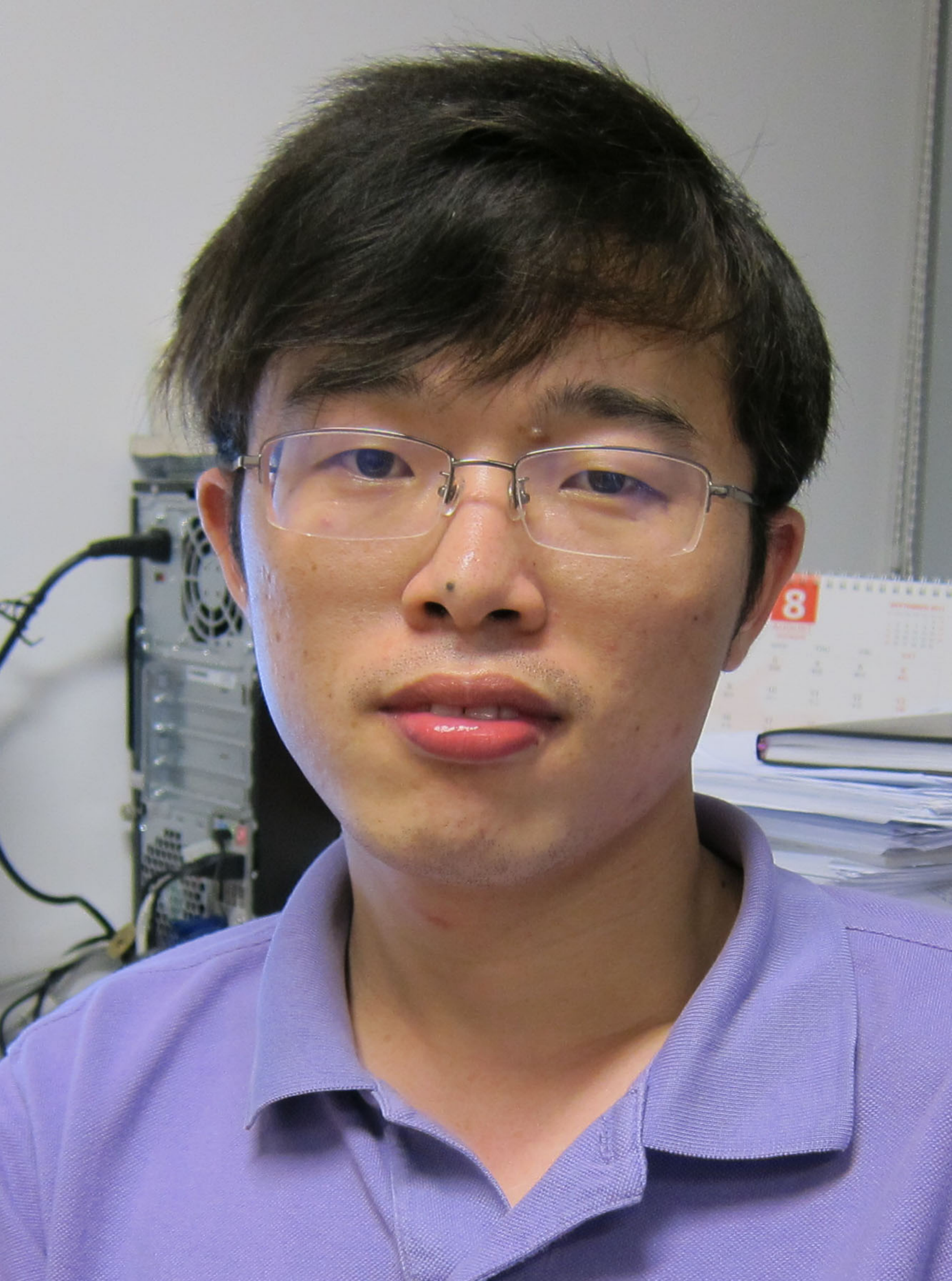}}]{Rong-Hua Li}
	received the PhD degree from the Chinese University of Hong Kong, in 2013. He is currently a professor with the Beijing Institute of Technology (BIT), Beijing, China. Before joining BIT in 2018, he was an assistant professor with Shenzhen University. His research interests include graph data management and mining, social network analysis, graph computation systems, and graph-based machine learning.
\end{IEEEbiography}
\vspace{-10mm}

\begin{IEEEbiography}[{\includegraphics[width=1in,height=1.25in,clip,keepaspectratio]{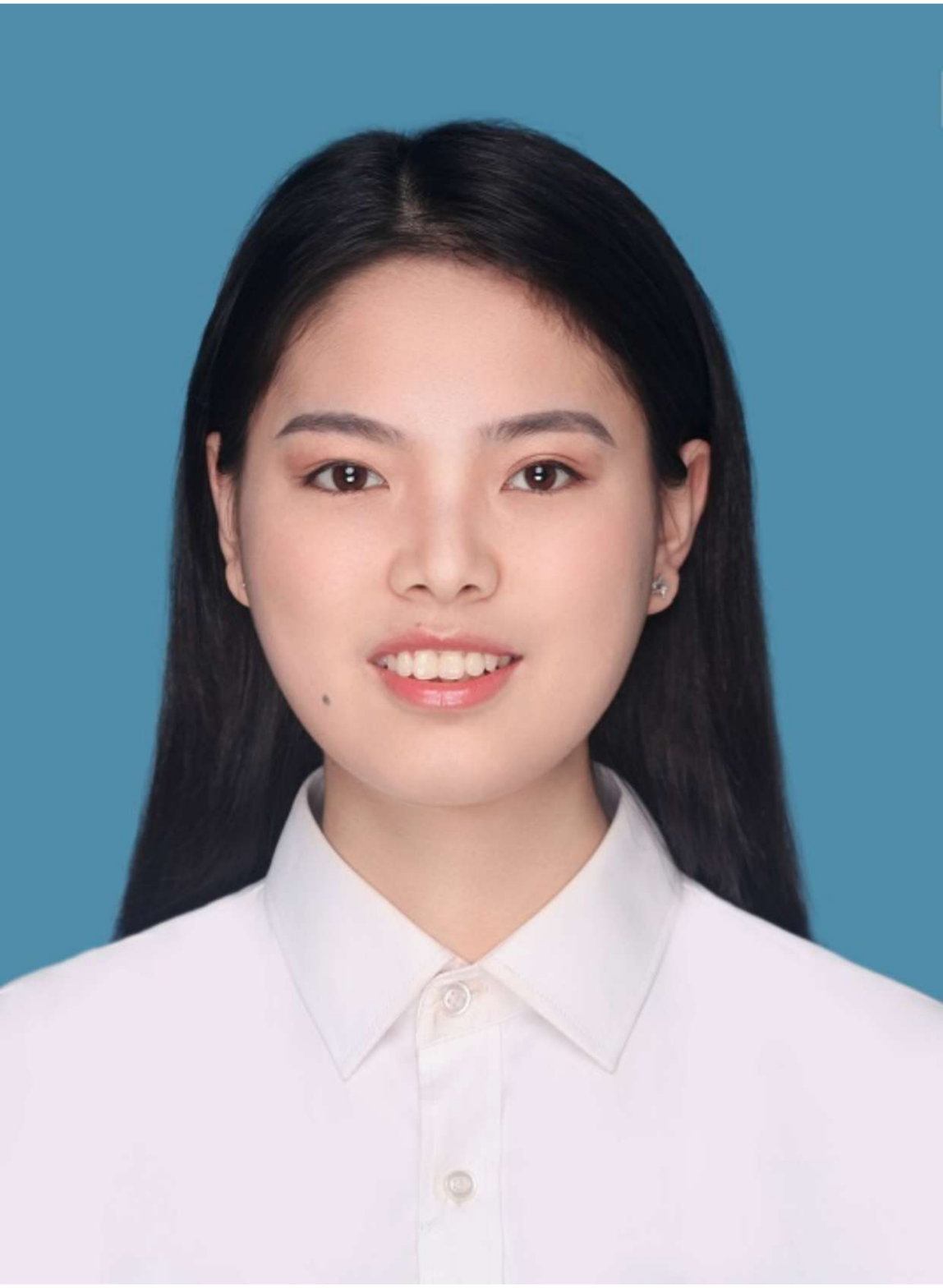}}]{Chunxue Zhu}
	received her B.E. degree in computer science and technology from Chongqing University, 	Chongqing, in 2020. She is currently pursuing her Master’s degree in School of Computer Science and Technology, Huazhong University of Science and Technology (HUST), Wuhan. Her current research interest is temporal network mining.
\end{IEEEbiography}
\vspace{-10mm}

\begin{IEEEbiography}[{\includegraphics[width=1in,height=1.25in,clip,keepaspectratio]{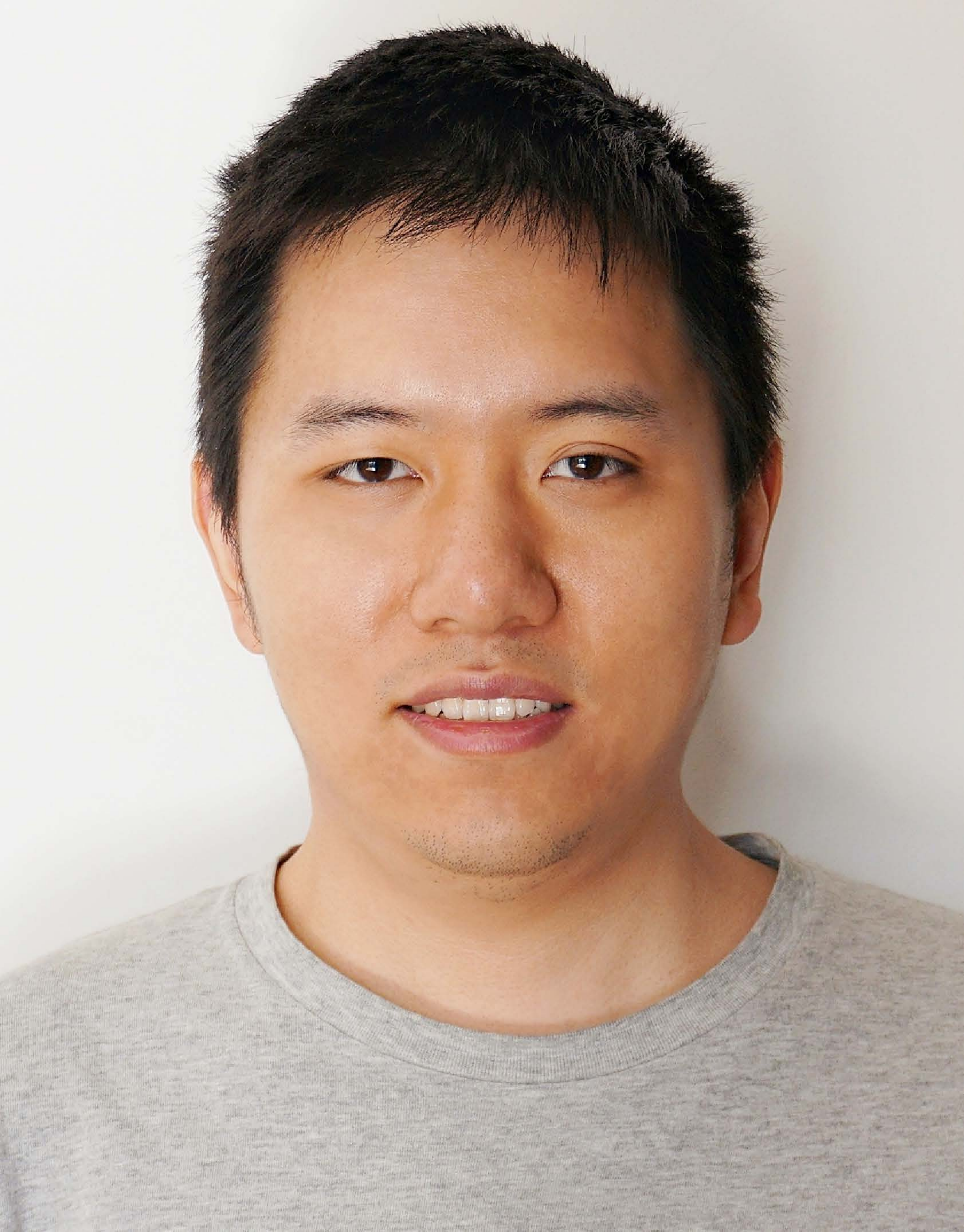}}]{Hongchao Qin}  is currently a Postdoc in Beijing Institute of Technology, China. He received the B.S. degree in mathematics, M.E. degree and Ph.D. degree in computer science from Northeastern University, China in 2013, 2015 and 2020, respectively. His current research interests include social network analysis and data-driven graph mining.
\end{IEEEbiography}
\vspace{-10mm}

\begin{IEEEbiography}[{\includegraphics[width=1in,height=1.25in,clip,keepaspectratio]{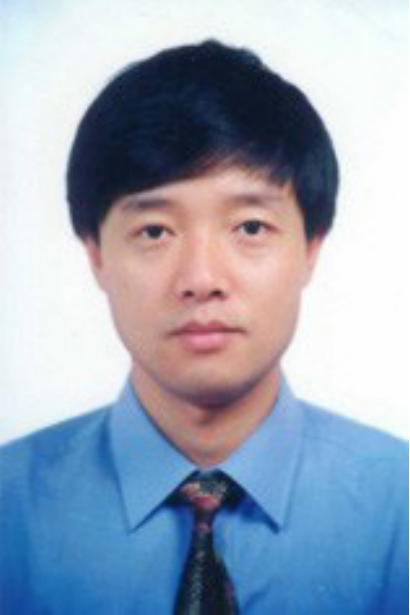}}]{Hai Jin}
	is a chair professor of computer science and engineering at Huazhong University of Science and Technology (HUST), Wuhan. Jin received his Ph.D. degree in computer engineering from HUST, Wuhan, in 1994. In 1996, he was awarded a German Academic Exchange Service fellowship to visit the Technical University of Chemnitz in
	Germany. Jin worked at The University of Hong Kong, Hong Kong, between 1998 and 2000, and as a visiting
	scholar at the University of Southern California between 1999 and 2000. He was awarded Excellent Youth Award
	from the National Science Foundation of China in 2001. Jin is a fellow of CCF and IEEE, and a life member
	of ACM. He has co-authored more than 20 books and published over 900 research papers. His research interests
	include computer architecture, parallel and distributed computing, big data processing, data storage, and system security.
\end{IEEEbiography}
\vspace{-10mm}

\begin{IEEEbiography}[{\includegraphics[width=1in,height=1.25in,clip,keepaspectratio]{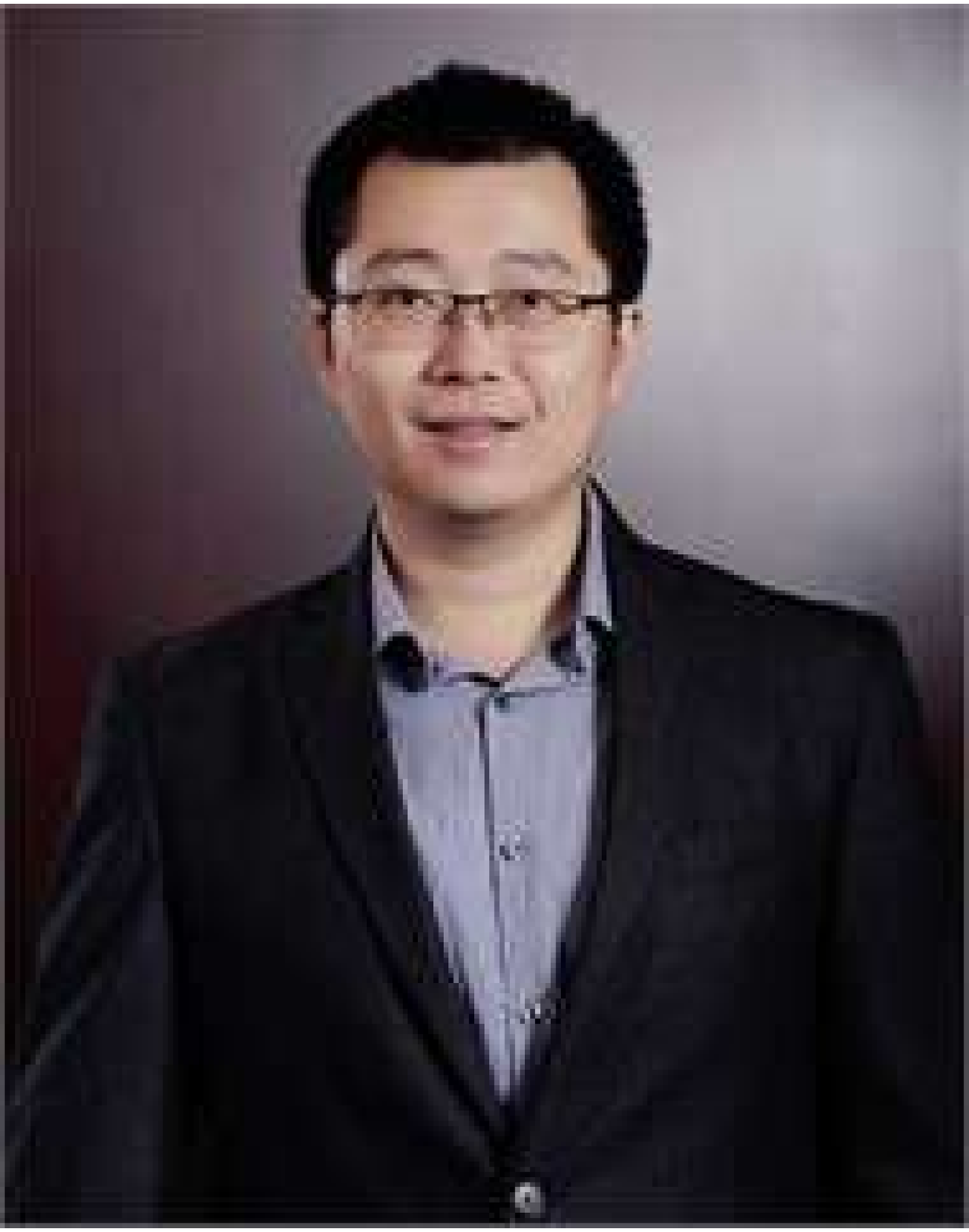}}]{Tao Jia} received the BSc degree from Nanjing University, China. He received his MSc and PhD degree from Virginia Tech, USA. He is currently a Professor at Southwest University, China. His research interest includes graph mining, brain networks, and social computing.
\end{IEEEbiography}
\vspace{-10mm}

%\newpage


%\vfill


\end{document}

%% file: introduction.tex
\section{Introduction} \label{sec:introduction}
Many real-life graphs exhibit rich community structures that are defined as densely connected subgraphs. Community mining is a significant vehicle for analyzing network organization. In general, the research on community mining can be divided into community detection \cite{newman2004fast, donetti2004detecting, rokach2005clustering, DBLP:conf/icde/ChangQ19} and community search \cite{DBLP:conf/focs/AndersenCL06,DBLP:conf/kdd/SozioG10, DBLP:conf/sigmod/CuiXWW14,  DBLP:conf/dasfaa/ZhangLYJ22,  DBLP:conf/sigmod/YangXWBZL19, DBLP:conf/kdd/ChenL0XY020}. The former aims to find all communities by some predefined criteria (e.g., modularity \cite{newman2004fast}), resulting in that it is time-consuming and not customized for user-specified query requests. To alleviate these defects, the latter identifies the specific community containing the user-specified query vertex, which is more efficient and personalized. Additionally, it also witnesses a series of applications such as social recommendation \cite{DBLP:conf/kdd/SozioG10}, protein complexes identification \cite{DBLP:conf/sigmod/CuiXWW14} and  impromptu activities organization \cite{DBLP:conf/kdd/ChenL0XY020}.

Despite the significant success of community search, most existing approaches are tailored to static networks. However, many real networks often contain complex time interaction information among vertices, which are typically named temporal networks \cite{DBLP:journals/corr/Holme15a}. For example, in e-commerce or social media, the connection between two parties was made at a specific time. Thus, conventional static community search methods may find a sub-optimal result. For example, Fig. \ref{fig:motivation} shows a sample money transfer network, in which the timestamps of each edge indicate when the two individuals make transactions. We assume \textit{Frank} is the query vertex. By using 3-core as the community model (i.e., a 3-core is a community in which each vertex has at least 3 neighbors), the vertices $S_1$ within the green circle is the answer if the time information of edges is ignored  \cite{DBLP:conf/kdd/SozioG10,  DBLP:conf/sigmod/CuiXWW14}. Although \textit{David}, \textit{Carol} and \textit{Erin} meet the structural constraints in $S_1$ (i.e., each of them has at least 3 neighbors in $S_1$), the occurrence time of the transactions among $S_1$ differs greatly. Thus, $S_1$ is an unpromising temporal community \cite{DBLP:conf/icde/LiSQYD18, DBLP:journals/pvldb/ChuZYWP19, DBLP:conf/icde/QinLWQCY19}. Recently, a few researches have been done on temporal community search \cite{DBLP:conf/bigdataconf/TsalouchidouBB20, DBLP:journals/tkdd/GalimbertiCBBCG21}. The vertices $S_2$ included in the blue circle is the answer if \cite{DBLP:conf/bigdataconf/TsalouchidouBB20} is executed. However, we can see that \textit{Frank} is on the periphery of $S_2$. This is because \textit{Alice} and \textit{Bob}  have poor temporal proximity with respect to \textit{Frank} (Section \ref{subsec:tppr}), resulting in  \textit{Frank}  being marginalized. 

\begin{figure}[t!]
	\centering
	\includegraphics [width=0.48\textwidth] {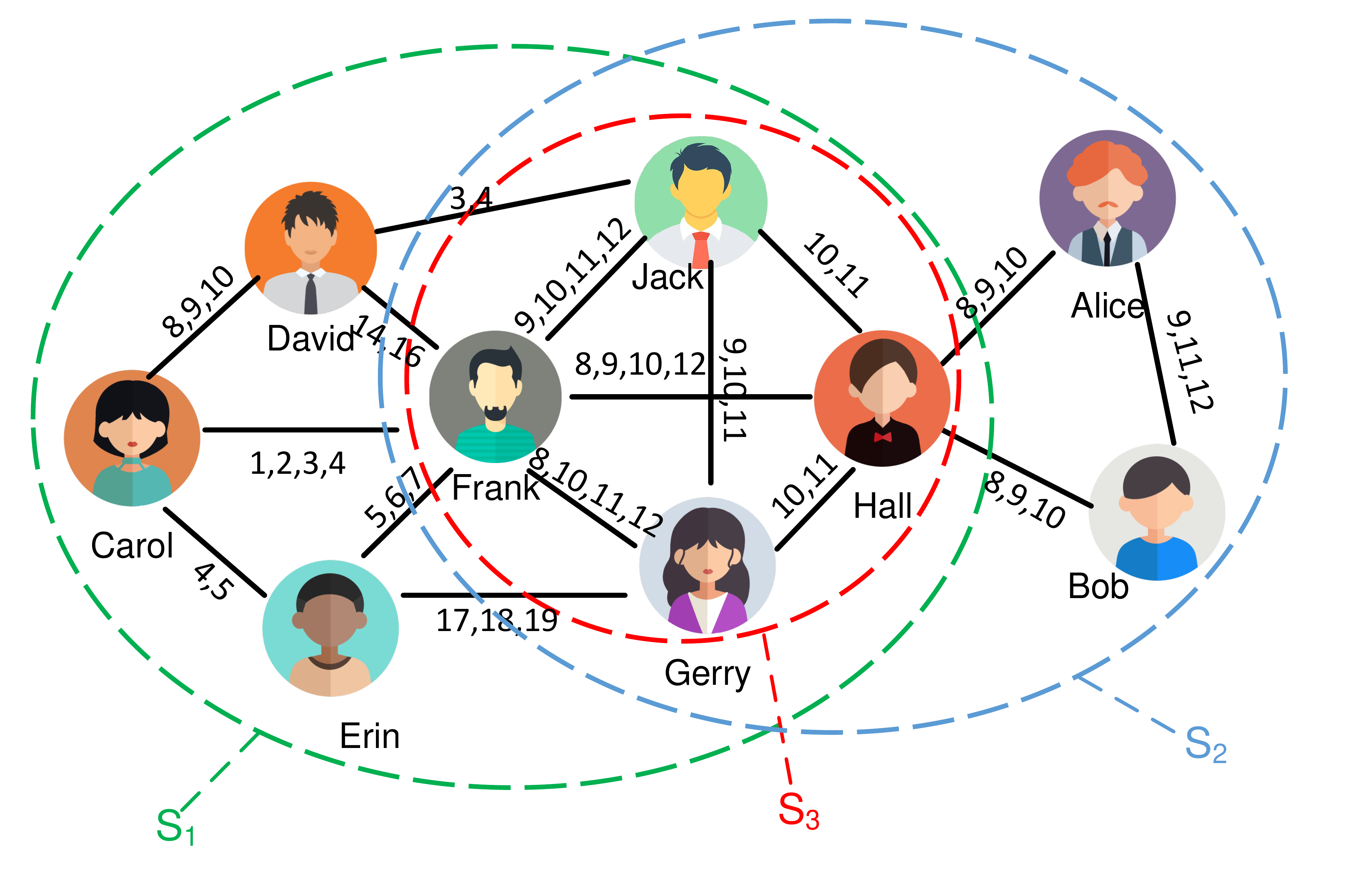}\vspace*{-0.3cm}
	\caption{Motivation Example}
	\label{fig:motivation} \vspace*{-0.7cm}
\end{figure}

In this paper, we study a novel problem named    \emph{query-centered} temporal community search (\emph{QTCS}), which aims to identify a community such that the theme of this  community revolvers around the query vertex.  Intuitively, the vertices $S_3$ included in the red circle may be the target community. This is because \textit{Hall}, \textit{Jack} and \textit{Gerry}  trade with \textit{Frank} frequently  at time 8-12. Thus,  \textit{Frank}-centered \emph{QTCS} may be a criminal gang headed by \textit{Frank} \cite{levi2006money}. Besides, on temporal collaboration networks, \emph{QTCS} may be the research group initiated by the given query vertex. Therefore, detecting \emph{QTCS} enables us to reveal some interesting applications.

There are  some studies on temporal community detection that can also solve temporal community search with simple adjustments. For instance, they first find all possible communities by the predefined criteria \cite{DBLP:conf/icde/LiSQYD18, DBLP:journals/pvldb/ChuZYWP19, DBLP:conf/icde/QinLWQCY19}, and then select the target community containing the query vertex from these communities. 
Unfortunately, existing  temporal community search methods suffer from two major defects in terms of \emph{QTCS}.   First,  the vertices in the target community should be closely related to the query vertex in community search problem \cite{ DBLP:journals/pvldb/WuJLZ15, DBLP:journals/pvldb/HuangLYC15}. However,  all existing methods do not consider the temporal proximity between the query vertex and other vertices but simply require the result to include the query vertex. Thus, they may find many temporal irrelevant vertices to the query vertex for satisfying their objective functions (e.g., structural and temporal cohesiveness), resulting in the query vertex being marginalized (Section \ref{sec:experiments}).  We refer to such temporal irrelevant vertices as \qds (Section \ref{subsection:query_drift}). Second, most existing methods are NP-hard, incurring either prohibitively high costs for exact solutions or severely compromised results for approximate/heuristic algorithms. For example, \cite{DBLP:conf/icde/LiSQYD18, DBLP:conf/bigdataconf/TsalouchidouBB20} cannot obtain the results within two days in our experiments, which is clearly impractical for online interactive graph explorations.

\stitle{\textbf{Solutions.}} For the first defect, we extend the well-known proximity metric Personalized PageRank to Time-Constrained Personalized PageRank (\textit{TPPR}) by integrating temporal constraint, which can more properly capture the temporal proximity between the query vertex and other vertices. Equipped with \textit{TPPR}, we then propose $\beta$-temporal proximity core to model the preference of user-specified query vertex by combining seamlessly the temporal proximity and structural cohesiveness.  As a result, by maximizing the value of $\beta$ of a $\beta$-temporal proximity core, we can ensure that these \qds are removed and the query vertex is \emph{centered} in the detected community (Section \ref{subsection:query_drift} and \ref{sec:experiments}). Besides, $\beta$-temporal proximity core has only one parameter (i.e., the teleportation probability $\alpha$ in Section \ref{subsec:problem}) like \cite{DBLP:conf/bigdataconf/TsalouchidouBB20, DBLP:journals/tkdd/GalimbertiCBBCG21}, which is user-friendly.  In contrast, \cite{DBLP:conf/icde/LiSQYD18, DBLP:journals/pvldb/ChuZYWP19, DBLP:conf/icde/QinLWQCY19} have many parameters which are heavily dependent on datasets and are often hard-to-tune. For the second defect, we propose two efficient algorithms. Specifically, we first develop an exact and near-linear time greedy removing algorithm called  $\gr$.  $\gr$ first computes \textit{TPPR} for every vertex and then greedily selects out the vertices with the minimum query-biased temporal degree (Definition \ref{def:degree}). To compute \textit{TPPR}, a straightforward solution is to apply the traditional power iteration method \cite{Page1999ThePC}, but it requires prohibitively high time costs. Based on in-depth observations, we propose a non-trivial dynamic programming approach to compute \textit{TPPR} for every vertex. To further boost efficiency, we then develop an approximate two-stage local search algorithm named $\ls$  with several powerful pruning techniques. The high-level idea of $\ls$ is to adopt the expanding and reducing paradigm. The expanding stage directly starts from the query vertex and progressively adds qualified vertices with proposed bound-based pruning techniques. Until it touches the termination condition with theoretical guarantees.  The reducing stage iteratively removes unqualified vertices to satisfy the approximation ratio. Our main contributions are listed as follows:

\begin{itemize} [leftmargin=8pt, topsep=0pt]
	\item  \underline{Novel Model.} We formulate the \emph{query-centered} temporal community search (\emph{QTCS}) problem in Section \ref{sec:pro}. To the best of our knowledge, the problem has never been studied in the literature.

	\item \underline{Theoretical Analysis.} We introduce the concept of \emph{query-drifted vertices} to analyze the limitations of the existing	solutions in Section \ref{sec:analysis}. We show that most existing
	methods contain many \emph{query-drifted vertices}, resulting in the query vertex being marginalized. However, our proposed model can circumvent these \emph{query-drifted vertices}, resulting in that the query vertex is \emph{centered} in the target community.
	
	\item \underline{Efficient Algorithms.} To solve our problem quickly, we propose two practical algorithms in Section \ref{sec:gr} and \ref{sec:ls}. One of them is the exact greedy removing algorithm $\gr$ with near-linear time complexity. The other is the approximate two-stage local search algorithm $\ls$.

	\item \underline{Comprehensive Experiments.} Experiments (Section \ref{sec:experiments}) on eight datasets with different domains and sizes demonstrate our proposed solutions indeed are more
	efficient, scalable, and effective than the existing nine competitors. For instance, on a million-vertex DBLP dataset, $\ls$ consumes about 13 seconds while $\gr$ takes 47 seconds. However, some competitors cannot get the results within two days on some datasets. Our model is much denser and more separable in terms of temporal feature than the competitors. Our model can find high-quality \emph{query-centered} temporal communities by eliminating \qds which the competitors cannot identify.
\end{itemize}

%% file: problem_formulation.tex
% !TeX spellcheck = <none>
\section{Problem Formulation} \label{sec:pro}
   \begin{figure*}[t!]
	\centering
	\subfigure[Temporal graph $\mathcal{G}$]
	{\includegraphics [width=0.17\textwidth,height=0.12\textwidth]
		 {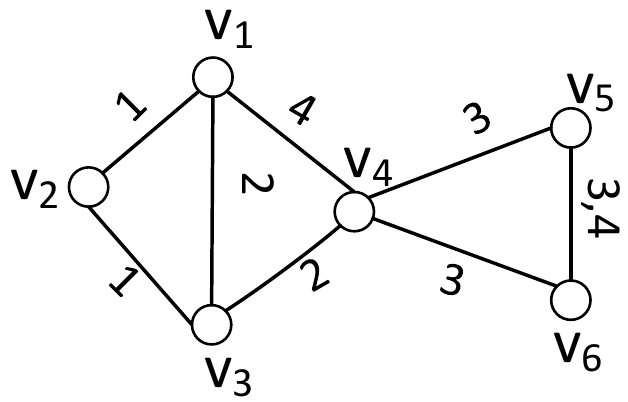}
		\label{fig:1(a)}
	}
	\subfigure[De-temporal graph $G$]
	{\includegraphics [width=0.17\textwidth,height=0.12\textwidth] {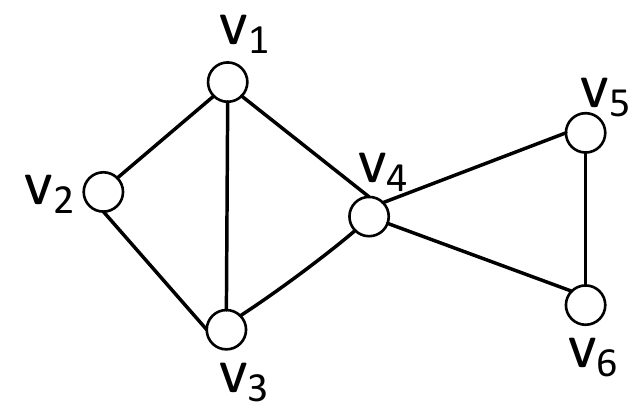}
		\label{fig:b}
	}
	\subfigure[Edge stream representation for $\mathcal{G}$]
	{\includegraphics [width=0.3\textwidth,height=0.11\textwidth]{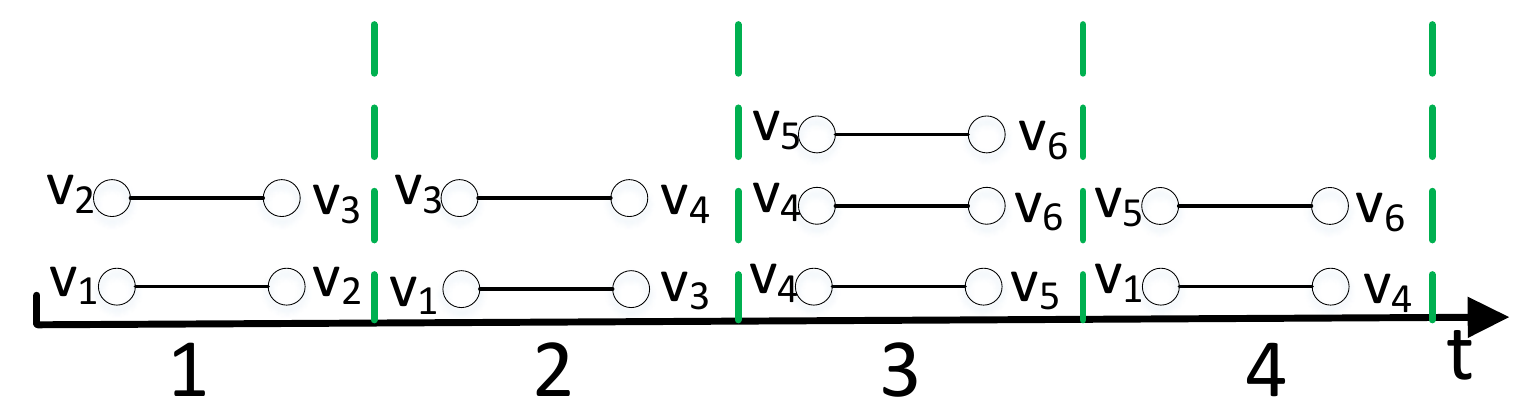}
		\label{fig:edge_stream}
	}
	\subfigure[State transition of $\mathcal{G}$]
	{\includegraphics [width=0.31\textwidth,height=0.12\textwidth]{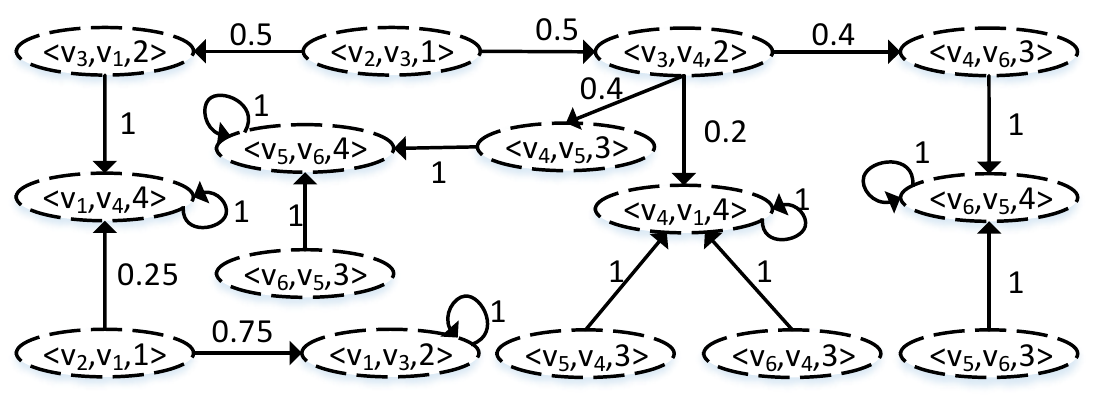}
		\label{fig:state}}
	\vspace*{-0.25cm}
	\caption{De-temporal graph, Edge stream and State transition of an example temporal graph} \vspace*{-0.3cm}
	\label{fig:example}  \vspace{-0.1cm}
\end{figure*}

%Here, we first give some important notations. Subsequently, we
%introduce a novel concept of Time-Constrained Personalized PageRank to capture the temporal proximity between the query vertex
%and other vertices. Finally, we  state our  problem.
\subsection{Notations}\label{subsec:pre}

We use $\mathcal{G}(V,\mathcal{E})$ to denote an undirected temporal graph, in which $V$ (resp. $\mathcal{E}$) indicates the vertex set (resp. the temporal edge set). Let $(u,v,t)$ $\in \mathcal{E}$ be any temporal edge which  indicates an interaction was made between $u$ and $v$ at timestamp $t$. Note that $(u,v,t_1)$ and $(u,v,t_2)$ are regarded as two different temporal edges if $t_1 \neq t_2$. That is, $u$ and $v$ may be connected at different timestamps.  Let $|V|=n$ and $|\mathcal{E}|=m$ be the number of vertices and the number of temporal edges, respectively. For example, Fig. \ref{fig:example}(a) illustrates a sample temporal graph $\mathcal{G}$ with 6 vertices and 9 temporal edges.   More generally, temporal graphs can also be modeled as edge stream \cite{DBLP:journals/corr/Holme15a}, which is a  sequence of all temporal edges ordered by timestamps. Fig. \ref{fig:example}(c) shows the edge stream representation for Fig. \ref{fig:example}(a). We use $G(V,E)$ to denote the de-temporal graph of $\mathcal{G}$, in which $E=\{(u,v)|(u,v,t)\in \mathcal{E}\}$ and $|E|=\bar{m}$. That is, $G$ is a static graph that ignores the timestamps of $\mathcal{G}$.  Fig. \ref{fig:example}(b) shows a de-temporal graph $G$. Let $G_{S}=(S,E_S)$ be the subgraph induced by  $S$ if $S \subseteq V$ and $E_S= \{(u,v)\in E| u,v \in S\}$. Let $N_{S}(v)=\{u\in S|(u,v)\in E\}$ be the neighbors of $v$  in $S$. 

\subsection{Time-Constrained Personalized PageRank}\label{subsec:tppr}
Recall that Personalized PageRank (\textit{PPR}) is a widely adopted proximity metric in  network analysis, which can measure the structural proximity between two vertices \cite{Page1999ThePC, DBLP:conf/wsdm/LofgrenBG16, DBLP:conf/sigmod/WeiHX0SW18}. Essentially, \textit{PPR} models a random walk process that has a unique stationary distribution and solves the following equation\footnote{We use lowercase letters to denote scalars (e.g., $\alpha$), bold lowercase letters to denote row vectors (e.g., \textbf{s} or \textbf{x}), bold capital letters to denote matrices (e.g., \textbf{W} or \textbf{P}).}:
\begin{equation}\label{eq:ppr}
\textbf{x}=\alpha \textbf{s} +(1-\alpha)\textbf{x}\textbf{W}
 \end{equation}
\textbf{x} is the stationary \textit{PPR} distribution, $\alpha$ is the teleportation probability, and $\textbf{s}$\footnote{$\textbf{s}$ is a distribute in the original \textit{PPR}. That is, multiple non-zero entries are allowed in $\textbf{s}$. When $\textbf{s}$ is a one-hot vector, \textit{PPR} is also called random walk with restart \cite{DBLP:conf/kdd/TongF06}.} is a start distribution named the teleportation vector. $\textbf{W}$ is the state transition matrix, where each entry $W_{vu}$ indicates the transition probability from vertex $v$ to vertex $u$.

Although \textit{PPR} has achieved significant success in static networks, the research on how to design effective temporal proximity is not sufficient (Section \ref{sec:relate}).  Thus, to preserve the rich temporal information in \textit{PPR}, we face the following two challenges. First, how to design an effective walk in temporal networks. In real-world scenarios, the information transmission follows the time-order and asynchronous interaction behaviors. For example, $(v_2,v_1,v_4,v_5)$ is a walk in Fig. \ref{fig:example} (b), but $(v_2,v_1,v_4,v_5)$ in Fig. \ref{fig:example} (a) is clearly problematic with respect to (w.r.t.) time-order. Second, how to design an effective state transition matrix in temporal networks. Intuitively, the preference of an interaction decreases as time goes by \cite{DBLP:conf/icde/XieTSBH15} (i.e., the tie between two vertices becomes stronger if the interaction between them happens in a more current time). For instance, in Fig. \ref{fig:1(a)}, when the walker walks to $v_1$ through temporal edge $(v_2,v_1,1)$, the probability that the walker chooses $(v_1,v_3,2)$  to walk is higher than $(v_1,v_4,4)$. But the traditional state transition matrix $\textbf{W}$ cannot distinguish such edge relationships. Additionally, more than an interaction may occur between two vertices in temporal networks. So, \textbf{W}  is not applicable for modeling temporal proximity.

For ease of description, we convert each temporal edge to two \emph{ordered} temporal edges of opposing directions. For example, $(u,v,t)$ converts to $<u,v,t>$ and $<v,u,t>$\footnote{To avoid confusion, we use $( )$ and $< >$ represent the temporal edge and \emph{ordered} temporal edge, respectively.}. Moreover, we use $\vec{e}$ to denote any \emph{ordered} temporal edge.  Let $head(\vec{e})$, $tail(\vec{e})$ and $time(\vec{e})$ be the head vertex, tail vertex and timestamp of $\vec{e}$,  $N^{>}(\vec{e})=\{<u,v,t>|u=tail(\vec{e}), t>time(\vec{e})\}$, $\vec{e}^{out}_u=\{\vec{e}|head(\vec{e})=u\}$, $\vec{e}^{in}_u=\{\vec{e}|tail(\vec{e})=u\}$. Based on these symbols, we present the following definition to overcome the challenges discussed above.
\begin{definition} \label{def:ttp}
	\textbf{[Temporal transition matrix]}
	Given a temporal graph $\mathcal{G}(V,\mathcal{E})$, the temporal transition matrix \textbf{P} $\in R^{m\times m}$ on two \emph{ordered} temporal edges $\vec{e}_i$ and $\vec{e}_j$ can be computed as
	\begin{equation} \label{eq:ttm}
	P(\vec{e}_i \to \vec{e}_j)=
	\begin{cases}
	\frac{g(time(\vec{e}_j)-time(\vec{e}_i))}{\sum\limits_ {\vec{e}_k\in N^{>}(\vec{e}_i)} g(time(\vec{e}_k)-time(\vec{e}_i))}, & \vec{e}_j \in N^{>}(\vec{e}_i) \\
		0, &\vec{e}_j \notin N^{>}(\vec{e}_i) \\
		\end{cases}
	\end{equation}
\end{definition}
$P(\vec{e}_i \to \vec{e}_j)$ indicates the temporal transition probability from $\vec{e}_i$ to $\vec{e}_j$ and $g(a-b)$ is a decaying function to capture the dependency between interactions. Here,  we apply a linear decaying function $g(a-b)=\frac{1}{a-b}$, which is often used in temporal settings \cite{DBLP:conf/sdm/LaiWY13, DBLP:conf/dasfaa/WuZCY17}. Our proposed solutions can trivially accommodate different functions (e.g.,exponential  or logarithmic function). In the case that $\sum_{\vec{e}_{j}}P(\vec{e}_i \to \vec{e}_j)=0$, we call $\vec{e}_i$ a dangling state as \cite{Page1999ThePC, DBLP:conf/wsdm/LofgrenBG16, DBLP:conf/sigmod/WeiHX0SW18}. For simplicity, we set $P(\vec{e}_i \to \vec{e}_i)=1$ to handle these dangling states. By doing so, we can guarantee that \textbf{P} is a stochastic matrix, that is,  $\sum_{\vec{e}_{j}}P(\vec{e}_i \to \vec{e}_j)=1$ for any $\vec{e}_i$ holds.  Note that  \textbf{P} is constructed only once for each dataset to support different queries. Fig. \ref{fig:state} shows the state transition for Fig. \ref{fig:1(a)}, in which we ignore the isolated \emph{ordered} temporal edges.

\begin{definition} \label{def:tppr}
	\textbf{[Time-Constrained Personalized PageRank]} Given a temporal graph $\mathcal{G}(V,\mathcal{E})$, a query vertex $q$ and a teleportation probability $\alpha$, the Time-Constrained Personalized PageRank of  vertex $u$ w.r.t. $q$ is denoted by $tppr_q(u)= \sum_{\vec{e}\in \vec{e}^{in}_u}\widetilde{ppr}(\alpha,\widetilde{\chi_{q}})(\vec{e})$.
	\begin{equation} \label{eq:tppr}
	\widetilde{ppr}(\alpha,\widetilde{\chi_{q}})=\alpha \widetilde{\chi_{q}} +(1-\alpha) \widetilde{ppr}(\alpha,\widetilde{\chi_{q}}) \textbf{P}
	\end{equation}
$\widetilde{\chi_{q}} \in R^{1 \times m}$ is a vector with $\widetilde{\chi_{q}}(\vec{e})=1/|\vec{e}^{out}_q|$ for $\vec{e} \in \vec{e}^{out}_q$.
\end{definition}

We explain the intuition behind the Definition \ref{def:tppr} as follows: (i) Equation \ref{eq:tppr} is also  a random walk process analogous to Equation \ref{eq:ppr}, except that each state in Equation \ref{eq:tppr} is an \emph{ordered} temporal edge instead of a vertex. Thus, $\widetilde{ppr}(\alpha,\widetilde{\chi_{q}})(\vec{e})$ reflects the temporal proximity of each \emph{ordered} temporal edge $\vec{e}$ w.r.t. $q$. (ii) Since \textbf{P} is a stochastic matrix, $\widetilde{ppr}(\alpha,\widetilde{\chi_{q}})$ is a probability distribution \cite{DBLP:conf/wsdm/LofgrenBG16, DBLP:conf/sigmod/WeiHX0SW18}. Thus, $\sum_{u} tppr_q(u)=\sum_{u}\sum_{\vec{e}\in \vec{e}^{in}_u}\widetilde{ppr}(\alpha,\widetilde{\chi_{q}})(\vec{e})=1$. That is, $tppr_q$ is also a probability distribution. So, it is reasonable to use $tppr_q(u)$ to describe temporal proximity of $u$ w.r.t. $q$. For simplicity, we use $tppr(u)$ to denote $tppr_q(u)$ if the context is clear.

\stitle{\textbf{Remark.}} Although ``new" edges (i.e., those associated with the largest timestamps) do not have a change to connect to any \emph{ordered} temporal edge by Definition \ref{def:ttp},  they have an $\alpha$ probability to jump back to $\widetilde{\chi_{q}}$  by Definition \ref{def:tppr}. Thus, ``new" edges does not get trapped in self-loops (see \ref{subsec:com_tppr} for details). 
 
\subsection{Problem Statement}\label{subsec:problem}
As mentioned above, the Time-Constrained Personalized PageRank (\textit{TPPR}) can be used to measure the temporal proximity between the query vertex and other vertices. Therefore, a naive way is to identify a connected subgraph containing the query vertex and has optimal \textit{TPPR} score. Unfortunately, it ignores the fact that a perfect temporal community should also have strong structural cohesiveness. Thus, another potential approach is to adopt the cohesive subgraph model $k$-core to model the structural cohesiveness of the community \cite{DBLP:conf/kdd/SozioG10, DBLP:conf/sigmod/CuiXWW14}. We call this model \textit{QTCS\_{Baseline}}, which serves as a baseline model for experimental comparison in Section \ref{sec:experiments}.

\begin{definition} \label{def:basemodel}
\textbf{[\textit{QTCS\_{Baseline}}]}
For a temporal graph $\mathcal{G}(V,\mathcal{E})$, a teleportation probability $\alpha$, a query vertex $q$ and a parameter $k$, \textit{QTCS\_{Baseline}} finds a vertex set $S$, satisfying (i) $q \in S$;
(ii) $G_S$ is a connected $k$-core (i.e., $|N_{S}(v)|\geq k$ for any $v \in S$); (iii) $\min\{tppr(u)|u\in S\}$ is maximum.
\end{definition}

However, \textit{QTCS\_{Baseline}} considers separately structural cohesiveness and temporal proximity, resulting in that it may identify a sub-optimal result (see Section \ref{sec:experiments} for details). For example, \textit{QTCS\_{Baseline}} may remove many vertices with good temporal proximity under the structural constraints of the $k$-core. Conversely, it may contain many vertices with poor temporal proximity to satisfy the structural cohesiveness. Thus, we propose the following novel metrics to combine seamlessly structural cohesiveness and temporal proximity.

\begin{definition}
	\textbf{[Query-biased temporal degree]} \label{def:degree}
	 Given a vertex set $C$, the query-biased temporal degree of vertex  $u$ w.r.t. $C$ is defined as	$\rho_{C}(u)=\sum_{v\in N_C(u)} tppr(v)$.
\end{definition}

By Definition \ref{def:degree}, we know that the query-biased temporal degree measures the quality of neighbors rather than quantity. For example,  $u$ has $10^5$ neighbors and each neighbor has a \textit{TPPR} value of $10^{-10}$. As a result, the query-biased degree of $u$ is $10^{-5}$. On the other hand, suppose $u$ has only 10 neighbors, but each neighbor has a \textit{TPPR} value of $10^{-2}$. In this case, the query-biased degree of $u$ is $10^{-1}$. So,  the higher the query-biased temporal degree of $u$, $u$ may have more neighbors that are closely related to the query vertex.
\begin{definition}
	\textbf{[$\beta$-temporal proximity core]} \label{def:tpcore}
	 The $\beta$-temporal proximity core is a vertex set $C$, satisfying (i) $G_C$ is connected; (ii) $\min\{\rho_{C}(u)|u \in C\}\geq \beta$.
\end{definition}

By maximizing the value of $\beta$ of a $\beta$-temporal proximity core, we can detect a community in which each vertex of the community has many neighbors that are closely related to the query vertex. As a result, it ensures that  the detected  community is  very related to the query vertex, which is easier to interpret why the community is formed (see case studies of Section \ref{sec:experiments} for details).  

\stitle{Problem 1 (\textit{QTCS}).} Given a temporal graph $\mathcal{G}(V,\mathcal{E})$, a teleportation probability $\alpha$ and a query vertex $q$,  \emph{\textbf{q}uery-centered} \textbf{t}emporal  \textbf{c}ommunity \textbf{s}earch aims to identify a vertex set $C$, satisfying (i) $q \in C$;
(ii) $C$ is a $\beta$-temporal core with the largest $\beta$; (iii) there does not exist another community $C'\supseteq C$ meets the above conditions.

\stitle{\textbf{Remark.}}  Our proposed model \textit{QTCS} is asymmetric. Namely, \textit{QTCS} depends on query nodes and different query nodes return different communities. For example, a user X is in a Trump-centered circle (i.e. the theme of this circle revolves around Trump), but Trump is not in X-centered circle. 
\section{Problem Analysis} \label{sec:analysis}
\subsection{Comparison with \textit{CSM}} \label{subsec:csm}
The community search by maximizing the minimum degree (\textit{CSM}) \cite{DBLP:conf/kdd/SozioG10, DBLP:conf/sigmod/CuiXWW14} does have many similarities with our methods, but there are also pivotal differences.  First, a key concept in \textit{CSM} is the degree of each vertex. So, we can simply adapt the \textit{CSM} model to solve the temporal community search problem by using a concept of temporal degree. Specifically, the temporal degree of a vertex $u$ is the number of temporal edges that $u$ participates in. Such a simple adaption, however, has some serious defects. For example, the temporal degree is a local metric used to measure the absolute importance of vertices in the network. However, for the community search problem, it may be more appropriate to consider the relative importance between the query vertex and other vertices \cite{ DBLP:journals/pvldb/WuJLZ15, DBLP:journals/pvldb/HuangLYC15}. Unlike \textit{CSM}, our solution is based on a new definition of query-biased temporal degree (Definition \ref{def:degree}) which can capture the relative importance for temporal community search. Second, in \textit{CSM}, the (temporal) degree of a vertex can be obtained by simply checking the number of neighbors. However, the proposed query-biased temporal degree is a global metric, needing more complicated techniques to calculate it. Finally, the technologies of \textit{CSM} are very hard to handle massive temporal networks. This is because their technologies are tailored to static networks. Even if a temporal network can be approximately transformed into a static network by existing methods, the size of the static network is often much larger than the original temporal network  (e.g., \cite{DBLP:journals/pvldb/WuCHKLX14}), resulting in prohibitively computational costs. However, our technologies are directly oriented to temporal networks which are very efficient as shown in our experiments. Besides, we have also empirically demonstrated the superiority of our approach by comparing it with \textit{CSM} in terms of community quality (Section \ref{sec:experiments}).

\subsection{Query Drift Issue} \label{subsection:query_drift}
Here, we want to prove that most existing methods may identify many temporal irrelevant vertices to the query vertex $q$ for optimizing their objective functions. For simplicity, we assume that  $f(.)$ is an objective function, and the larger the value of $f(C)$, the better the quality of the community $C$. Let $C^*(f)$  be any optimal community based on $f(.)$, and $C_q$  be any community containing  $q$.

\begin{definition} \label{def:query_drift} Given an objective function $f(.)$, we say $C^*(f)-C_q (\neq \emptyset)$ is \qds and $f(.)$ suffers from the \qd issue  iff the following two conditions hold: (i) $f(C^*(f)\cup C_q)\geq f(C_q)$; (ii) $\min\{\rho_{C^*(f)\cup C_q}(u)|u\in C^*(f)\cup C_q \} \leq \min\{\rho_{C_q}(u)|u\in C_q \}$.
\end{definition}
By Definition \ref{def:query_drift}, we know that adding \qds $C^*(f)-C_q$ to $C_q$ can improve its objective function score (i.e., condition (i)), but reduce the query-biased temporal degree (i.e., condition (ii)).  In other words, if an objective function $f(.)$  finds many temporal irrelevant vertices to the query vertex (i.e., condition (ii))  for optimizing $f(.)$ (i.e., condition (i)), then we say that $f(.)$ suffers from the \qd issue.

\stitle{\textbf{Remark.}} Surprisingly, the condition (i) of Definition \ref{def:query_drift} is also
called the free rider issue, which has been widely considered in static
community search \cite{ DBLP:journals/pvldb/WuJLZ15, DBLP:journals/pvldb/HuangLYC15}. Specifically, if an objective function $f(.)$ has the free rider issue (i.e., condition (i)), $f(.)$-based community search methods tend to include some redundant vertices (e.g., $C^*(f)-C_q$) in the detected community. However, the free rider issue cannot measure the temporal proximity between the query vertex and the redundant vertices. Thus, we introduce condition (ii) to  further measure how these redundant vertices affect the temporal proximity of the detected community. As a result,  our proposed \qd issue is more strict than the free rider issue. That is, if $f(.)$ suffers from the \qd issue, then $f(.)$ must  have the free rider issue, and vice versa is not necessarily true.

\begin{proposition}\label{prop:our_model}
Given a temporal graph $\mathcal{G}$ and a query vertex $q$, \textit{QTCS} does not suffer from the query drift issue.
	
\end{proposition}

\begin{proof}
	Let $S^{*}$ be the solution for the \textit{QTCS} problem, and thus $q\in S^{*}$. The Proposition can be proved by contradiction. Assume that  there is a vertex set $S$ such that $f(S\cup S^*)\geq f(S^*)$ and $\min\{\rho_{S\cup S^{*}}(u)|u\in S\cup S^{*} \} \leq \min\{\rho_{S^{*}}(u)|u\in S^{*}\}$. By Definition \ref{def:tpcore} and Problem 1, we have $f(C)=\min\{\rho_{C}(u)|u\in C \}$ for \textit{QTCS}. Thus, $f(S\cup S^*)\geq f(S^*)$ is equivalent to $\min\{\rho_{S\cup S^{*}}(u)|u\in S\cup S^{*} \} \geq \min\{\rho_{S^{*}}(u)|u\in S^{*}\}$. So, $\min\{\rho_{S\cup S^{*}}(u)|u\in S\cup S^{*} \} =\min\{\rho_{S^{*}}(u)|u\in S^{*}\}$. As a result, (i) $q \in S\cup S^{*}$; (ii) $S\cup S^{*}$ is a $\beta$-temporal core with the largest $\beta$. This contradicts the maximality of $S^{*}$ (i.e., condition (iii) of Problem 1). Thus, there does not exit \emph{query-drifted vertices} $S$ for \textit{QTCS}.
\end{proof}

\begin{proposition}\label{prop:other_model}
 Given a temporal graph $\mathcal{G}$ and a query vertex $q$,  \cite{DBLP:journals/pvldb/ChuZYWP19, DBLP:journals/tkdd/GalimbertiCBBCG21, DBLP:conf/icde/LiSQYD18, DBLP:conf/icde/QinLWQCY19} suffer from the query drift issue.
\end{proposition}
\begin{proof} Let $C_q$ be a vertex set that satisfies conditions (i) and (ii) of Problem 1. Thus, by Definition \ref{def:tpcore} and Problem 1, we know that condition (ii) of Definition \ref{def:query_drift} holds for $C_q$ and any $C^*(f)$. Next, we prove that \cite{DBLP:journals/pvldb/ChuZYWP19, DBLP:journals/tkdd/GalimbertiCBBCG21, DBLP:conf/icde/LiSQYD18, DBLP:conf/icde/QinLWQCY19} meet the condition (i) of Definition \ref{def:query_drift}.

\underline{For \cite{DBLP:journals/pvldb/ChuZYWP19}:} The objective function $f(C)=\frac{m(\mathcal{G}_C)}{|C| \cdot |\mathcal{T}_C|}$, in which $m(\mathcal{G}_C)$ is the sum of edge weights within the temporal subgraph $\mathcal{G}_C$ and $\mathcal{T}_C$ is the time set of $\mathcal{G}_C$. For example, in Fig. \ref{fig:1(a)}, we let $C=\{v_1,v_3,v_4,v_5, v_6\}$, thus $m(\mathcal{G}_C)=7$ and $\mathcal{T}_C=\{2,3,4\}$.  So, $C^*(f)$ is a vertex set with the largest $f$ value. Since $m(\mathcal{G}_C)/\mathcal{T}_C$ is a monotonically increasing supermodular and $|C|>0$ is a submodular,  $f(C^*(f)\cup C_q)\geq f(C_q)$ according to \cite{DBLP:journals/pvldb/WuJLZ15}. Thus, \cite{DBLP:journals/pvldb/ChuZYWP19} has the \qd issue.

\underline{For \cite{DBLP:journals/tkdd/GalimbertiCBBCG21}:} Given a fixed interval $I$ and a static "AND" graph $G_I(C)=\cap_{t\in I}\{(u,v)|u,v \in C, (u,v,t) \in \mathcal{G}\}$, the objection function $f(C)=\min_{u \in C} d_{I}(u,C)$, in which $d_{I}(u,C)$ is the degree of $u$ in $G_I(C)$. So, $C^*(f)$ is a vertex set with the largest $f$ value. Thus, $\min_{u \in C^*(f)\cup C_q} d_{I}(u,C^*(f)\cup C_q)\geq \min_{u \in C_q} d_{I}(u,C_q)$. That is, $f(C^*(f)\cup C_q)\geq f(C_q)$. Thus, \cite{DBLP:journals/tkdd/GalimbertiCBBCG21} suffer from the \qd issue.

\underline{For \cite{DBLP:conf/icde/LiSQYD18}:} If $C$ is a $(\theta,\tau)$-persistent $k$-core, then  $f(C)=|C|$, otherwise $f(C)=0$. Thus, $C^*(f)$ is a $(\theta,\tau)$-persistent $k$-core with the largest $f$ value. When $C_q$ is a $(\theta,\tau)$-persistent $k$-core, then we have $C^*(f)\cup C_q$ is also a $(\theta,\tau)$-persistent $k$-core and  $f(C^*(f)\cup C_q)=|C^*(f)\cup C_q|\geq f(C_q)=|C_q|$. When $C_q$ is not a $(\theta,\tau)$-persistent $k$-core, we have $f(C_q)=0$ and $f(C^*(f)\cup C_q)\geq f(C_q)$.  So, \cite{DBLP:conf/icde/LiSQYD18} has the \qd issue.

\underline{For \cite{DBLP:conf/icde/QinLWQCY19}:} If $C$ is a periodic clique, then  $f(C)=1$, otherwise $f(C)=0$. Thus $C^*(f)$ is any periodic clique. When $C_q$ is a periodic clique, we let $C^*(f)$ contains $C_q$. Thus we have $C^*(f)\cup C_q$ is also a periodic clique and  $f(C^*(f)\cup C_q)=1\geq f(C_q)=1$. When $C_q$ is not a periodic clique, we have $f(C_q)=0$ and $f(C^*(f)\cup C_q)\geq f(C_q)$.  As a consequence, \cite{DBLP:conf/icde/QinLWQCY19} suffer from the \qd issue.
\end{proof}
\stitle{Remark.}
The objection function $f(C)=\sum_{u,v \in C} (k(\frac{1}{dist_C(u,v)}+\frac{1}{dist_C(v,u)})-1)$ in \cite{DBLP:conf/bigdataconf/TsalouchidouBB20}, in which $k\in [0,1/2]$  and $dist_C(.)\geq 1$ is an asymmetric distance function within $\mathcal{G}_C$ that linearly integrates the temporal and spatial dimensions. When $C_q\subseteq C^*(f)$, we have $f(C^*(f))=f(C^*(f)\cup C_q)\geq f(C_q)$ because $C^*(f)$ is the vertex set with the largest $f$ value. As a result, \cite{DBLP:conf/bigdataconf/TsalouchidouBB20}  has the \qd issue when  $C_q\subseteq C^*(f)$. Unfortunately, the formal proof for $C_q\not\subseteq  C^*(f)$ is quite difficult and we leave it as an open problem. In this regard, we note as follows. First, \cite{DBLP:conf/bigdataconf/TsalouchidouBB20} involves complex distance calculations, so it has high time complexity and even it is NP-hard (more details in \cite{DBLP:conf/bigdataconf/TsalouchidouBB20}). In particular, \cite{DBLP:conf/bigdataconf/TsalouchidouBB20} cannot obtain the results within two days on some datasets (Exp-1 of Section \ref{sec:experiments}). Second, \cite{DBLP:conf/bigdataconf/TsalouchidouBB20} has poor community quality (Exp-6 of Section \ref{sec:experiments}). This is because \cite{DBLP:conf/bigdataconf/TsalouchidouBB20} only applied distance to measure the quality of the community, resulting in that it is a local measure and ignores the cohesiveness of the community.

%\iffalse

\subsection{Handle Multiple Query Vertices} \label{subsec:multi}
In many applications, multiple query vertices may be initiated by users. We show that our proposed frameworks for single query vertex can be generalized to deal with multiple query vertices. Let $S$ be the query vertex set, the \textit{TPPR} of vertex $u$ w.r.t. $S$ is denoted by $tppr_{S}(u)=\sum_{q \in S}tppr_{q}(u)/|S|$. \footnote{Other alternatives are possible for defining $tppr_{S}(u)$. For example, $tppr_{S}(u)=\min\{tppr_{q}(u)| q\in S\}$ or $tppr_{S}(u)=\prod_{q \in S}tppr_{q}(u)$.} By doing so, we propose a new definition and a new problem as follows.
\begin{definition}
	Given a vertex set $C$ and a query vertex set $S$, the query-biased temporal degree  of vertex  $u$ w.r.t. $C$ and $S$ is defined as:
	$\rho_{C}^{S}(u)=\sum_{v\in N_C(u)} tppr_S(v)$.
\end{definition}
\stitle{Problem 2 (\textit{QTCS} with multiple query vertices).} Given a temporal graph $\mathcal{G}(V,\mathcal{E})$, a teleportation probability $\alpha$ and a query vertex set $S$, the problem is to identify a vertex set $C$, satisfying (i) $S \subseteq C$ and $G_C$ is connected; (ii) $\min\{\rho_{C}^{S}(u)|u \in C\}$ is the maximum; (iii) there does not exist another community $C'\supseteq C$ meets the above conditions.
%\fi

%% file: algorithm.tex
\section{Exact Greedy Removing for \textit{QTCS}} \label{sec:gr}
 In this section, we devise an exact greedy removing algorithm  $\gr$ to address our problem \textit{QTCS}. The main idea of $\gr$ is first to calculate the \textit{TPPR} of each vertex and then greedily remove the vertices with the minimum query-biased temporal degree.
\subsection{Edge Stream For \textit{TPPR} Computation} \label{subsec:com_tppr}

Here, we focus on calculating \textit{TPPR} of every vertex. Straightforwardly, we can use the classic power iteration method \cite{Page1999ThePC} to solve the Equation \ref{eq:tppr} by utilizing the knowledge of linear algebra (i.e., matrix-vector product operations). However, such a method has a high time overhead when handling temporal networks. The reasons are as follows. The time complexity of the  power iteration method is O($MN$), in which $M$ is the number of non-zero elements in the state transition matrix and $N$ is the number of iterations. For temporal graphs, since each state in our model is an ordered temporal edge instead of a vertex,  $M=O(m^2)$ ($m$ is the number of temporal edges). Thus, the time complexity of the power iteration method is O($m^2N$). Motivated by this, we propose an efficient algorithm with near-linear time by simulating the process of the temporal walk and applying edge stream to reduce computational cost.

\begin{definition} \label{def:tem_walk}
	\textbf{[$l$-hop  temporal  walk]} A $l$-hop  temporal  walk from vertex $i$ to vertex $j$ is a sequence of ordered temporal edges $\{\vec{e}_1,\vec{e}_2,..., \vec{e}_l\}$, satisfying $head(\vec{e}_1)=i$, $tail(\vec{e}_l)=j$, $tail(\vec{e}_i) = head(\vec{e}_{i+1})$ and $time(\vec{e}_i) \leq time(\vec{e}_{i+1})$ for all $1\leq i \leq l-1$. For simplicity, we denote $tw_{l}$ and $TW_{l}^{u \leadsto v}$ as the  $l$-hop  temporal  walk and the set of $l$-hop  temporal  walk from $u$ to $v$, respectively.
\end{definition}

\begin{definition} \label{def:trp}
	\textbf{[$l$-hop temporal transition probability]} Given a $l$-hop  temporal walk $tw_l=\{\vec{e}_1,\vec{e}_2,..., \vec{e}_l\}$, the $l$-hop temporal transition probability of $tw_l$, denoted by $P(tw_l)$, is $P(tw_l)=P(\vec{e}_1 \to \vec{e}_2)*P(\vec{e}_2 \to \vec{e}_3)*...*P(\vec{e}_{l-1} \to \vec{e}_{l})$. For completeness, we set  $P(tw_0)=0$, $P(tw_1)=1/|\vec{e}^{out}_u|$ if $tw_1=\{<u,v,t>\}$.
\end{definition}

\begin{lemma} \label{lemma:tppr_0}
	Given a temporal graph $\mathcal{G}(V,\mathcal{E})$, a query vertex $q$ and a teleportation probability $\alpha$, we have
	$tppr(u)= \sum_{i=0}^{\infty}\alpha(1-\alpha)^{i} \sum_{tw_{i+1} \in TW_{i+1}^{q \leadsto u}} P(tw_{i+1})$.
\end{lemma}
%\iffalse

\begin{proof}
First, the equation $\textbf{x}=\alpha \textbf{s} +(1-\alpha)\textbf{x}\textbf{P}$ is equivalent to $\textbf{x}(\textbf{I}-(1-\alpha)\textbf{P})=\alpha \textbf{s}$. Furthermore, the matrix $(\textbf{I}-(1-\alpha)\textbf{P})$ is nonsingular because it is strictly diagonally dominant, so this equation has a unique solution $\textbf{x}$ according to Cramer's Rule.

Second, let $\textbf{y}= \alpha \sum_{i=0}^{\infty}(1-\alpha)^{i}\textbf{P}^{i}$, we have $\alpha \textbf{s} +(1-\alpha) \textbf{s} \textbf{y}\textbf{P}=\alpha \textbf{s} +(1-\alpha) \textbf{s} \alpha \sum_{i=0}^{\infty}(1-\alpha)^{i}\textbf{P}^{i} \textbf{P}=\alpha \textbf{s} +\textbf{s} \alpha \sum_{i=1}^{\infty}(1-\alpha)^{i}\textbf{P}^{i} =\textbf{s} \alpha \sum_{i=0}^{\infty}(1-\alpha)^{i}\textbf{P}^{i}=\textbf{s} \textbf{y}$. That is $\alpha \textbf{s} +(1-\alpha) \textbf{s} \textbf{y}\textbf{P}=\textbf{s} \textbf{y}$. Since $\textbf{\textbf{x}}=\alpha \textbf{s} +(1-\alpha)\textbf{x}\textbf{P}$ and $\textbf{x}$ has a unique solution,  $\textbf{x}=\textbf{s} \textbf{y}=\alpha \sum_{i=0}^{\infty}(1-\alpha)^{i} \textbf{s} \textbf{P}^{i}$.

Third, for $\widetilde{ppr}(\alpha,\widetilde{\chi_{q}})=\alpha \widetilde{\chi_{q}} +(1-\alpha) \widetilde{ppr}(\alpha,\widetilde{\chi_{q}}) \textbf{P}$, we have $\widetilde{ppr}(\alpha,\widetilde{\chi_{q}})=\alpha\sum_{i=0}^{\infty} (1-\alpha)^{i} \widetilde{\chi_{q}} \textbf{P}^{i}$ by the previous proof. Therefore, $\widetilde{ppr}(\alpha,\widetilde{\chi_{q}})(\vec{e})=\alpha\sum_{i=0}^{\infty}(1-\alpha)^{i}P(q \stackrel{(i+1)hop}{\leadsto}\vec{e})$, in which $P(q \stackrel{(i+1)hop}{\leadsto}\vec{e})$ represents the probability that first from $q$ to $head(\vec{e})$ by $i$-hop temporal walk and then walking to $tail(\vec{e})$. So, $tppr(u)= \sum_{\vec{e}\in \vec{e}^{in}_u}\widetilde{ppr}(\alpha,\widetilde{\chi_{q}})(\vec{e})=\alpha\sum_{i=0}^{\infty}(1-\alpha)^{i}\sum_{\vec{e}\in \vec{e}^{in}_u}P(q \stackrel{(i+1) hop}{\leadsto}\vec{e})=\alpha\sum_{i=0}^{\infty}(1-\alpha)^{i}\sum_{tw_{i+1} \in TW_{i+1}^{q \leadsto u}} P(tw_{i+1})$.
\end{proof}
%\fi

\stitle{A failed attempt.}  According to Lemma \ref{lemma:tppr_0}, a naive solution is first to  enumerate all  temporal walks from query vertex $q$ to any vertex $u$. Then, it computes the $l$-hop temporal transition probability from $q$ to $u$ by previous  temporal walks, and finally obtains $tppr(u)$ by Lemma \ref{lemma:tppr_0}. Unfortunately,  it is impossible to calculate exactly the $tppr(u)$  as the summation goes to infinity. So, it is very challenging to directly apply Lemma \ref{lemma:tppr_0} to  compute $tppr(u)$. To tackle this challenge, we present an important observation as follows.

\stitle{An important observation.}  According to Definition \ref{def:ttp} and \ref{def:trp}, for  $tw_\infty =\{\vec{e}_1, \vec{e}_2,...\}$, we observe that $P(tw_\infty) \neq 0$ iff there is an integer $l$ such that (1) $time(\vec{e}_i)<time(\vec{e}_{i+1})$ and $\vec{e}_i$ is not a dangling state for $1\leq i \leq l-1$; (2) $\vec{e}_l$ is a dangling state and $\vec{e}_{l}=\vec{e}_{l+k}$ for any integer $k$.

Based on this observation, we further present an important lemma (Lemma \ref{lemma:tppr}).   Before proceeding further, we denote a $\alpha$-discount temporal walk as the following random walk process: (1) it starts from $q$; (2) at each step it stops in the current state with probability $\alpha$, or it continues to walk according to Equation \ref{eq:ttm} with probability 1-$\alpha$. Furthermore, we use $u^{t}$ to denote any ordered temporal edge $\vec{e}$ with $tail(\vec{e})=u$ and $time(\vec{e})=t$. Let $D[u][t]$ be the probability that a  $\alpha$-discount temporal walk stops in  $u^t$ given the $\alpha$-discount temporal walk at most one dangling state $u^t$ if any.

\begin{lemma} \label{lemma:tppr}
	Given a temporal graph $\mathcal{G}(V,\mathcal{E})$, a query vertex $q$, and a teleportation probability $\alpha$,  we have $tppr(u)=\sum_{t\in T_1}D[u][t]+\sum_{t\in T_2}D[u][t]/\alpha$, in which $T_1=\{t|u^{t}$ is not a dangling state$\}$ and $T_2=\{t|u^{t}$ is a dangling state$\}$.
\end{lemma}
%\iffalse

\begin{proof}
 Assume that there is a temporal walk $\{\vec{e}_1, \vec{e}_2,  ...\vec{e}_l\}$ such that $head(\vec{e}_1)=q$ and $\vec{e}_l=u^t$.

	\underline{Case 1:} If $\vec{e}_{i}$ is not a dangling state for $1 \leq i \leq l$ and $P(tw_{i+1}) \neq 0$, we have $l \neq \infty$ by the previous observation. Let $l_{max}$ be the maximum $l$ that satisfies the above condition, we have $ \sum_{i=0}^{l_{max}}\alpha(1-\alpha)^{i} \sum_{tw_{i+1} \in TW_{i+1}^{q \leadsto u^t}} P(tw_{i+1})=D[u][t]$.
	
	\underline{Case 2:}   If there are some dangling states, there must exist an integer $k$ such that $\vec{e}_i$ is  not a dangling state for $i<k$ and $\vec{e}_{j}$ is a dangling state for $k \leq j \leq l$. Let $l_{max}$ be the maximum $l$ that satisfies the above condition, note that $l_{max}$ may be $\infty$.  Thus,  we have $ \sum_{i=0}^{l_{max}}\alpha(1-\alpha)^{i} \sum_{tw_{i+1} \in TW_{i+1}^{q \leadsto u^t}} P(tw_{i+1})=\sum_{i=0}^{l_{max}-k}(1-\alpha)^{i}D[u][t]$.  So, $ \sum_{i=0}^{\infty}\alpha(1-\alpha)^{i} \sum_{tw_{i+1} \in TW_{i+1}^{q \leadsto u^t}} P(tw_{i+1})=\sum_{i=0}^{\infty}(1-\alpha)^{i}D[u][t]=D[u][t]*(1+(1-\alpha)+(1-\alpha)^{2}+...(1-\alpha)^{\infty})=D[u][t]*(1/(1-(1-\alpha)))=D[u][t]/\alpha$.
	
	In short, if $u^t$ is not a dangling state, $ \sum_{i=0}^{\infty}\alpha(1-\alpha)^{i} \sum_{tw_{i+1} \in TW_{i+1}^{q \leadsto u^t}} P(tw_{i+1})= D[u][t]$. If $u^{t}$ is a dangling state, we have $ \sum_{i=0}^{\infty}\alpha(1-\alpha)^{i} \sum_{tw_{i+1} \in TW_{i+1}^{q \leadsto u^t}} P(tw_{i+1})= D[u][t]/\alpha$. Thus, we have $tppr(u)=\sum_{i=0}^{\infty}\alpha(1-\alpha)^{i} \sum_{tw_{i+1} \in TW_{i+1}^{q \leadsto u}} P(tw_{i+1})=\sum_{t\in T_1}D[u][t]+\sum_{t\in T_2}\\D[u][t]/\alpha$ according to Lemma \ref{lemma:tppr_0}.
\end{proof}
%\fi

Based on Lemma \ref{lemma:tppr}, we devise an efficient and non-trivial dynamic programming approach (Algorithm \ref{algor:tppr}) to compute \textit{TPPR} for every vertex with one pass over all temporal edges. Algorithm \ref{algor:tppr} first initializes $tppr(u)$ as 0 and $D[u]$ as a dictionary structure for every vertex $u \in V$ (Line 1).  In Line 2, we represent the temporal graph as edge stream to ensure the time of temporal edges is  non-decreasing, which can facilitate the $D[u][t]$
calculation (see the definition of $D[u][t]$ for details). Thus, for each temporal edge $(u,v,t)$, we update the dictionary structures  $D[u][t]$ and $D[v][t]$ accordingly (Lines 3-10). As a result, the \textit{TPPR} of $u$ is the sum of $D[u][t]$ for different $t$ according to Lemma \ref{lemma:tppr} (Lines 11-15).

\begin{theorem} \label{thm:alg1}
Algorithm \ref{algor:tppr} can compute \textit{TPPR} for each vertex. The  time complexity of Algorithm \ref{algor:tppr} is $O(\mathcal{T}_{max}\cdot (m+n))$, where $\mathcal{T}_{max}=\max\{\mathcal{T}_u|u \in V\}$, $\mathcal{T}_{u}=|\{t|(u,v,t) \in \mathcal{E}\}|$.
\end{theorem}
%\iffalse

\begin{proof}
For the correctness, we know that $tppr(u)$ is the probability that a  temporal walk from $q$ stops at $u$ according to Lemma \ref{lemma:tppr_0} and \ref{lemma:tppr}. $D[u][t]$ of Algorithm \ref{algor:tppr} records the probability that the  walk stops at $u$ at time $t$. So, Algorithm \ref{algor:tppr} can correctly compute the \textit{TPPR} for every vertex. The algorithm takes $m$ rounds to update the dictionaries $D[u]$ (Line 2).  In each round, it consumes $\mathcal{T}_{max}$ time to perform the update process. In Lines 11-15, it consumes  $O(\mathcal{T}_{max}\cdot n)$ time to calculate \textit{TPPR} of every vertex. Therefore, the   time complexity of Algorithm \ref{algor:tppr} is $O(\mathcal{T}_{max}\cdot (m+n))$.
\end{proof}
%\fi

\begin{algorithm}[t]
	\scriptsize
	\caption{\textit{Compute\_tppr} ($\mathcal{G}, q, \alpha$)} \label{algor:tppr}
		\begin{flushleft}
		\hspace*{0.02in} {\bf Input:}
temporal graph $\mathcal{G}$; query vertex $q$; teleportation probability $\alpha$\\
		\hspace*{0.02in} {\bf Output:}  the \textit{TPPR} for every vertex.
	\end{flushleft}
	\begin{algorithmic}[1]
		
		\State $tppr(u) \leftarrow 0$, $D[u] \leftarrow \{\}$ for any $u \in V$
		\For{ $(u,v,t)$ in the edge stream of $\mathcal{G}$}
		\For {$t_1 \in D[u]$}
		\State  $D[v][t]=D[v][t]+ (1-\alpha) D[u][t_1] P(u^{t_1} \rightarrow <u,v,t>)$
		\EndFor
		\If {$u==q$}
		\State $D[v][t]=D[v][t]+ \frac{\alpha}{|\vec{e}^{out}_q|}$
		\EndIf
		\For {$t_2 \in D[v]$}
		\State  $D[u][t]=D[u][t]+ (1-\alpha) D[v][t_2]  P(v^{t_2}  \rightarrow <v,u,t>)$
		\EndFor
		\If {$v==q$}
		\State $D[u][t]=D[u][t]+ \frac{\alpha}{|\vec{e}^{out}_q|}$
		\EndIf
		\EndFor
		
		\For {$u \in D$}
		\For {$t \in D[u]$}
		\If {$u^{t}$ is a dangling state}
		\State $D[u[t]=D[u][t]/\alpha$
		\EndIf
		\State $tppr[u]=tppr[u]+D[u][t]$
		\EndFor
		\EndFor
		\State \Return \emph{tppr}
	\end{algorithmic}
\end{algorithm}

\subsection{The $\gr$ Algorithm} \label{subsec:gr}
Below, we show that the query-biased temporal degree satisfies a monotonic property, which supports an exact greedy removing algorithm to solve our problem.
\begin{lemma} \label{lem:anti_degree}
\textbf{[Monotonic property]}
	 Given two vertex sets $S$ and $H$ and $S\subseteq H$, we have  $\rho_{S}(u) \leq \rho_{H}(u)$ for any vertex $u \in S$ holds.
	\end{lemma}
%\iffalse
\begin{proof}
	By Definition \ref{def:degree}, we have $\rho_S(u)=\sum_{v \in N_S(u)}tppr(v)$ and $\rho_H(u)=\sum_{v \in N_H(u)}tppr(v)$. Since $S \subseteq H$, $N_{S}(u) \subseteq N_{H}(u)$, we have $\rho_S(u)\leq \rho_H(u)$.
\end{proof}
%\fi

By Lemma \ref{lem:anti_degree}, we know that the larger the vertex set, the greater the query-biased temporal degree of vertex $u$. Inspired by this, we devise an exact greedy removing  algorithm called $\gr$ (Algorithm \ref{algor:gr}).  Algorithm \ref{algor:gr} first calls Algorithm \ref{algor:tppr} to calculate \textit{TPPR} of every vertex (Line 1). Then, it initializes the current search space $temp$ as $V$, candidate result $R$ as $V$, the optimal value $\beta^{*}$ of \textit{QTCS} as 0, and the query-biased temporal degree $\rho(u)$ for every vertex $u\in V$ according to Definition \ref{def:degree} (Lines 2-3). Subsequently, it executes the  greedy removing process in each round to improve the quality of the target community (Lines 4-12). Specifically, in each round, it obtains one vertex $u$ with the minimum query-biased temporal degree (Line 5). Lines 8-12 update the candidate result $R$, the optimal value $\beta^{*}$, the search space $temp$, and the query-biased temporal degree. The iteration terminates once the current search space is empty (Line 4) or the query vertex $q$ is removed (Line 6-7). Finally, it returns  $CC(R,q)$ as the exact \emph{query-centered} temporal community (Line 13).

\begin{theorem} \label{thm:alg2}
	 Algorithm \ref{algor:gr} can identify the exact \emph{query-centered} temporal community.	The  time complexity and space complexity of Algorithm \ref{algor:gr} are $(\mathcal{T}_{max}\cdot (m+n)+n\log n+\bar{m})$ and $O( \mathcal{T}_{max} \cdot n+m)$ respectively.	
\end{theorem}	

%\iffalse
\begin{proof}
Let $S$ be the exact \emph{query-centered} temporal community. In Lines 4-12, Algorithm \ref{algor:gr} executes the  greedy removing process. That is, in each round, it greedily deletes the vertex with the minimum query-biased temporal degree. Consider the round $t$ when the first vertex $u$ of $S$ is deleted. Let $V_t$ be the vertex set from the beginning of round $t$. Clearly, $S$ is the subset of $V_t$ because $u$ is the first deleted vertex of $S$. This implies that there must be a connected subgraph $G_H$ of $G_{V_t}$ such that $G_S \subseteq G_H$. Thus, $\rho_{S}(u)\leq \rho_{H}(u)$ according to Lemma \ref{lem:anti_degree}. Moreover, $\rho_{H}(w)\geq \rho_{H}(u)$ for any $w \in H$ since $u$ has the minimum query-biased temporal degree in $V_t$. Thus, $\rho_{H}(w)\geq \rho_{H}(u)\geq\rho_{S}(u)$, which implies that $H$ has optimal minimum query-biased temporal degree. Since  Algorithm \ref{algor:gr} maintains the optimal solution during greedy removing process in Lines 8-9, $H$ will be returned as the exact \emph{query-centered} temporal community in Line 13.
	
Algorithm \ref{algor:gr} first consumes $O(\mathcal{T}_{max}\cdot (m+n))$ time to calculate the \textit{TPPR} for each vertex (Line 1). Subsequently, it consumes $O(n+\bar{m})$ time to initialize the query-biased temporal degree (Line 3). Finally, it consumes $O(n\log n+\bar{m})$ time to perform the greedy removing process (Lines 4-12).  Thus,  Algorithm \ref{algor:gr} consumes a total of $O(\mathcal{T}_{max}\cdot (m+n)+n\log n+\bar{m})$.  Algorithm \ref{algor:gr} takes $O(\mathcal{T}_{max} \cdot n)$ extra space to maintain dictionaries of Algorithm \ref{algor:tppr} for computing \textit{TPPR}. Additionally, we also take $O(m+ n)$ space to maintain the entire temporal graph. Thus, the space complexity of Algorithm \ref{algor:gr} is $O( \mathcal{T}_{max} \cdot n+m)$.	
\end{proof}
%\fi

\begin{algorithm}[t]
	\scriptsize
	\caption{$\gr$ ($\mathcal{G}, q, \alpha$)} \label{algor:gr}
	\begin{flushleft}
		\hspace*{0.02in} {\bf Input:}
		temporal graph $\mathcal{G}$; query vertex $q$; teleportation probability $\alpha$\\
		\hspace*{0.02in} {\bf Output:}
		the exact \textit{QTCS}
	\end{flushleft}
	\begin{algorithmic}[1]
		\State $tppr \leftarrow$ Compute\_tppr ($\mathcal{G}, q, \alpha$)
		\State $temp \leftarrow V$; $R \leftarrow V$; $\beta^{*} \leftarrow  0$
		\State 	$\rho(u)\leftarrow \sum_{v\in N_{V}(u)} tppr(v)$ for each vertex $u \in V$.
		
		\While{$temp\neq \emptyset$}
		\State $u \leftarrow \arg \min\{\rho(u)|u\in temp\}$
		\If{$u==q$}
		\State break
		\EndIf
		\If{$\rho(u)\geq \beta^{*}$}
		\State $R \leftarrow temp$; $\beta^{*} \leftarrow \rho(u)$
		\EndIf
		\State $temp \leftarrow temp \setminus \{u\}$
		\For {$v \in N_V(u) \cap temp$}
		\State  $\rho(v)$=$\rho(v)-tppr(u)$
		\EndFor
		\EndWhile
		\State \Return $CC(R,q)$, in which $CC(R,q)$ is the vertex set from the maximal connected component of $G_{R}$ containing $q$
	\end{algorithmic}
\end{algorithm}

In most real-life temporal graphs, $n \log n \leq m$ and  $\bar{m} \leq m$  as stated in Section \ref{sec:experiments}. Thus, the  time complexity of Algorithm \ref{algor:gr} can be further reduced to $O(\mathcal{T}_{max}\cdot m)$. Moreover, Algorithm \ref{algor:gr} is even near-linear in practice because $\mathcal{T}_{max}$ is usually small (Section \ref{sec:experiments}).   Clearly, the time complexity of \textit{QTCS} is $\Omega(m)$ because it has to visit the whole graph at least once for calculating the exact \textit{TPPR} of each vertex. Therefore, Algorithm \ref{algor:gr} is nearly optimal.

%\iffalse
\stitle{Remark.} We can simply adapt Algorithm \ref{algor:gr} to solve Problem 2. Specifically, in Line 1, we can get $tppr_q$  by executing Compute\_tppr ($\mathcal{G}, q, \alpha$) for each $q\in S$, in which $S$ is the query vertex set. Then, we modify Line 3 as $\rho(u)\leftarrow \sum_{v\in N_V(u)}\sum_{q \in S}tppr_{q}(v)/|S|$ and the iteration terminates (i.e., Lines 4-12) once the current search space is empty or any query vertex $q\in S$ is removed or there is no connected component containing $S$. Finally, we return the vertex set from the maximal  connected component of $G_R$ containing $S$.
%\fi

\stitle{\underline{Discussion for $\gr$.}} Although  $\gr$ has near-linear time complexity, it is still inefficient for handling huge temporal graphs, especially for processing online real-time queries. For example, on the DBLP dataset, $\gr$ takes 47 seconds to process a query (see Section \ref{sec:experiments}), which is disruptive to the online user experience. The reasons can be explained as follows: (1) It needs to compute the \textit{TPPR} for all vertices in advance, which dominates the time of $\gr$. In particular, $\gr$ takes $99\%$ of the time to compute \textit{TPPR} on most datasets. (2) Computing \textit{TPPR} and the greedy removing process are isolated, which makes the search space of $\gr$ relatively large. Fortunately, in many real-life scenarios, users may allow some inaccuracy for better response time in large networks. Thus, it is desirable to devise approximate solutions for queries. Inspired by this, we propose an approximate local search algorithm to tackle these issues.

\section{Approximate Two-Stage Local Search for \textit{QTCS}} \label{sec:ls}
In this section, we develop an approximate two-stage local search algorithm named $\ls$ for solving our problem \textit{QTCS}.  $\ls$ adopts the expanding and reducing paradigm. The expanding stage estimates the \textit{TPPR} for some vertices, which essentially reduces unnecessary computation. Besides, it also obtains a small vertex set (say $C$) covering all target community members with theoretical guarantees.  The reducing stage identifies an approximate solution directly from $C$ instead of the original large graph, reducing the search space.

\vspace{-0.2cm}
\subsection{The Expanding Stage}  \label{subsec:expanding}

Inspired by the problem of estimating \textit{PPR} \cite{DBLP:conf/focs/AndersenCL06}, we devise a local expanding algorithm. Before proceeding further, we briefly review the simple but efficient algorithm named \textit{Forward\_Push} proposed by Andersen et.al \cite{DBLP:conf/focs/AndersenCL06}.  \textit{Forward\_Push} starts from the source state $s$ and propagates information. The procedure iteratively updates two variables for each state $v$: its  reserve  $\pi(s,v)$ and residue $r(s,v)$. $\pi(s,v)$ indicates the approximate \textit{PPR} value of $v$ w.r.t. $s$ and $r(s,v)$ indicates the information that will be propagated to other states from state $v$. In each iteration, for each state $v$ that needs to propagate information,  \textit{Forward\_Push} propagates $\alpha$$r(s,v)$ to $\pi(s,v)$ and the remaining $(1-\alpha)r(s,v)$ is propagated along its neighbors. After finishing the propagation, \textit{Forward\_Push} sets $r(s,v)$ to zero. \textit{Forward\_Push} has the following equation \cite{DBLP:conf/focs/AndersenCL06}.
\begin{equation} \label{eq:forward}
\textit{PPR}(s,v)=\pi(s,v)+\sum_{w}r(s,w)\textit{PPR}(w,v)
\end{equation}
Where \textit{PPR}$(s,v)$ (resp. \textit{PPR}$(w,v)$) is the \textit{PPR} value of $v$ w.r.t. $s$ (resp. $w$). Our proposed expanding stage is built upon \textit{Forward\_Push}, but incorporates more novel strategies to adapt to ordered temporal edges (because each state in \textit{TPPR} is an ordered temporal edge instead of a vertex).  We first propose one key sub-algorithm in Algorithm \ref{algor:push}, which will be  invoked later to estimate the \textit{TPPR} for some vertices.
The process is similar to \textit{Forward\_Push}, except that the propagation is executed on ordered temporal edges instead of vertices. Note that we set $r(\vec{e})\geq 1/m$ in Algorithm \ref{algor:push} to speed up the propagation and enhance the subsequent pruning technologies.

\begin{lemma} \label{lem:bound}
For any vertex set  $H$ and any vertex $u \in H$, we have $\sum_{v \in N_H(u)}\sum_{\vec{e}_i\in \vec{e}^{in}_{v}}\pi(\vec{e}_i) \leq \rho_{H}(u) \leq \sum_{v \in N_H(u)}\sum_{\vec{e}_i\in \vec{e}^{in}_{v}}\pi(\vec{e}_i)+\sum_{\vec{e}}r(\vec{e})$. 
\end{lemma}
%\iffalse

\begin{proof}
	Let $nnz(\textbf{s})$ and $\textbf{e}_i(\textbf{s})$ be the number of non-zero elements in \textbf{s} and the one-hot vector with only value-1 entry corresponding to the $i$-th non-zero element in \textbf{s}, respectively. Thus, we can write $\textbf{s}=\sum_{i=1}^{nnz(\textbf{s})}s_{i}\textbf{e}_i(\textbf{s})$, where $s_i$ is the i-th non-zeros element in \textbf{s}. According to the linearity \cite{DBLP:conf/focs/AndersenCL06} and Equation \ref{eq:tppr}, we have $\widetilde{ppr}(\alpha,\widetilde{\chi_{q}})=\sum_{i=1}^{|\vec{e}^{out}_q|}(1/|\vec{e}^{out}_q|)\widetilde{ppr}(\alpha,\textbf{e}_i(\widetilde{\chi_{q}}))$. Furthermore, according to Equation \ref{eq:tppr} and \ref{eq:forward}, we have 	$\widetilde{ppr}(\alpha,\textbf{e}_i(\widetilde{\chi_{q}}))(\vec{e})=\pi(\widetilde{\chi_{q}}^{i},\vec{e})+\sum_{\vec{e}_{j}}r(\widetilde{\chi_{q}}^{i}, \vec{e}_j)\\\textit{PPR}(\vec{e}_j, \vec{e})$, where $\widetilde{\chi_{q}}^{i}$ is the ordered temporal edge corresponding to the i-th non-zero element of $\widetilde{\chi_{q}}$. Thus,
	$\rho_H(u)=\sum_{v \in N_H(u)}\sum_{\vec{e}\in \vec{e}^{in}_{v}}\widetilde{ppr}(\alpha,\widetilde{\chi_{q}})(\vec{e})
	=\sum_{v \in N_H(u)}\sum_{\vec{e}\in \vec{e}^{in}_{v}}\sum_{i=1}^{|\vec{e}^{out}_q|}(1/|\vec{e}^{out}_q|)\sum_{\vec{e_{j}}}r(\widetilde{\chi_{q}}^{i}, \vec{e}_j)\\\textit{PPR}(\vec{e}_j, \vec{e})+\sum_{v \in N_H(u)}\sum_{\vec{e}\in \vec{e}^{in}_{v}} \sum_{i=1}^{|\vec{e}^{out}_q|}(1/|\vec{e}^{out}_q|)\pi(\widetilde{\chi_{q}}^{i},\vec{e})$. So, $\rho_H(u) \geq	\sum_{v \in N_H(u)}\sum_{\vec{e}\in \vec{e}^{in}_{v}} \sum_{i=1}^{|\vec{e}^{out}_q|}(1/|\vec{e}^{out}_q|)\pi(\widetilde{\chi_{q}}^{i},\vec{e}) 
 = \sum_{v \in N_H(u)}\sum_{\vec{e}\in \vec{e}^{in}_{v}} \pi(\vec{e})$. $\sum_{v \in N_H(u)}\sum_{\vec{e}\in \vec{e}^{in}_{v}}\sum_{i=1}^{|\vec{e}^{out}_q|}(1/|\vec{e}^{out}_q|)\\\sum_{\vec{e_{j}}}\textit{PPR}(\vec{e}_j, \vec{e})r(\widetilde{\chi_{q}}^{i}, \vec{e}_j)
=\sum_{\vec{e_{j}}}\sum_{v \in N_H(u)}\sum_{\vec{e}\in \vec{e}^{in}_{v}}\textit{PPR}(\vec{e}_j, \vec{e})\\\sum_{i=1}^{|\vec{e}^{out}_q|}(1/|\vec{e}^{out}_q|)r(\widetilde{\chi_{q}}^{i}, \vec{e}_j)
	=\sum_{\vec{e_{j}}}\sum_{v \in N_H(u)}\sum_{\vec{e}\in \vec{e}^{in}_{v}}\\\textit{PPR}(\vec{e}_j, \vec{e})r(\vec{e}_j)
	=\sum_{\vec{e_{j}}}r(\vec{e}_j)\sum_{v \in N_H(u)}\sum_{\vec{e}\in \vec{e}^{in}_{v}}\textit{PPR}(\vec{e}_j, \vec{e})
	\leq \sum_{\vec{e_{j}}}r(\vec{e}_j)$. So, $ \rho_H(u)\leq 	\sum_{v \in N_H(u)}\sum_{\vec{e}\in \vec{e}^{in}_{v}} \pi(\vec{e})+\sum_{\vec{e_{j}}}r(\vec{e}_j)$.
\end{proof}
%\fi
Based on Lemma \ref{lem:bound}, we present two powerful pruning techniques used in the expanding stage. These techniques can delete some unqualified vertices or early terminate the expanding stage with theoretical guarantees.  For simplicity, we denote $C$ as the expanded vertex set for the following reducing stage, $Q$ as the candidate vertices which are neighbors of $C$ and not in $C$, $\widehat{\beta}$ as the best estimate of minimum query-biased temporal degree so far,  $D$ as the visited vertices to avoid repeated visits. Let $\widehat{tppr}(v)=\sum_{\vec{e}_i\in \vec{e}^{in}_{v}}\pi(\vec{e}_i)$ be the lower bound of \textit{TPPR} for vertex $v$ by Lemma \ref{lem:bound}.

\begin{lemma}\label{lem:upper_pruning} \textbf{[bound-based pruning]}
For a vertex $v$,  we can safely prune the vertex $v$ if $ \sum_{\vec{e}}r(\vec{e})+\sum_{w\in N_V(v)}\widehat{tppr}(w) < \widehat{\beta}$.
\end{lemma}
%\iffalse

\begin{proof}
 Assume that there is a \emph{query-centered} temporal community $S$ such that $v\in S$. Since the query-biased temporal degree is  monotonically increasing by Lemma \ref{lem:anti_degree}, $\rho_S(v) \leq \rho_V(v)$ for $v$ holds due to $S \subseteq V$.	According to Lemma \ref{lem:bound}, we have $\rho_S(v)  \leq \rho_V(v) \leq \sum_{\vec{e}}r(\vec{e})+\sum_{w\in N_V(v)}\widehat{tppr}(w)$.  If $\sum_{\vec{e}}r(\vec{e})+\sum_{w\in N_V(v)}\widehat{tppr}(w)<\widehat{\beta}$, we have that $\rho_S(v)< \widehat{\beta}$. Clearly, $\min\{\rho_
	S(u)|u\in S\} \leq \rho_S(v) <\widehat{\beta}$, which
	contradicts with $S$ being a \emph{query-centered} temporal community. So, we can safely remove $v$ without loss of accuracy.
\end{proof}
%\fi

\begin{lemma}\label{lem:stop1} \textbf{[stop expanding-I]}
	Given the current expanded vertices $C$ and candidate vertices $Q$, we can safely  terminate the expanding stage if $Q=\emptyset$.
\end{lemma}
%\iffalse

\begin{proof}
	Let $N_V(C)=\{u|N_V(u) \cap C \neq \emptyset\}$, we can clearly prune every vertex $u \in N_V(C)$ if $Q=\emptyset$. Assume that there is a query-centered temporal community $S$ containing $C$,  we have $N_V(v) \cap C = \emptyset$ for any $v \in S \setminus C$. Namely, $G_S$ is a disconnected subgraph, which contradicts with $G_S$ is connected by (i) of Definition \ref{def:tpcore}. So, we can safely stop the expanding stage when $Q=\emptyset$.
\end{proof}
%\fi
\begin{lemma}\label{lem:stop2} \textbf{[stop expanding-II]}
Given the current expanded vertices $C$ and candidate vertices $Q$, we can set $C= C \cup Q$ and safely terminate the expanding stage	if $ \sum_{\vec{e}}r(\vec{e})+\sum_{w \in Q} \widehat{tppr}(w) < \widehat{\beta}$.
\end{lemma}
%\iffalse

\begin{proof}
By Algorithm \ref{algor:push} and \ref{algor:expanding}, we have $\widehat{tppr}(v) \neq 0$ for  vertex $v \in D$. For any unvisited vertex $u \in V \setminus D$, we assume that there is a \emph{query-centered} temporal community $S$ such that $u \in S$. Thus, we have $\sum_{w\in N_S(u)}\widehat{tppr}(w) = \sum_{w\in N_S(u)\cap D}\widehat{tppr}(w) \leq \sum_{w\in Q}\widehat{tppr}(w)+\sum_{w\in N_S(u)\cap (D\setminus (C\cup Q))}\widehat{tppr}(w)=\sum_{w\in Q}\widehat{tppr}(w)$, because $D\setminus (C\cup Q)$ is the unqualified vertex set during the expanding stage and $S \cap (D\setminus (C\cup Q)) = \emptyset$.  If $\sum_{\vec{e}}r(\vec{e})+ \sum_{w\in Q}\widehat{tppr}(w) <\widehat{\beta}$, we have  $\sum_{\vec{e}}r(\vec{e})+\sum_{w\in N_S(u)}\widehat{tppr}(w)<\widehat{\beta}$. Moreover, according to Lemma \ref{lem:bound}, we have that $\rho_S(u)< \widehat{\beta}$. Clearly, $\min\{\rho_S(v)|v\in S\} \leq \rho_S(u) <\widehat{\beta}$, which
contradicts with $S$ is a \emph{query-centered} temporal community. So, we can safely remove $u$. That is, we can prune any vertex $u \in V \setminus D$ if $\sum_{\vec{e_{j}}}r(\vec{e}_j)+ \sum_{w\in Q}\widehat{tppr}(w) <\widehat{\beta}$. There is no evidence to remove any vertex $u\in Q$, thus we directly set $C=C\cap Q$ for simplicity.
\end{proof}
%\fi

\begin{algorithm}[t]
	\scriptsize
	\caption{\textit{Propagation}($\vec{e}$)} \label{algor:push}
	\begin{algorithmic}[1]
		\If{$r(\vec{e})\geq 1/m$}
		
		\For{each $\vec{e}_1 \in N^{>}(\vec{e})$}
		\State $r(\vec{e}_1) \leftarrow r(\vec{e}_1)+(1-\alpha)r(\vec{e})P(\vec{e} \to \vec{e}_1)$
		\EndFor
		\State $\pi(\vec{e}) \leftarrow \pi(\vec{e})+\alpha r(\vec{e})$, $\widehat{tppr}(tail(\vec{e})) \leftarrow \widehat{tppr}(tail(\vec{e}))+ \alpha r(\vec{e})$
		\State $r(\vec{e}) \leftarrow 0$
		\EndIf
	\end{algorithmic}
\end{algorithm}

\begin{algorithm}[t]
	\scriptsize
	\caption{\textit{Expanding} $(\mathcal{G}, q, \alpha)$} \label{algor:expanding}
			\begin{flushleft}
		\hspace*{0.02in} {\bf Input:}
		temporal graph $\mathcal{G}$; query vertex $q$; teleportation probability $\alpha$\\
		\hspace*{0.02in} {\bf Output:}
	expanded vertex set $C$, $r$  and  $\widehat{tppr}$
	\end{flushleft}
	\begin{algorithmic}[1]
				\State $r \leftarrow \{\}$; $\pi \leftarrow \{\}$; $\widehat{tppr} \leftarrow \{\}$
		\State  $r(\vec{e}) \leftarrow 1/|\vec{e}^{out}_q|$ for all $\vec{e} \in \vec{e}^{out}_q$
		\State  $C \leftarrow \emptyset$; $\widehat{\beta}\leftarrow 0$; $Q \leftarrow \{q\}$; $D \leftarrow \{q\}$
		\While{$Q \neq \emptyset$}
		\State $u\leftarrow Q.pop()$;   $C\leftarrow C \cup \{u\}$
		\For {$\vec{e}\in \vec{e}^{out}_u$}
\State $Propagation(\vec{e})$
\EndFor

	\If{$\min\{\sum_{v\in N_C(w)}\widehat{tppr}(v)|w \in C\}>\widehat{\beta}$}
		\State $\widehat{\beta}\leftarrow \min\{\sum_{v\in N_C(w)}\widehat{tppr}(v)|w \in C\}$
		\EndIf
		\For {$v \in N_V(u)$ and $v \notin D$}
		\State $D \leftarrow D \cup \{v\}$
		\If{$\sum_{\vec{e}}r(\vec{e})+\sum_{w\in N_V(v)}\widehat{tppr}(w) \geq \widehat{\beta}$}
		\State $Q.push(v)$
		\EndIf
		\EndFor	
		\If {$\sum_{\vec{e}}r(\vec{e})+\sum_{w \in Q} \widehat{tppr}(w) < \widehat{\beta}$}
		\State $C \leftarrow C \cup Q$
		\State break
		\EndIf

				\EndWhile
		\State \Return $C$, $r$ and $\widehat{tppr}$
	\end{algorithmic}
\end{algorithm}

With these powerful pruning techniques, we introduce  Algorithm
\ref{algor:expanding} to implement the expanding stage. Specifically, in Lines 1-2, the algorithm first initializes $r$ and $\pi$ for ordered temporal edges, which are used to estimate the query-biased temporal degree (Lemma \ref{lem:bound}).  In Lines 4-16, it executes the expanding process. In particular, it pops a vertex $u$ from queue $Q$ to execute the propagation process and adds $u$ into the expanded vertex set $C$ (Lines 5-7). After the propagation, it updates the estimate of minimum query-biased temporal degree (Lines 8-9). In Lines 10-13,  for each neighbor vertex $v$ of $u$, it uses the bound-based pruning technique (Lemma \ref{lem:upper_pruning}) to remove unqualified vertices. Once the queue $Q$ becomes the empty set or $\sum_{\vec{e}}r(\vec{e})+\sum_{w \in Q} \widehat{tppr}(w) < \widehat{\beta}$, the algorithm  stops expanding according to stop expanding pruning techniques in Lemma \ref{lem:stop1} and Lemma \ref{lem:stop2}. Clearly, the  vertex set $C$ returned by Algorithm \ref{algor:expanding} covers all target community members.

\begin{theorem} \label{thm:alg5}
The  time complexity and space complexity of Algorithm \ref{algor:expanding} are $O(\sum_{u \in C} \sum_{\vec{e}\in \vec{e}^{out}_{u}} |N^{>}(\vec{e})|)$ and $O(n+m)$ respectively.
\end{theorem}
%\iffalse
\begin{proof}	
Algorithm \ref{algor:push} consumes $O(|N^{>}(\vec{e})|)$ time to execute the  propagation process for each ordered temporal edge $\vec{e}$. Thus, in Lines 6-7 of Algorithm \ref{algor:expanding}, it takes $O(\sum_{\vec{e}\in \vec{e}^{out}_{u}} |N^{>}(\vec{e})|)$ time for every vertex $u \in C$. So, Algorithm \ref{algor:expanding} consumes $O(\sum_{u \in C} \sum_{\vec{e}\in \vec{e}^{out}_{u}} |N^{>}(\vec{e})|)$ in total. 
Algorithm \ref{algor:expanding} uses $O(m)$ extra space to maintain the reserve $r$ and residue $\pi$  for estimating the query-biased temporal degree. Besides, we also need $O(m+ n)$ space to maintain the whole temporal graph. So, the space complexity of Algorithm \ref{algor:expanding} is $O(n+m)$.
\end{proof}
%	\fi
\stitle{Remark.} By Theorem \ref{thm:alg5}, the time complexity of Algorithm \ref{algor:expanding} depends on the vertex set $C$, while our experiments (Section \ref{sec:experiments}) show $C$  is typically very small due to the proposed powerful pruning techniques in Lemma \ref{lem:upper_pruning}, \ref{lem:stop1} and \ref{lem:stop2}. Thus, the expanding stage can drastically delete many unqualified vertices, saving the time of the following reducing stage.

\subsection{The Reducing Stage} \label{subsec:reducing}

In the reducing stage, we identify an approximate \emph{query-centered} temporal community directly from the subset $C$ found by the previous expanding stage. At a high level, this stage progressively removes the vertices in $C$ that are not contained in the approximate solution. Until the remaining vertices meet the given approximation ratio. Choosing which vertices to remove is a significant challenge. Thus, we devise the following definition and lemma to guarantee the quality of the search.

\begin{definition}
\label{def:approximation}
For a vertex set $H$ and  $\epsilon \geq 1$, if $\min\{\rho_{H}(u)|u \in H\} \leq \beta^{*}\leq \epsilon \cdot \min\{\rho_{H}(u)|u \in H\}$, we say $H$ is an $\epsilon$-approximate \textit{QTCS}, where $\beta^{*}$ is the optimal value for \textit{QTCS}.
\end{definition}

\begin{lemma} \label{lem:approximation}
	For the current search space $R$ and  $\epsilon \geq 1$, we can safely prune $u\in R$ without losing any $\epsilon$-approximate \textit{QTCS} if $\epsilon \cdot \sum_{v \in N_R(u)}\widehat{tppr}(v)<\max \{\sum_{w \in N_C(v)}\widehat{tppr}(w)| v \in C\}+\sum_{\vec{e}}r(\vec{e})$.
\end{lemma}

%\iffalse

\begin{proof}
Assume that there is an $\epsilon$-approximate \textit{QTCS} $H\subseteq R$ such that $u \in H$,  we have $\epsilon \cdot \rho_{H}(u)\geq \beta^{*}$ due to Definition \ref{def:approximation}. Thus, if $\epsilon \cdot \rho_{R}(u)< \beta^{*}$, we can derive that there does not exist an $\epsilon$-approximate \textit{QTCS} $H\subseteq R$ such that $u \in H$. 
Moreover, $\rho_{R}(u)\geq \sum_{v \in N_R(u)} \widehat{tppr}(v)$ by Lemma \ref{lem:bound}. So, $\epsilon \cdot \sum_{v \in N_R(u)}\widehat{tppr}(v)< \beta^{*}$. On the one hand, since $C$ covers all target community members (Algorithm \ref{algor:expanding}), $\beta^{*} \leq \max\{\rho_C(v)|v \in C\}$ due to Definition \ref{def:tpcore} and Lemma \ref{lem:anti_degree}. On the other hand, we have $\max\{\rho_C(v)|v \in C\} \leq  \max \{\sum_{w \in N_C(v)}\widehat{tppr}(w)| v \in C\}+\sum_{\vec{e}}r(\vec{e})$ by  Lemma \ref{lem:bound}. Therefore, $\epsilon \cdot \sum_{v \in N_R(u)}\widehat{tppr}(v)<  \max \{\sum_{w \in N_C(v)}\widehat{tppr}(w)| v \in C\}+\sum_{\vec{e}}r(\vec{e})$. So,  vertex $u$ can be removed from $R$ if $\epsilon \cdot \sum_{v \in N_R(u)}\widehat{tppr}(v)<\max \{\sum_{w \in N_C(v)}\widehat{tppr}(w)| v \in C\}+\sum_{\vec{e}}r(\vec{e})$.
\end{proof}
%\fi

Unfortunately, $\epsilon$ does not know in advance. Thus, to obtain a high-quality estimation error $\epsilon$, we use a binary search to continuously refine $\epsilon$. The idea of the reducing stage is outlined in Algorithm \ref{algor:reducing}. Specifically, it first initializes the current search space $R$ as vertex set $C$ found by the previous expanding stage and the estimated query-biased temporal degree $\widehat{\rho}(u)$ by the lower bound of \textit{TPPR}  (Lines 1-3). Subsequently, in Line 4, it computes $\overline{\epsilon}$ as the upper bound of the approximation ratio. In Lines 5-21, it proceeds by continuously refining $\overline{\epsilon}$ and iteratively removing the unpromising vertices in each round to meet the current approximation ratio $\overline{\epsilon}$ by  Lemma \ref{lem:approximation}. In particular, in each round, it first initializes a queue
$Q$ to collect vertices to be deleted
and a set $D$ to maintain all deleted vertices (Line 6). Then it applies Lemma \ref{lem:approximation} to push those unpromising vertices into $Q$ in Lines 7-9 and processes iteratively the vertices in $Q$ to remove more unpromising vertices in Lines 12-17. The algorithm uses $flag$ to indicate whether query vertex $q$ is removed or not. If $flag$ is $True$, it updates the target approximation ratio $\epsilon$, search space $R$ and $\overline{\epsilon}$ (in Lines 20-21). The iteration terminates once query vertex $q$ is removed. Finally, the algorithm returns $CC(R,q)$ as the $\epsilon$-approximate \emph{query-centered} temporal community (Line 22).  Clearly, Algorithm \ref{algor:reducing} can correctly find an $\epsilon$-approximate \emph{query-centered} temporal community based on  Lemma \ref{lem:approximation}.

\begin{algorithm}[t] 	
	\scriptsize
	\caption{\textit{Reducing} $(C, r, \widehat{tppr}, q, \alpha)$} \label{algor:reducing}
	\begin{flushleft}
		\hspace*{0.02in} {\bf Input:}
		expanded vertex set $C$, $r$ and $\widehat{tppr}$ from Algorithm \ref{algor:expanding}; query vertex $q$; teleportation probability $\alpha$\\
		\hspace*{0.02in} {\bf Output:}
		the $\epsilon$-approximate \textit{QTCS}
	\end{flushleft}
	\begin{algorithmic}[1]
		
		\State  $R \leftarrow C$; $\widehat{\rho} \leftarrow \{\}$; $flag \leftarrow True$
		\For{$u \in C$}
		\State $\widehat{\rho}(u) \leftarrow \sum_{v \in N_C(u)}\widehat{tppr}(v)$
		\EndFor
		\State $temp \leftarrow \max \{\widehat{\rho}(u)| u \in C\}+ \sum_{\vec{e}}r(\vec{e})$; $\overline{\epsilon} \leftarrow \frac{temp}{\min \{\widehat{\rho}(u)| u \in C\}}$
		\While {$flag$}
		\State  $Q \leftarrow \emptyset$; $D \leftarrow \emptyset$
		\For{$u \in R$}
		\If{$\overline{\epsilon} \widehat{\rho}(u)\leq temp$}
		\State $Q.push(u)$
		\If{$u==q$}
		\State $flag\leftarrow False$; $Q \leftarrow \emptyset$
		\EndIf
		\EndIf
		\EndFor
		
		\While{$Q \neq \emptyset$}
		\State $u\leftarrow Q.pop$ and $D \leftarrow D \cup\{u\}$
		
		\For {$v \in N_R(u)$ and $v \notin D$}
		\State  $\widehat{\rho}(v)= \widehat{\rho}(v)-\widehat{tppr}(u)$
		\If{$\overline{\epsilon} \widehat{\rho}(v) \leq temp$}
		\State $Q.push(v)$
		\If{$v==q$}
		\State $flag\leftarrow False$; $Q \leftarrow \emptyset$
		\EndIf
		\EndIf
		\EndFor
		\EndWhile
		\If {$flag$}
		\State $\epsilon \leftarrow \overline{\epsilon}$; $R \leftarrow R \setminus D$; $\overline{\epsilon}\leftarrow \overline{\epsilon}/2$
		\EndIf
		
		\EndWhile
		
		\State \Return  ($\epsilon$, $CC(R,q)$), in which $CC(R,q)$ is the vertex set from the maximal connected component of $G_{R}$ containing $q$ and $\epsilon$ is the corresponding approximation ratio
	\end{algorithmic}
\end{algorithm}

\begin{theorem} \label{thm:alg6}
The  time complexity and space complexity of Algorithm \ref{algor:reducing} are $O(|G_C| \log m)$ and $O(|G_C|)$ respectively, where $G_C=\{(u,v)\in E|u,v \in C\}$.	
\end{theorem}	
%\iffalse

\begin{proof}
Algorithm \ref{algor:reducing} first takes $O(|G_C|)$ time to compute  the estimated query-biased temporal degree (Lines 2-3). Then, in Lines 5-21, it executes  the iterative update process. In each round, it takes $O(|G_C|)$ time to remove unpromising vertices and update the search space. Moreover, there are at most $\log_{2}(\frac{temp}{\min \{\widehat{\rho}(u)| u \in C\}})$ rounds  due to the binary search. Since $temp \leq 1$ and $\min \{\widehat{\rho}(u)| u \in C\}\geq 1/m$ (by the previous expanding stage),  $\log_{2}(\frac{temp}{\min \{\widehat{\rho}(u)| u \in C\}}) \leq \log m$.  Putting these together,  Algorithm \ref{algor:reducing} takes $O(\log m \cdot |G_C|)$ time in total. Algorithm \ref{algor:reducing} needs $O(|C|)$ space to store $\widehat{\rho}$ for the vertex set $C$. And we also require $O(|G_C|)$ space to store the subgraph graph $G_C$. So, the space complexity of Algorithm \ref{algor:reducing} is $O(|G_C|)$.	
\end{proof}

\stitle{Remark.} We can simply adapt Algorithm \ref{algor:expanding} and \ref{algor:reducing} to solve Problem 2. Let $S$ is the query vertex set. For Algorithm \ref{algor:expanding},  we set $r(\vec{e})\leftarrow 1/|\vec{e}_S^{out}|$ for all $\vec{e} \in \vec{e}_S^{out}$ where $\vec{e}_S^{out}=\cup_{q\in S} \vec{e}_q^{out}$ (Line 2), $Q \leftarrow S$, $D \leftarrow S$ (Line 3). For Algorithm \ref{algor:reducing}, the iteration terminates (i.e., Lines 5-21) once any query vertex $q\in S$ is removed or there is no connected component containing $S$. Finally, we return the vertex set from the maximal  connected component of $G_R$ containing $S$.
	
%\fi

%% file: experiment.tex
\section{Experimental Evaluation}\label{sec:experiments}
In this section, we conduct comprehensive experiments to test the efficiency, effectiveness, and scalability of the proposed solutions. These experiments are executed on a server with an Intel Xeon 2.50GHZ CPU and 32GB memory running Ubuntu 18.04.

\subsection{Experimental setup}
\underline{\textit{Datasets.}} We evaluate our solutions on eight graphs\footnote{\scriptsize http://snap.stanford.edu/,  http://konect.cc/, http://www.sociopatterns.org/} which are used in recent work \cite{DBLP:conf/icde/LiSQYD18, DBLP:conf/icde/QinLWQCY19, 9351686, DBLP:journals/tsmc/LinYLWLJ22, Chun} as benchmark datasets (Table \ref{tab:data}).  Reality Mining (Rmin for short), Lyonschool (Lyon), and Thiers13 (Thiers) are temporal face-to-face networks,  in which a vertex represents a person, and a temporal edge indicates when the corresponding persons had physical contact. Facebook and Twitter are temporal social networks, in which vertices represent users
and temporal edges indicate when they had online interactions.  Lkml and Enron are temporal communication networks in which a vertex indicates an ID
and a temporal edge signifies when the corresponding IDs had a message.
DBLP is a temporal collaboration network, in which each temporal edge denotes when the authors coauthored a paper.

\underline{\textit{Algorithms.}} We implement several state-of-the-art methods for comparison.  Specifically, \textit{CSM} \cite{DBLP:conf/sigmod/CuiXWW14} identifies the maximal $k$-core containing the query vertex with largest $k$.  \textit{TCP} \cite{DBLP:conf/sigmod/HuangCQTY14} applies the triangle connectivity and $k$-truss   to model the higher-order truss community. \textit{PPR\_NIBBLE} \cite{DBLP:conf/focs/AndersenCL06} is a local clustering method, which adopts the conductance as the criterion of a community. Note that \textit{CSM}, \textit{TCP}, and  \textit{PPR\_NIBBLE} are static community search methods. \textit{MPC} \cite{DBLP:conf/icde/QinLWQCY19} extends the concept of clique to adapt the temporal setting. \textit{PCore}\cite{DBLP:conf/icde/LiSQYD18} maintains persistently a $k$-core structure.  \textit{DBS} \cite{DBLP:journals/pvldb/ChuZYWP19}  uses the density and duration to model bursting communities.   But \textit{MPC}, \textit{PCore}, \textit{DBS} address the problem of temporal community detection. Thus, to fit our problem,  we first find all possible communities by the predefined criteria\cite{DBLP:conf/icde/QinLWQCY19, DBLP:journals/pvldb/ChuZYWP19, DBLP:conf/icde/LiSQYD18}, and then select the target community containing the query vertex from these communities. \textit{MTIS} \cite{DBLP:conf/bigdataconf/TsalouchidouBB20} and \textit{MSCS} \cite{ DBLP:journals/tkdd/GalimbertiCBBCG21} are temporal community search methods. \textit{MTIS} and  \textit{MSCS} model the temporal cohesiveness of the community by extending the network-inefficiency and $k$-core to temporal setting, respectively. \textit{QTCS\_{Baseline}} is an intuitive variant model (Definition \ref{def:basemodel}). \textit{EGR} and \textit{LAS} are our proposed methods.

\underline{\textit{Effectiveness metrics.}} Evaluating the utility of temporal community  is more difficult than static community since there are no ground-truth communities for temporal networks yet. Thus, we adopt the following two widely used effectiveness metrics  \cite{DBLP:journals/pvldb/ChuZYWP19, DBLP:conf/www/SilvaSS18, 9351686, DBLP:journals/tsmc/LinYLWLJ22, Chun}: temporal density (\textit{TD}) and temporal conductance (\textit{TC}). Specifically, let $S$ be the target community, the two metrics are defined as follows.    $TD(S)= 2*|\{(u,v,t)\in \mathcal{E}| u, v \in S\}|/|S|(|S|-1)|T_S|$, in which $T_S=\{t|(u,v,t)\in \mathcal{E}, u, v \in S\}$. Clearly, \textit{TD} computes the average density of the internal structure of the temporal community. $TC(S)= |Tcut(S,V\setminus S)|/\min\{|Tvol(S)|,|Tvol(V\setminus S)|\}$, where $Tcut(S,V\setminus S)=\{(u,v,t)\in \mathcal{E}|u \in S, v\in V \setminus S\}$, $Tvol(S)=\sum_{u \in S} \{(u,v,t) \in \mathcal{E}\}$. Clearly, \textit{TC} measures the separability of the temporal  community. Thus, the  larger the value of \textit{TD}($ S $), the denser $S$ is in the temporal network. The smaller the value of \textit{TC}($ S $), the farther $S$  is away from the rest of the  temporal network. In addition, we also report the value of our proposed objective function. Let \textit{MD}($S$)=$\min\{\rho_C(u)|u\in S\}$ be the minimum query-biased temporal degree within $S$. So, the  larger the value of \textit{MD}($S$), the better the quality of $S$ in terms of \emph{query-centered} temporal community search.

\begin{table}[t!]
	\centering
	\caption{\small Dataset statistics.  $TS$ is the time scale of the timestamp} \vspace{-0.3cm}
	\footnotesize
	\scalebox{1}{
		\begin{tabular}{c|cccccc}
			\toprule
			Dataset & $|V|$ & $|\mathcal{E}|$ & $|E|$  & \scriptsize $\mathcal{T}_{max}$ &  TS\\
			\midrule
			Rmin & 96  &76,551  & 2,539 &2,478 &Hour \\
			Lyon & 242  &218,503  & 26,594 &20 &Hour  \\
			Thiers & 328  &352,374  & 43,496 &49 &Hour  \\ 	
			\midrule
			Facebook & 45,813  &585,743  & 183,412 & 552  &Day  \\
			Twitter & 304,198 &464,653 & 452,202 & 7  &Day  \\
			\midrule
			Lkml & 26,885 & 547,660  &159,996 &2,663 &Day \\
			Enron & 86,978 & 912,763   &297,456 &765 &Day  \\
			\midrule
			DBLP & 1,729,816  &12,007,380  & 8,546,306 &49&Year\\
			\bottomrule		
	\end{tabular}}\vspace{-0.3cm}
	\label{tab:data}
\end{table}

\iffalse
\begin{table}[t!]
	\centering
	\caption{\small State-of-the-art methods}\vspace{-0.3cm}
	\scriptsize
	\scalebox{1}{
		\begin{tabular}{c|c|c|c}
			\toprule
			\multicolumn{2}{c|}{Methods} & \multicolumn{1}{c|}{Temporal}& \multicolumn{1}{c}{Remark}\\
			\midrule
			\multirow{2}{*}{Community Detection} 		&  \textit{MPC} \cite{DBLP:conf/icde/QinLWQCY19} & $\checkmark$ & Clique-based \\
			&\textit{PCore}\cite{DBLP:conf/icde/LiSQYD18} & $\checkmark$ & $k$-Core-based\\
			&  \textit{DBS} \cite{DBLP:journals/pvldb/ChuZYWP19} & $\checkmark$ & Density-based \\
			\midrule
			\multirow{7}{*}{Community Search} & \textit{CSM} \cite{DBLP:conf/sigmod/CuiXWW14} & $\times$ &$k$-Core-based   \\
			& \textit{TCP} \cite{DBLP:conf/sigmod/HuangCQTY14} & $\times$ &$k$-Truss-based\\
			& \textit{PPR\_NIBBLE} \cite{DBLP:conf/focs/AndersenCL06} &$\times$ & Conductance-based \\
			&  \textit{MTIS} \cite{DBLP:conf/bigdataconf/TsalouchidouBB20} &$\checkmark$ & 	Inefficiency-based\\
			&\textit{MSCS} \cite{ DBLP:journals/tkdd/GalimbertiCBBCG21} &$\checkmark$ & 	$k$-Core-based \\
			&\textit{QTCS\_{Baseline}} & $\checkmark$ & \textit{TPPR}-based  \\
			& \textit{EGR} &$\checkmark$ & \textit{TPPR}-based  \\
			& \textit{ALS} &$\checkmark$ & \textit{TPPR}-based \\
			\bottomrule	
	\end{tabular}} 
	\label{tab:alg}
\end{table}
\fi

\underline{\textit{Parameters.}} Unless otherwise stated, the teleportation probability $\alpha$ is set to 0.2 in all experiments as \cite{DBLP:conf/wsdm/LofgrenBG16, DBLP:conf/sigmod/WeiHX0SW18}. For other methods, we take their corresponding default parameters. To be more reliable,  we randomly select 50 vertices as query vertices and report the average running time and quality.

\subsection{Efficiency testing}
\vspace{-0.2cm}
\stitle{Exp-1: Running time of various temporal methods.} From Table \ref{tab:running}, we can see that $\ls$ is consistently faster than other methods on most datasets. For example, $\ls$ takes 3.038 seconds and 191.889 seconds to obtain the result from  Facebook and Lkml, respectively, while \textit{PCore} and \textit{MTIS} cannot get the result within two days. Moreover, our methods (i.e., \textit{QTCS\_Baseline}, \textit{EGR}, and \textit{ALS}) are more efficient than the existing methods.  The reasons can be explained as follows. (1) \textit{MPC}, \textit{PCore} and \textit{DBS}  need to enumerate all possible temporal communities in advance and then select the target community containing the query vertex from these communities, resulting in very high time overheads. (2) \textit{MTIS} and \textit{MSCS} first perform the very time-consuming Steiner tree procedure to identify a tree $T$ containing all query vertices, and then greedily add some desirable vertices to $T$ to derive the final result. (3) they are NP-hard in theory, thus they cannot be solved in polynomial time unless P=NP.  Furthermore, $\ls$ is faster than $\gr$ on all datasets. For example, $\ls$ only consumes about 13 seconds to identify the result from DBLP,  while $\gr$ consumes over 47 seconds. These results give some preliminary evidence that the proposed pruning strategies (Section \ref{sec:ls}) are efficient in practice.

\begin{table*}[t!]
	\centering
	\caption{\footnotesize Running time of various temporal methods (second).  AVG.RANK is the average rank of each method across testing datasets.} \vspace{-0.3cm}
	\footnotesize
	\scalebox{1}{
		\begin{tabular}{c|ccccccccc}
			\toprule
			\multicolumn{1}{c|}	{Temporal methods} & Rmin& Lyon&Thiers& Facebook&Twitter& Lkml& Enron& DBLP &AVG.RANK\\
			\midrule
			
			\textit{MPC}& 2133.440 & 6.153& 59.746 & 3.987 & 1.318& 47563.571& 729.380 &2605.572& 4\\
			\textit{PCore}& 35913.248 & 28561.989& $>$48h & $>$48h& 148.447& $>$48h& 21221.338 &24.493 & 7\\	
			\textit{DBS}& 47.363 & 1722.200& 2150.320  & 48.792 & 33179.300&\textbf{91.411}& 614.998 & 2462.040& 5\\
			\textit{MTIS}& $>$48h & 42.339& 154.161  & $>$48h & 152.064&$>$48h& $>$48h & 78252.764& 8\\
			\textit{MSCS}& 241.613 & 25.204& 28.786  & 753.186 & 42.699&859.255& 1290.521 & 3083.327& 6\\
			\midrule
			\textit{QTCS\_Baseline}  & 47.283 & 1.879& 6.703& 16.107 & 1.800& 226.457& 82.66 & 45.391& \underline{2}\\
			
			\textit{EGR} \ & 47.293 & 1.881& 6.711 &16.067 & 2.604& 224.592& 83.168 & 47.259& \underline{3}\\
			
			\textit{ALS}& \textbf{28.326} &\textbf{1.030}& \textbf{3.049} & \textbf{3.038} & \textbf{1.257}& 191.889&\textbf{30.557} & \textbf{13.707} & \underline{\textbf{1}}\\

			\bottomrule	
	\end{tabular}} \vspace{-0.3cm}
	\label{tab:running}
\end{table*}

\begin{figure*}[t!]
	\centering
	\subfigure[Rmin (vary rank)]{
		\includegraphics[width=3.5cm]{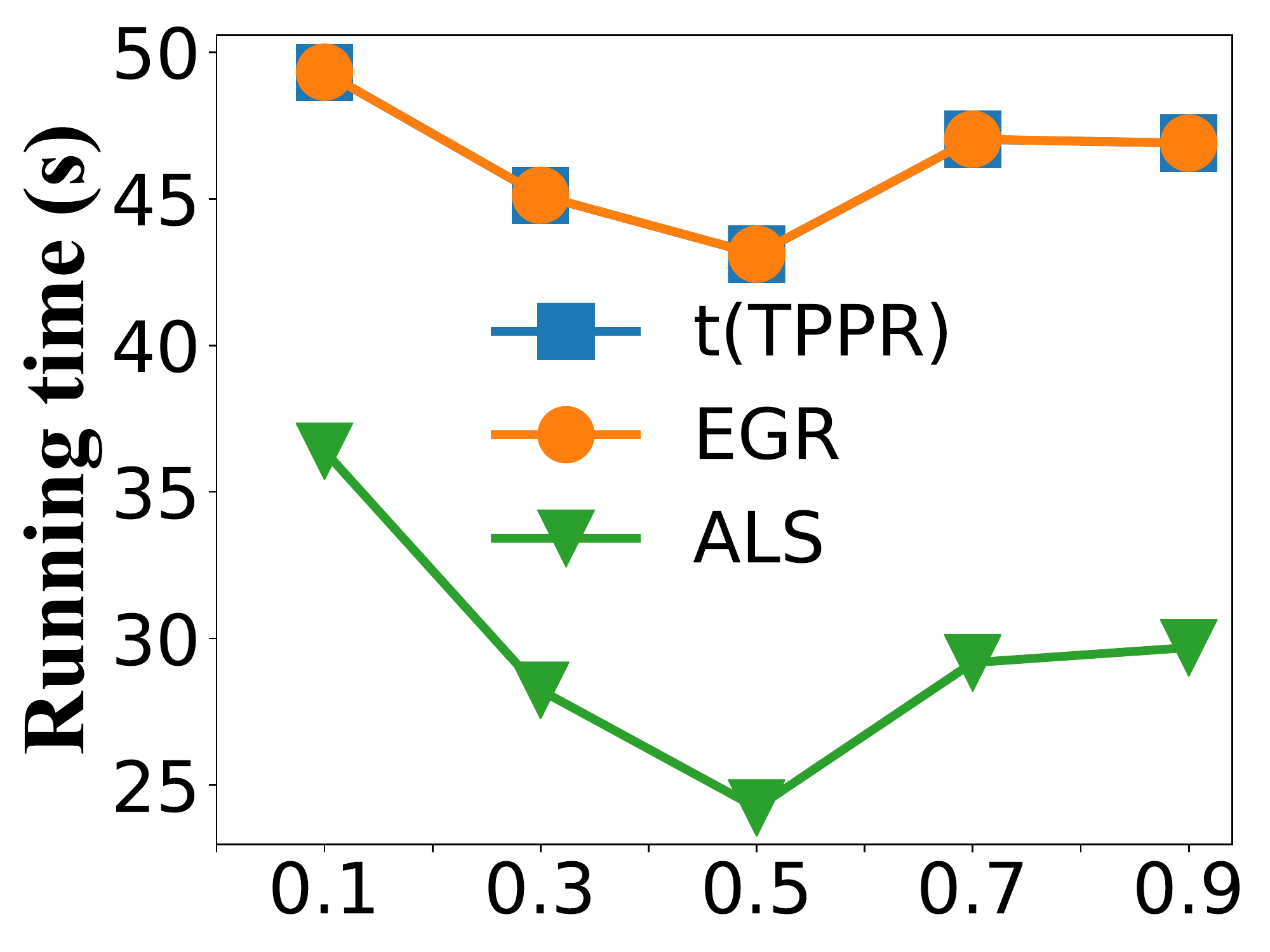}
	}\hspace{0.7cm}
	\subfigure[Facebook (vary rank)]{
		\includegraphics[width=3.5cm]{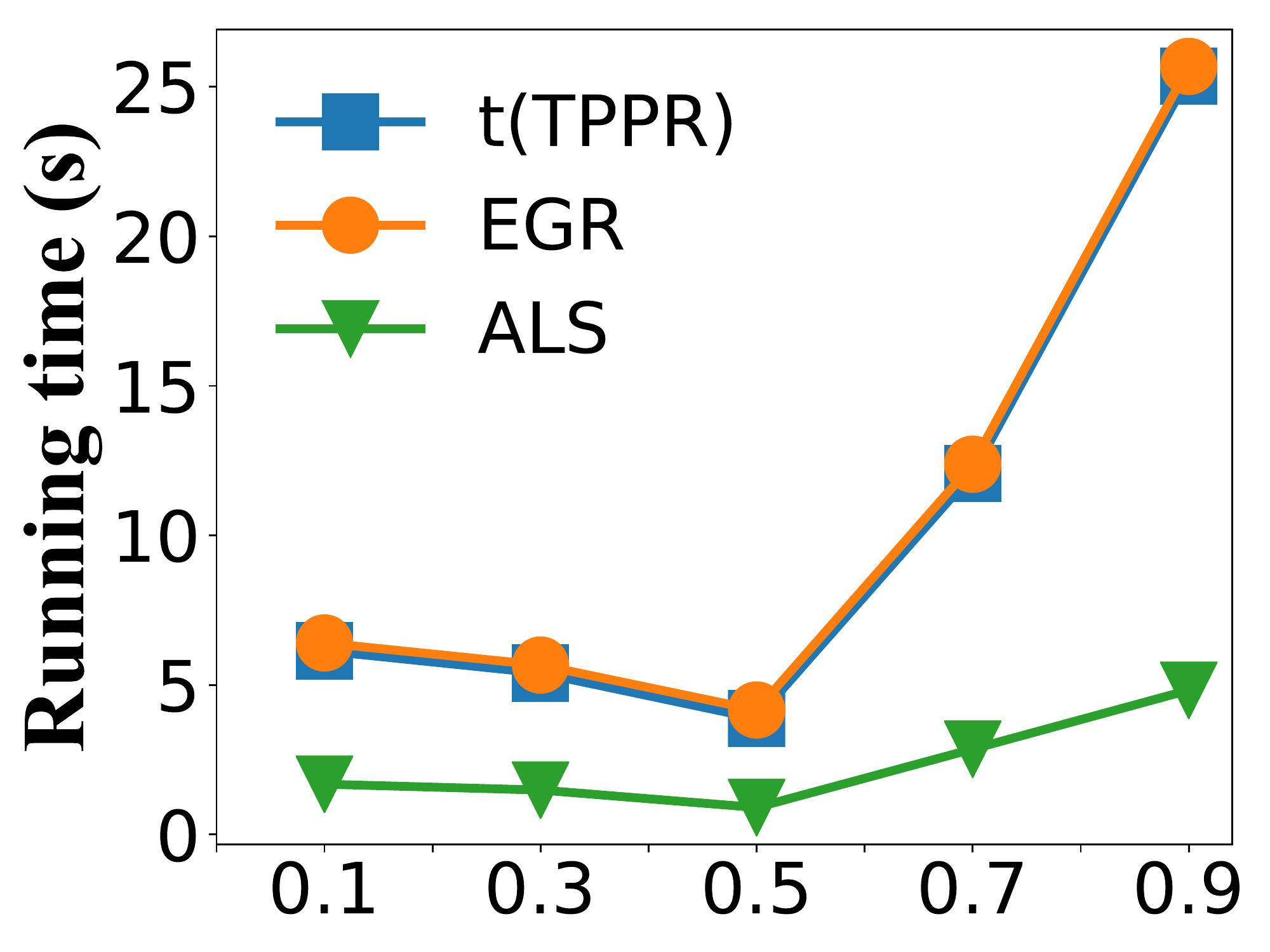}
	}\hspace{0.7cm}
	\subfigure[Enron (vary rank)]{
		\includegraphics[width=3.5cm]{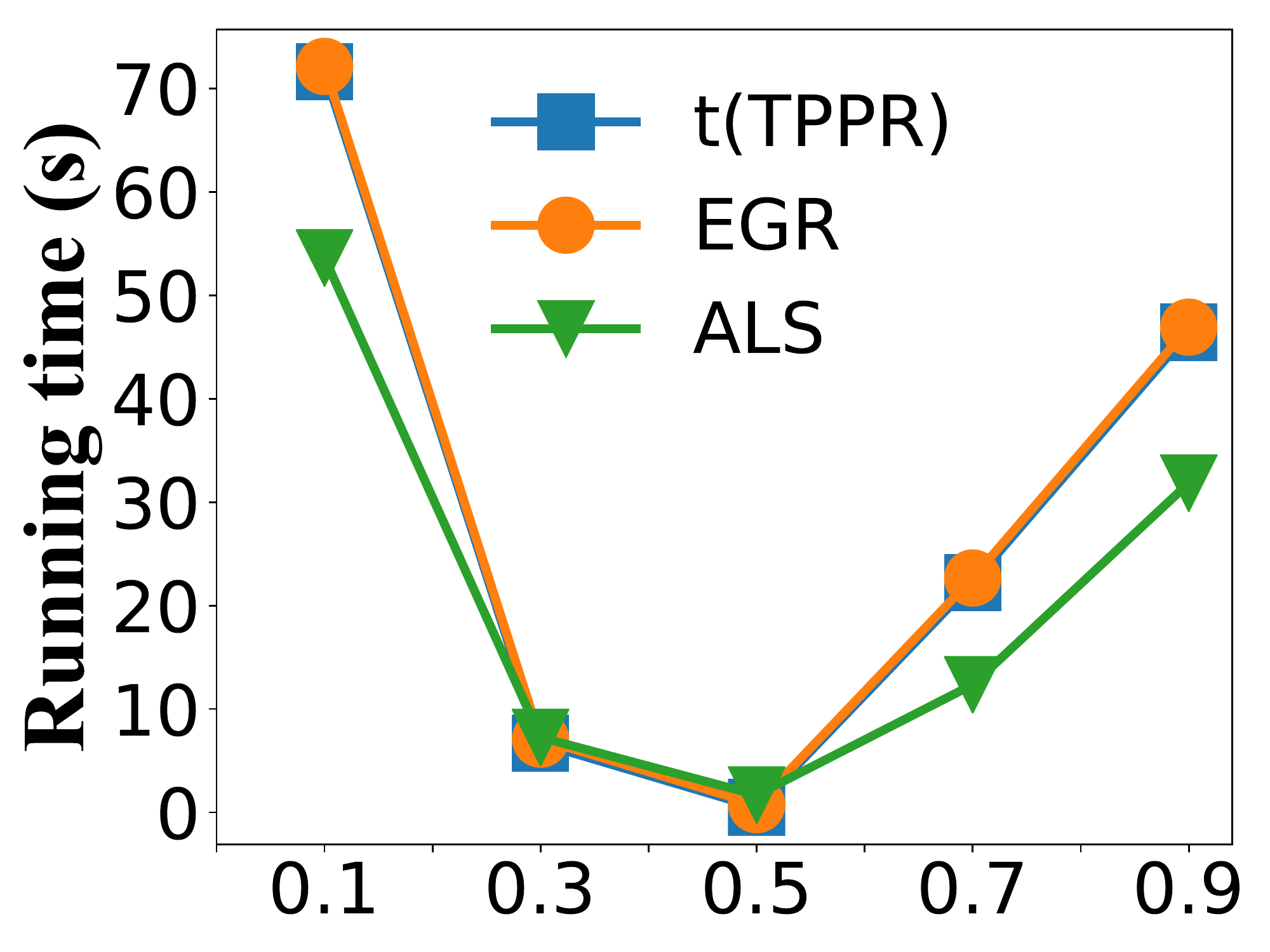}
	}\hspace{0.7cm}
	\subfigure[DBLP (vary rank)]{
		\includegraphics[width=3.5cm]{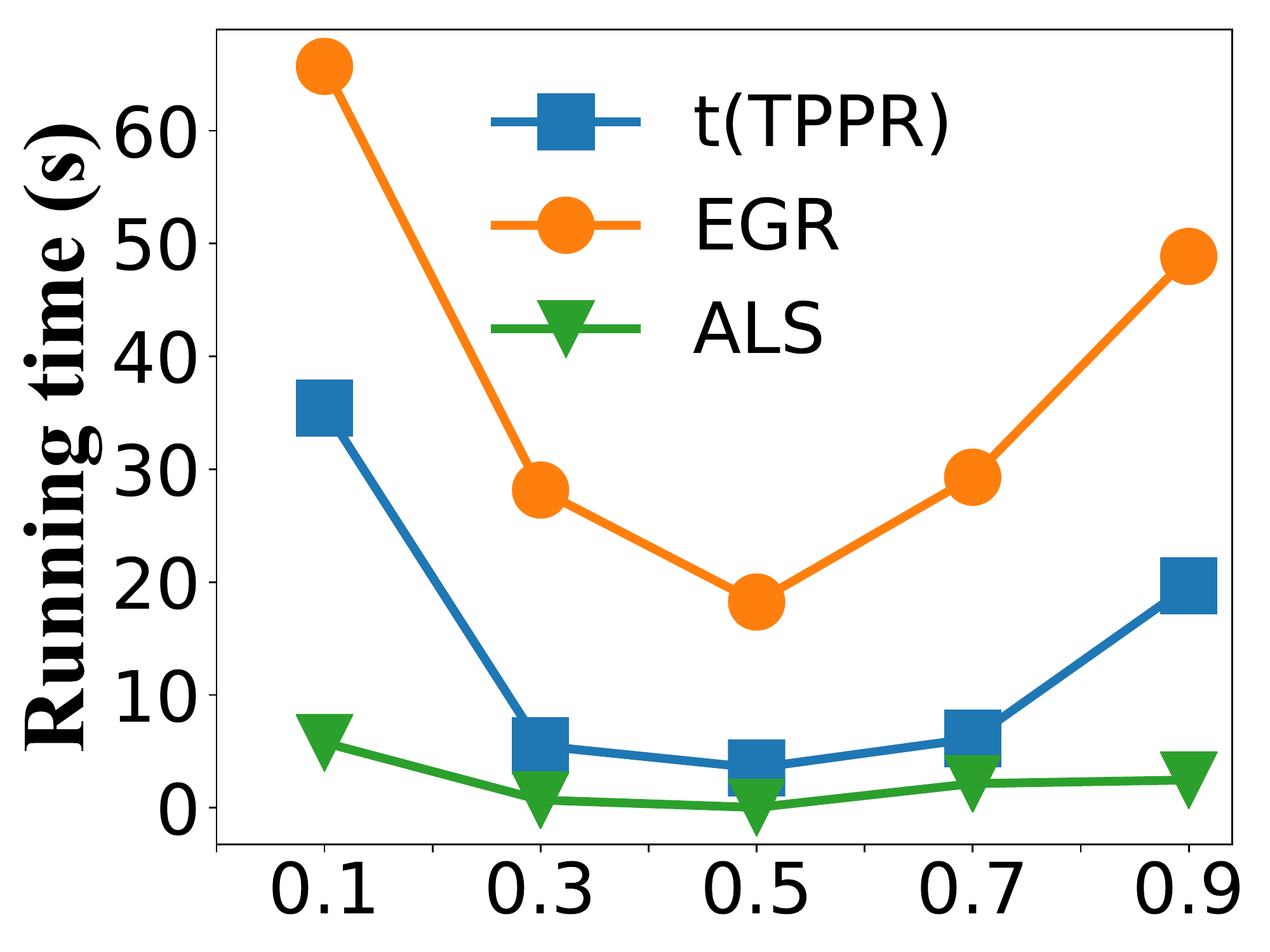}
	}\hspace{0.7cm}
	\subfigure[Rmin (vary $\alpha$)]{
		\includegraphics[width=3.5cm]{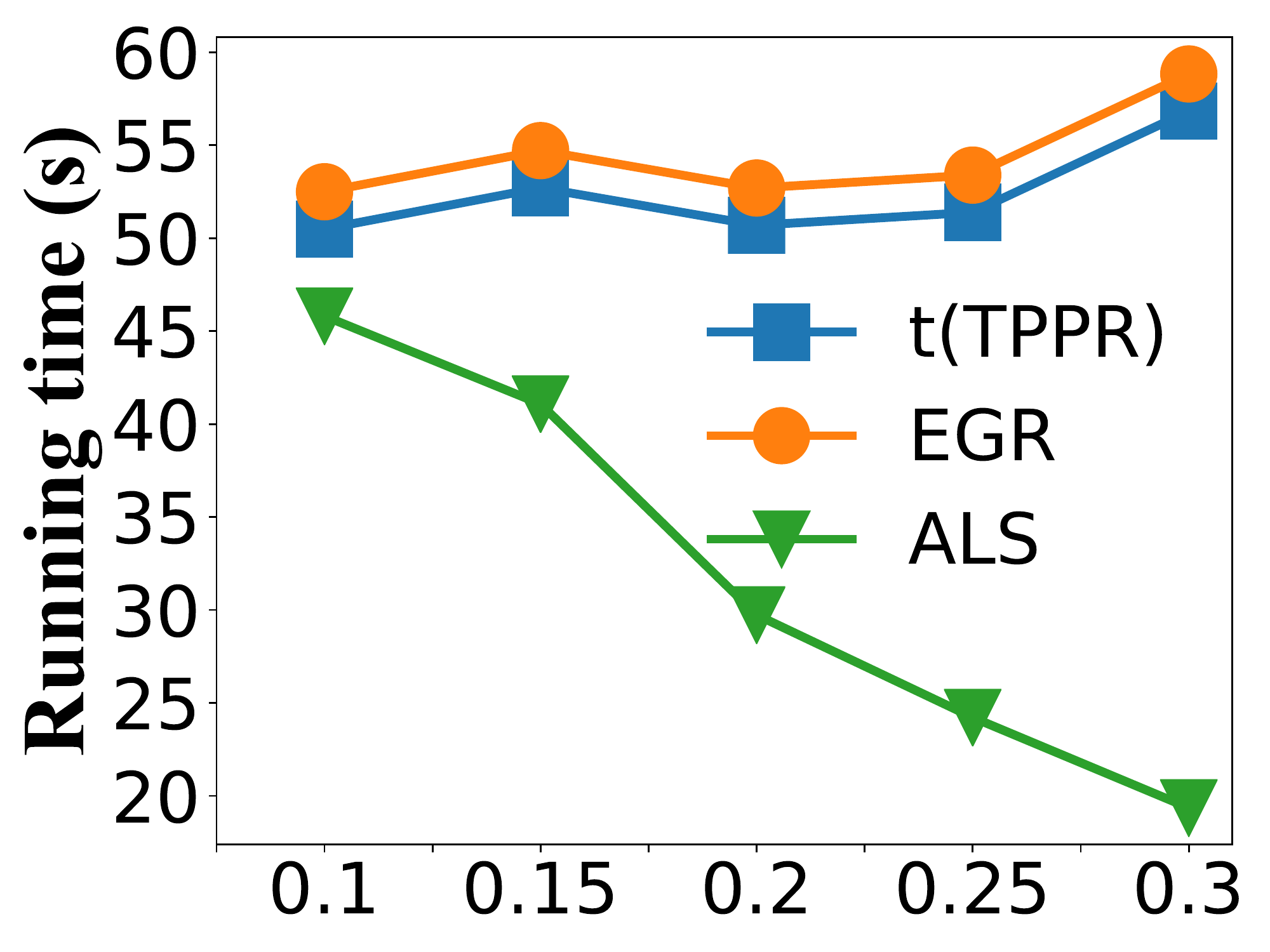}
	}\hspace{0.7cm}
	\subfigure[Facebook (vary $\alpha$)]{
		\includegraphics[width=3.5cm]{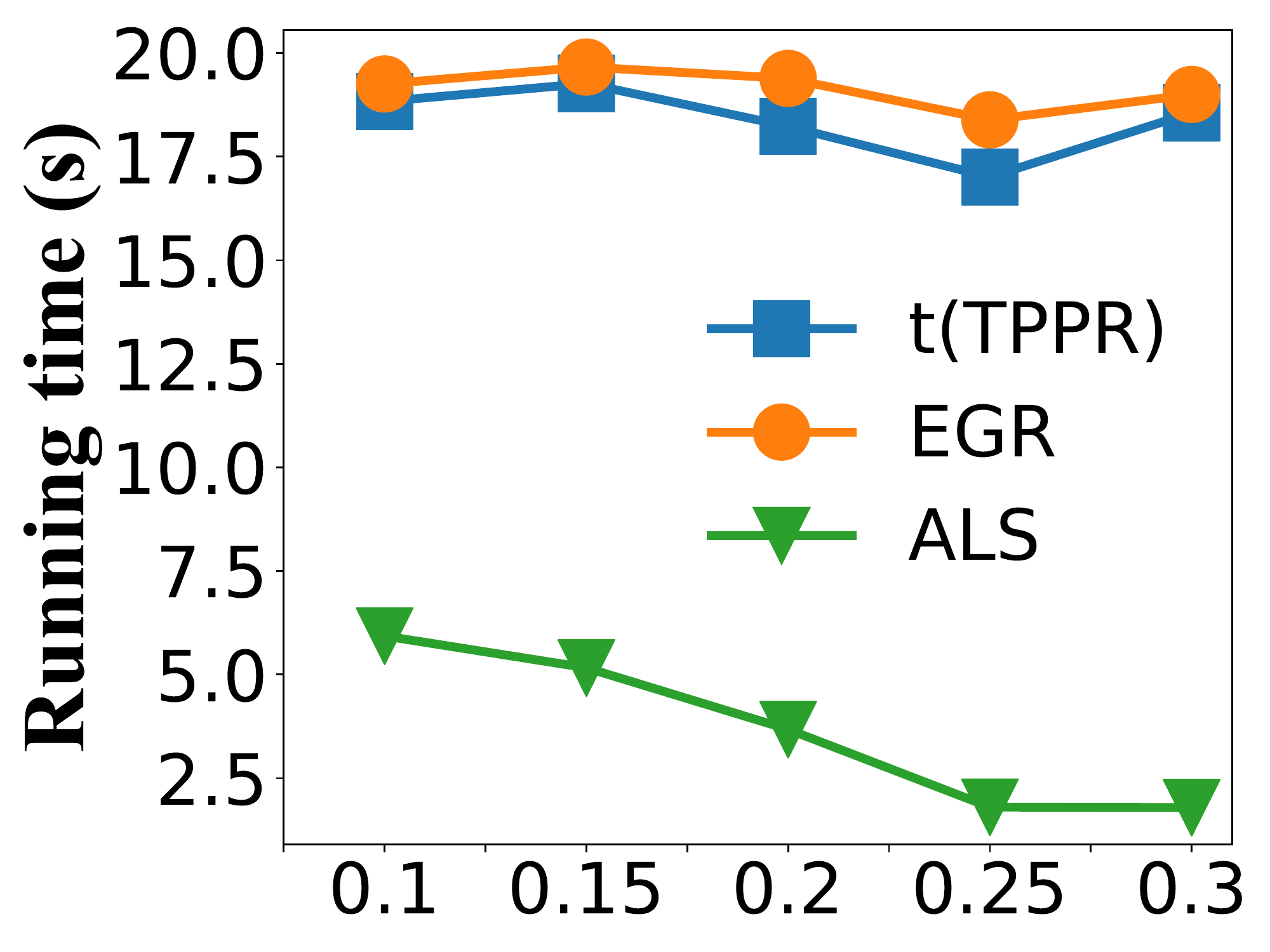}
	}\hspace{0.7cm}
	\subfigure[Enron (vary $\alpha$)]{
		\includegraphics[width=3.5cm]{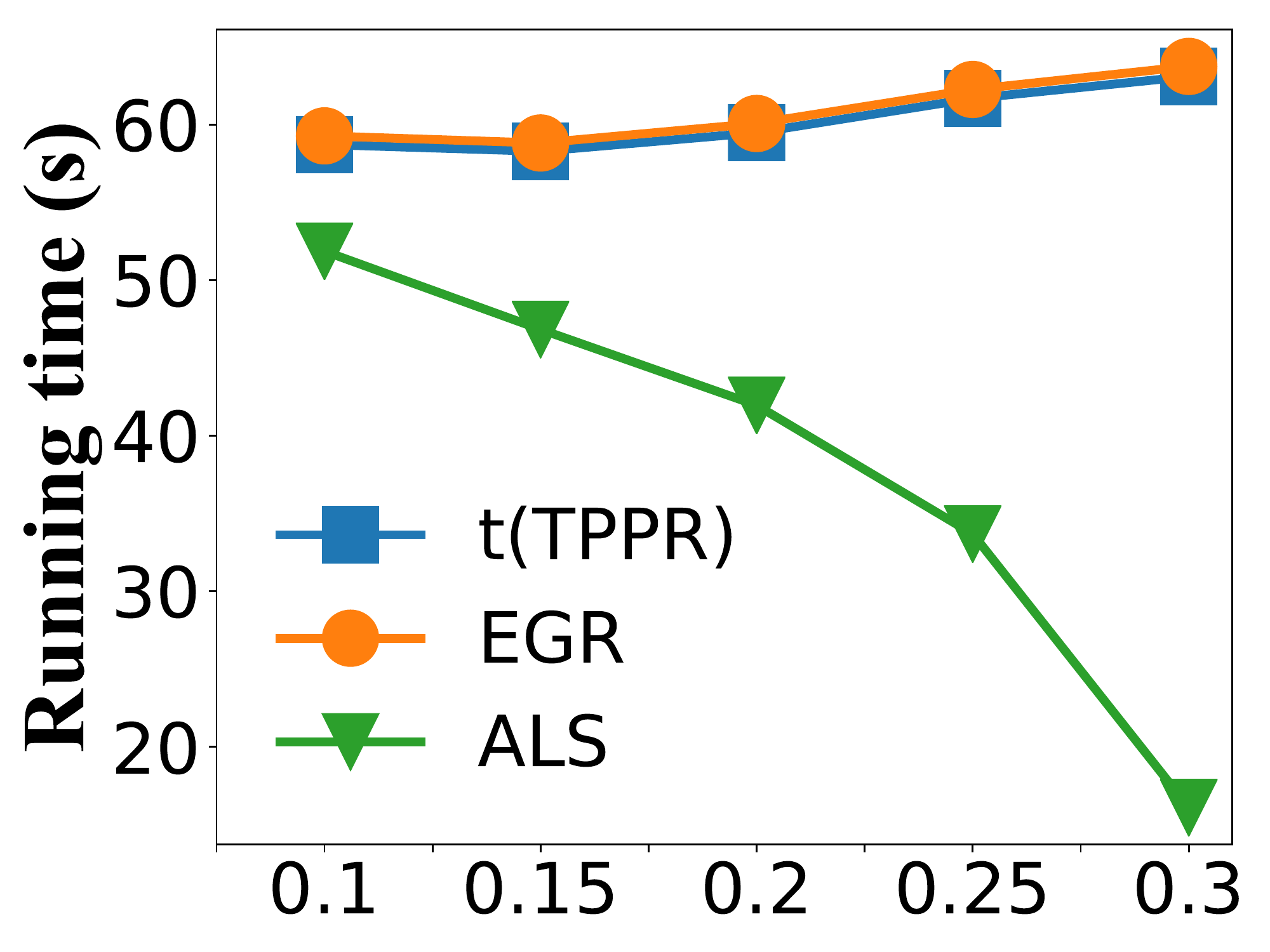}
	}\hspace{0.7cm}
	\subfigure[DBLP (vary $\alpha$)]{
		\includegraphics[width=3.5cm]{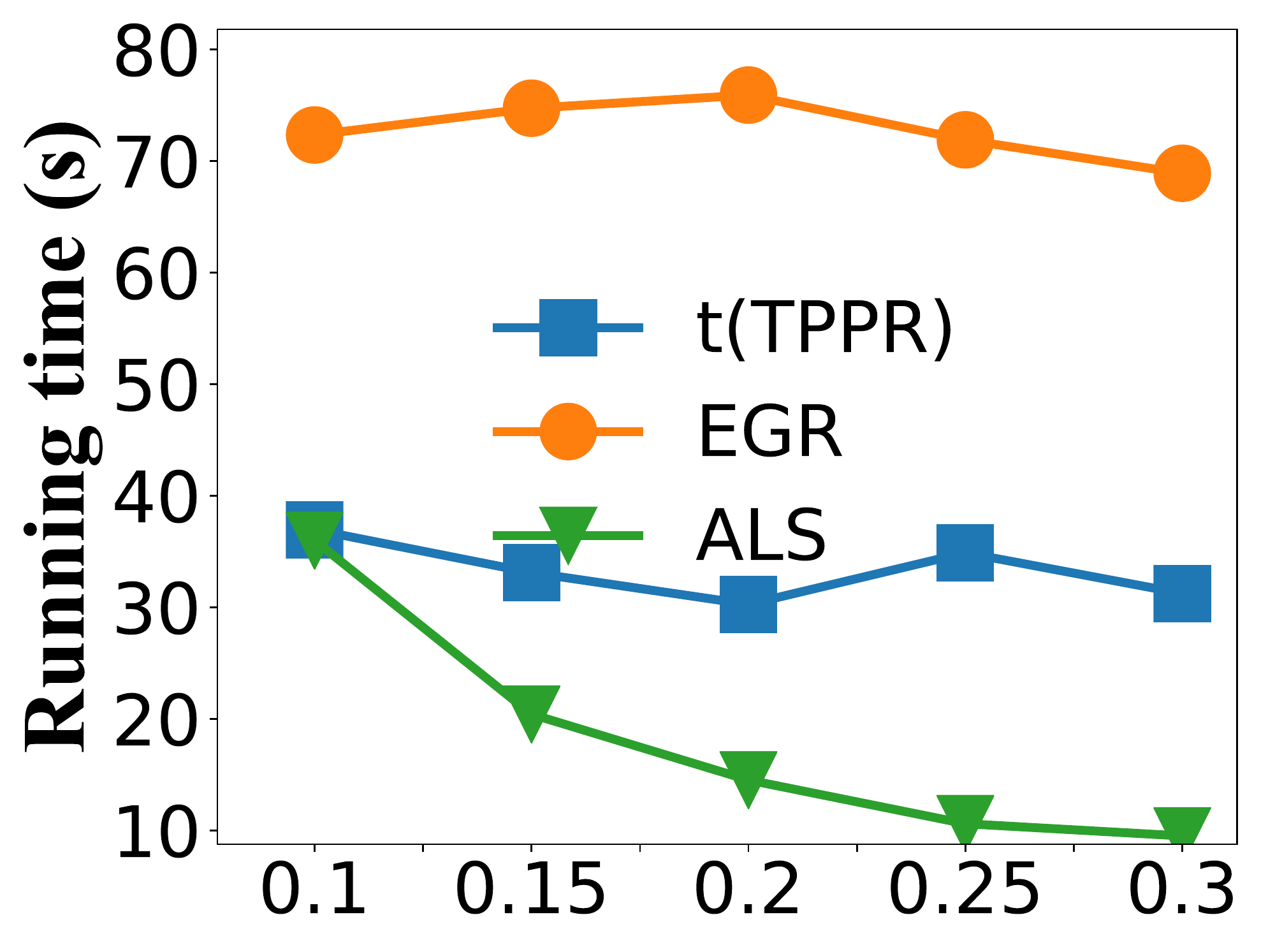}
	}
	\subfigure[Rmin (vary $\alpha$)]{
		\includegraphics[width=3.5cm]{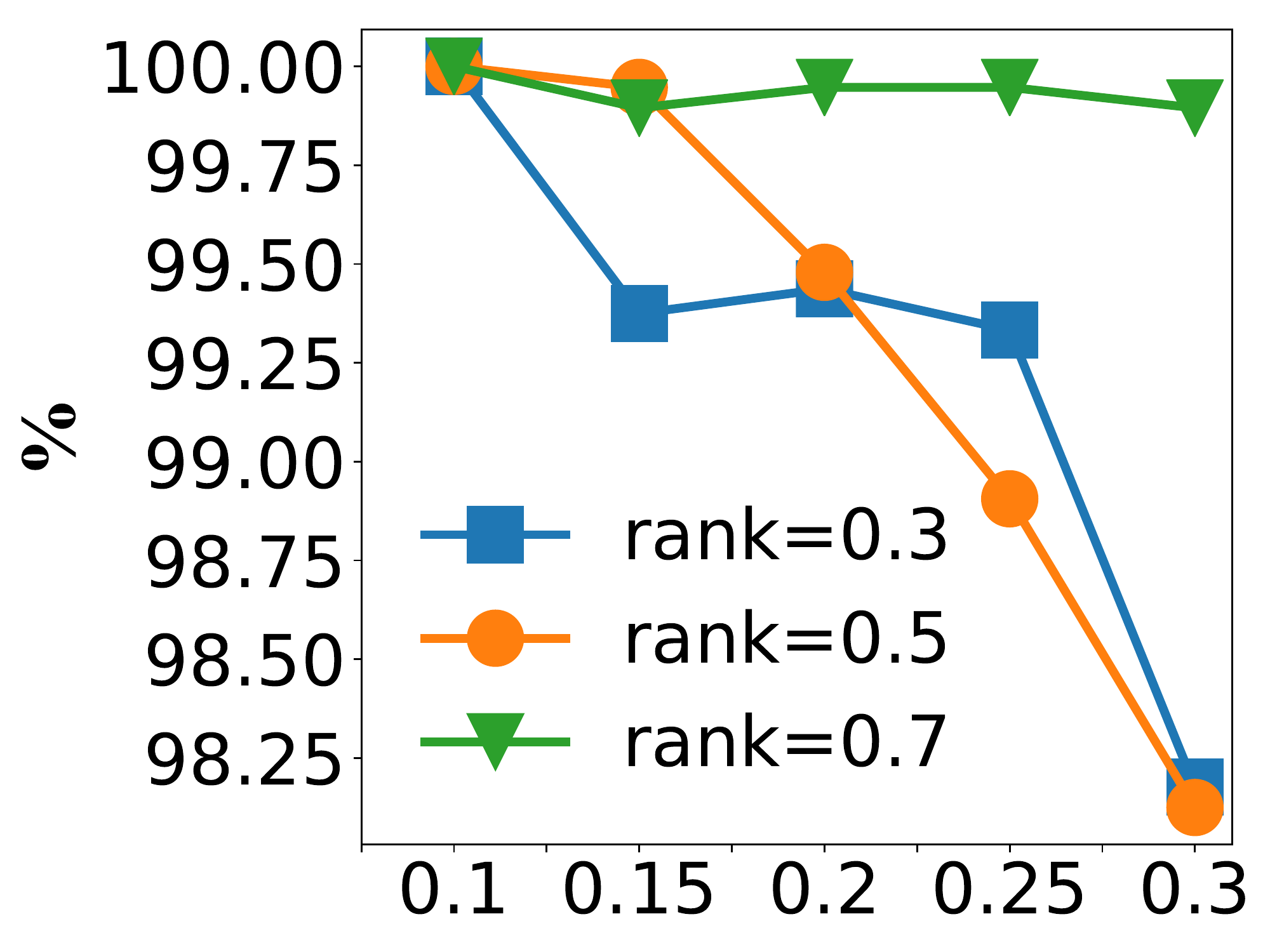}
	}\hspace{0.7cm}
	\subfigure[Facebook (vary $\alpha$)]{
		\includegraphics[width=3.5cm]{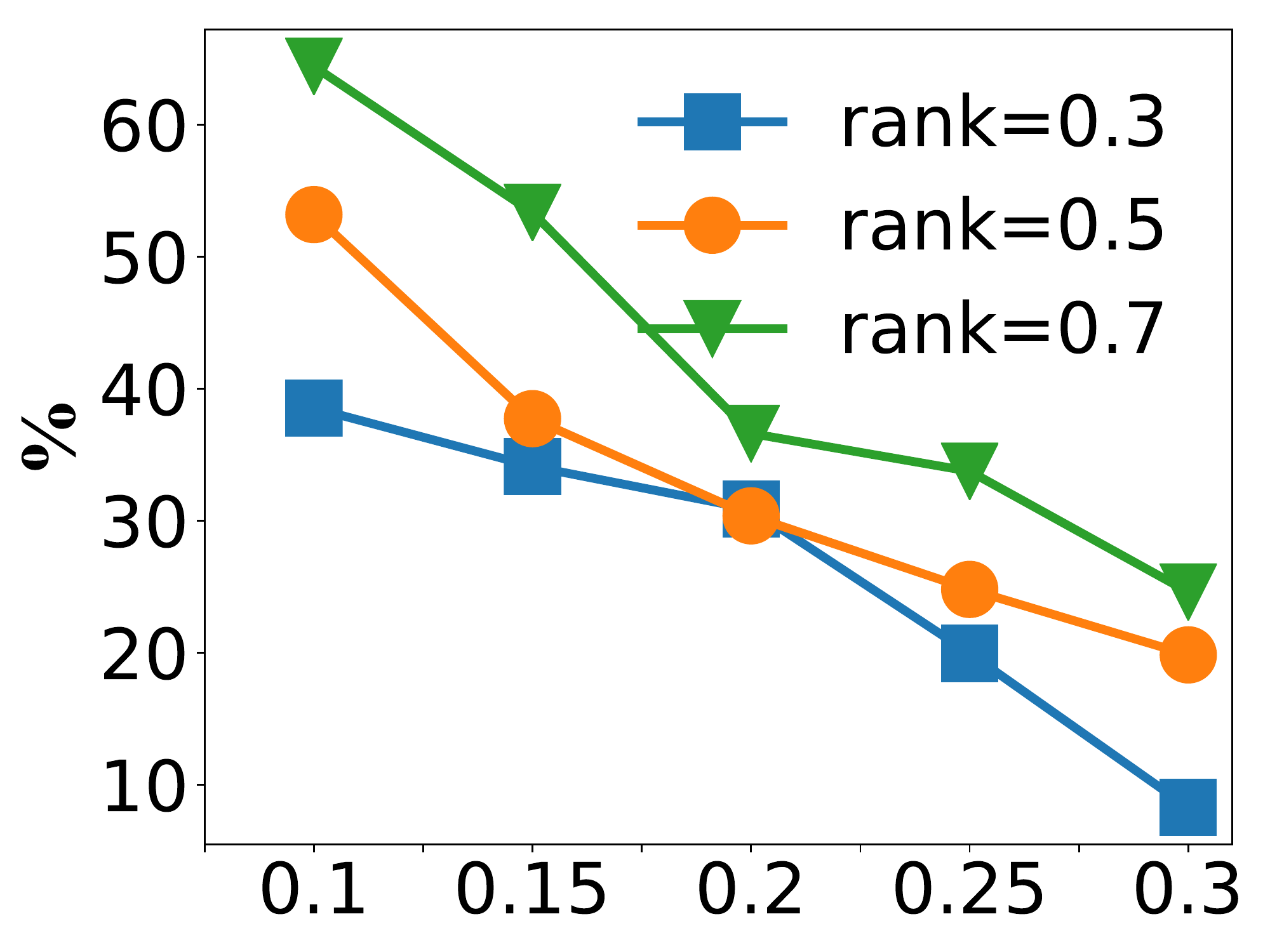}
	}\hspace{0.7cm}
	\subfigure[Enron (vary $\alpha$)]{
		\includegraphics[width=3.5cm]{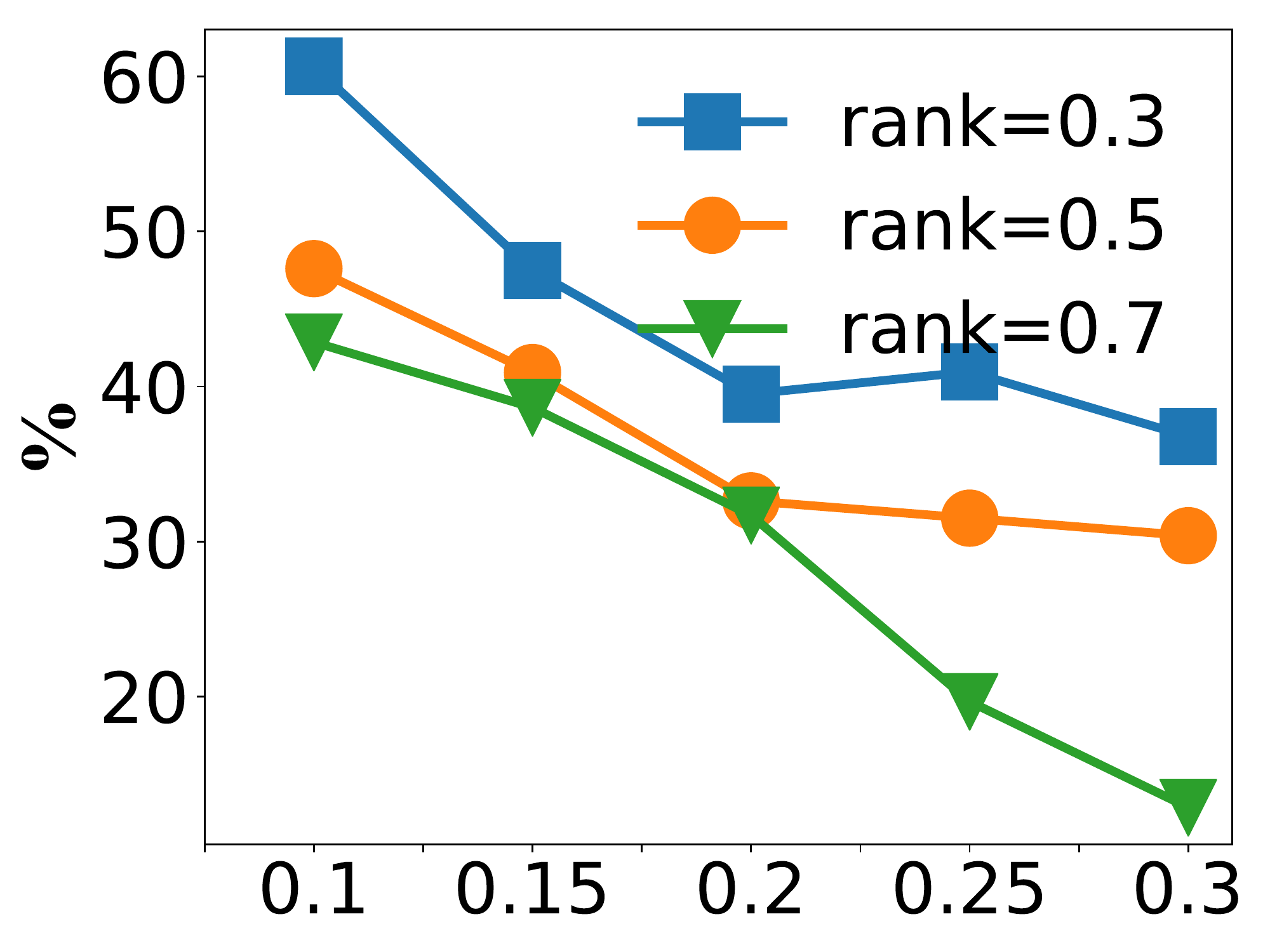}
	}\hspace{0.7cm}
	\subfigure[DBLP (vary $\alpha$)]{
		\includegraphics[width=3.5cm]{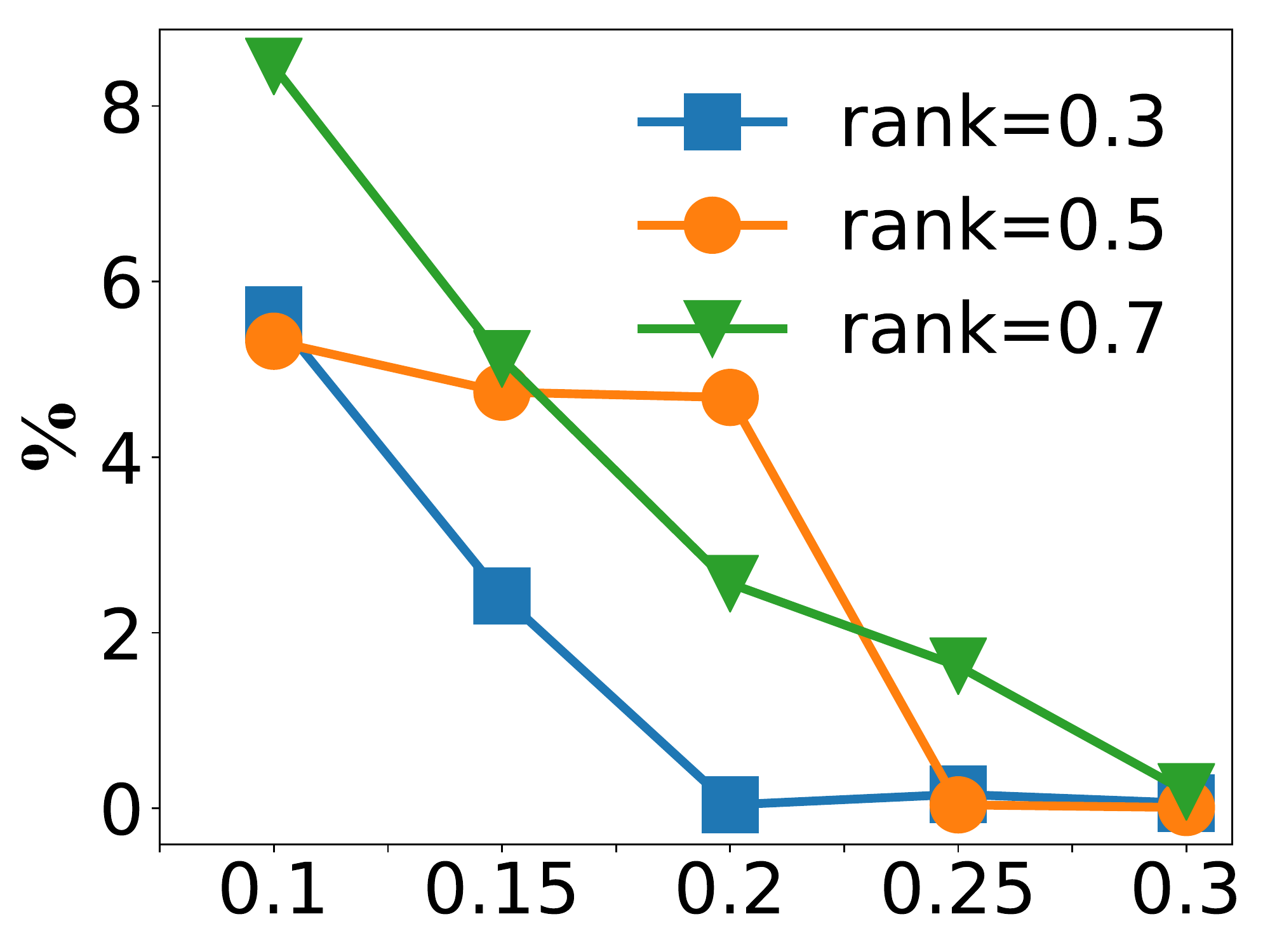}
	} \vspace{-0.2cm}
	\caption{ The efficiency of various algorithms with varying parameters} \vspace{-0.3cm}
	\label{fig:alg}
\end{figure*}

\begin{figure}[t!]
	\centering
	\subfigure[$\gr$]{
		\includegraphics[width=0.22\textwidth]{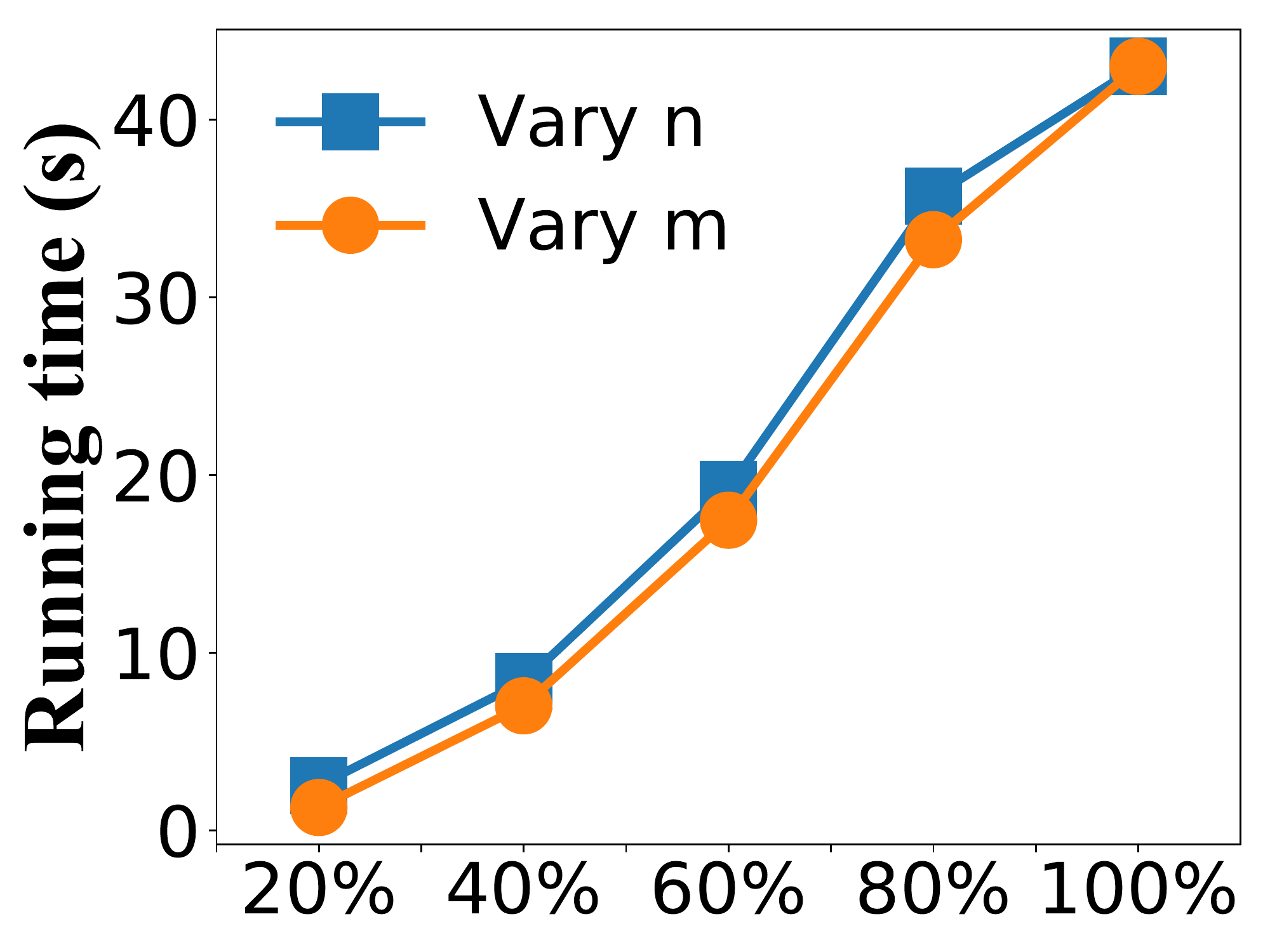}
		\label{fig:sca(a)}
	}
	\subfigure[$\ls$]{
		\includegraphics[width=0.22\textwidth]{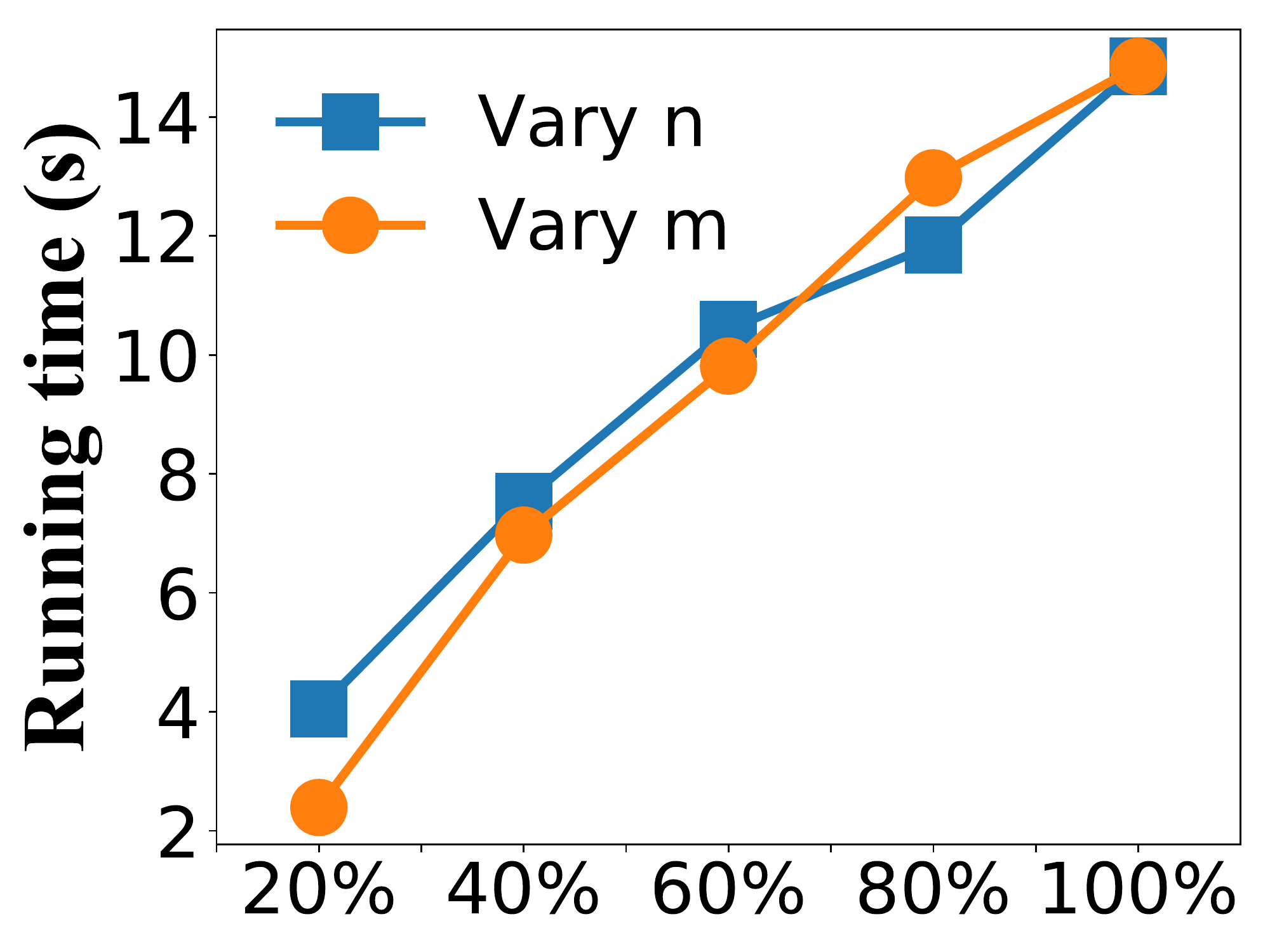}
		\label{fig:sca(b)}
	}\vspace*{-0.3cm}
	\caption{Scalability testing}
	\label{fig:sca} 
\end{figure}

\begin{table}[t!]
	\centering
	\caption{\small  Memory overhead of $\gr$ and $\ls$ (MB)}
	\footnotesize \vspace{-0.3cm}
	\scalebox{1}{
		\begin{tabular}{c|ccc}
			\toprule
			& \small Graph in memory& Memory of $\gr$& Memory of $\ls$ \\ \midrule
			Rmin & 9.291  &12.871  & 16.669\\
			Lyon & 34.780  &35.236  & 35.072\\
			Thiers & 62.381  &63.917  & 63.430  \\ 	
			Facebook & 149.538  &162.873 & 159.564  \\
			Twitter & 311.206 &393.152 & 331.207  \\
			Lkml &131.514  &148.0143  &182.439  \\
			Enron & 244.577 & 272.900  &247.764\\
			DBLP & 5190.925 &5758.229&5302.925  \\
			\bottomrule		
	\end{tabular}} \vspace{-0.5cm}
	\label{tab:memory}
\end{table}

\vspace{-0.2cm}

\stitle{Exp-2: Running time of various \textit{QTCS} algorithms with varying parameters.} In this experiment, we investigate how the parameter $\alpha$ affects the running time of different \textit{QTCS} algorithms. Additionally, we also study the effect of the temporal occurrence rank of query vertices. Let $\mathcal{T}_{u}=|\{t|(u,v,t) \in \mathcal{E}\}|$ be the temporal occurrence of the vertex $u$, which indicates how many timestamps are associated with $u$. Thus, we denote the temporal occurrence rank of a vertex
as 0.1 if its temporal occurrence is in the bottom 1\%- 10\%, and the temporal occurrence ranks
0.2, . . ., 0.9 are defined accordingly. For $\gr$
algorithm, we know that the search time is composed of Algorithm \ref{algor:tppr}
and the greedy removing process. We denote t(TPPR) as the time spent in Algorithm \ref{algor:tppr}. Fig. \ref{fig:alg} (a-h) show the results with varying rank and $\alpha$ on Rmin, Facebook, Enron, and DBLP.  Other datasets can also obtain similar results.  As can be seen, t(TPPR)  dominates the time of $\gr$ on all datasets except for DBLP. This is because the size of DBLP is relatively large, so it needs more time to perform the greedy removing process. Moreover, as shown in Fig. \ref{fig:alg} (a-d), the running time decreases first and then increases as rank increases, and the optimal time is taken when rank=0.5. Thus, we recommend users set the vertex with rank 0.5 as the query vertex for faster performance. On the other hand, by Fig. \ref{fig:alg} (e-h), we know that the running time of $\ls$ decreases with increasing $\alpha$. An intuitive explanation is that when $\alpha$ increases, the vertices have a higher probability of running temporal random walk around the query vertex, resulting in the locality of $\ls$ being stronger. As a result, the techniques of bound-based pruning and stop expanding are enhanced with increasing $\alpha$, thus more search spaces or vertices are pruned (Section  \ref{subsec:expanding}). Note that the running time of t(TPPR) and $\gr$ is stable with varying $\alpha$. This is because the time complexity of t(TPPR) and $\gr$ is independent of $\alpha$.

\stitle{Exp-3: The size of the expanded graph with varying parameters.} Fig \ref{fig:alg} (i-l) shows the size of the expanded graph obtained by the expanding stage (i.e., $|C|$ in Section  \ref{subsec:expanding}), divided by the size of the original graph, with varying rank and $\alpha$. We can see that the expanding stage obtains a very small graph. For instance, on Enron and DBLP, the number of vertices obtained by the expanding stage are only about 35$\%$ and 4$\%$ of the original graph, respectively. And the size of the expanded graph decreases with increasing $\alpha$. This is because the power of both bound-based pruning and stop expanding are enhanced when $\alpha$ increases. These results give some preliminary evidence that the proposed expanding algorithm (Section \ref{subsec:expanding}) is very effective when handling real-life temporal graphs. Moreover, we also observe that the size of the expanded graph is irregular as rank increases.

\begin{table*}[t!]
	\caption{\small Effectiveness of different methods. AVG.RANK is the average rank of each method across the testing datasets.} \vspace{-0.3cm}
	\centering
	\scriptsize
	\scalebox{1}{
		\begin{tabular}{c|cccccccccc}
			\toprule
			\multicolumn{1}{c|}	{TC/TD/MD} & Rmin& Lyon&Thiers& Facebook&Twitter& Lkml& Enron& DBLP & AVG.RANK\\
			\midrule
			
			\multirow{8}{*} \textit{CSM} & 0.33/0/0.35 & 0.87/0.42/0.76&0.92/0.14/0.49  & 0.43/0.08/0 &0.71/0.04/0& 0.68/ 0.06/0.07& 0.48/ 0.02/0 &0.72/ 0.30/0.01& 4/9/\underline{3}\\
			
			\textit{TCP}&0.92/0/0.10 & 1/0.38/0.55& 1/0.13/0.32  & 0.50/0.28/0.03 & 0.71/0.52/0.03& 0.36/0.08/0& 0.40/0.09/0 &0.68/0.40/0 &5/8/4\\
			\textit{PPR\_NIBBLE}  & 0.48/0/0.07 &0.50/0.51/0.28& 0.44/0.17/0.17 & 0.17/0.01/0 & \textbf{0.11}/0/0& 0.07/0/0&\textbf{0.27}/0.01/0 & 0.09/0/0& \underline{2}/10/9\\
			\textit{MPC}& 0.71/\textbf{0.29}/0.03 & 0.79/0.76/0.13& 0.82/\textbf{0.64}/0.02 & 0.50/\textbf{0.50}/0 & 1/\textbf{0.79}/0& 0.96/\textbf{0.22}/0& 0.94/\textbf{0.44}/0 & 0.84/\textbf{0.59}/0&9/\underline{\textbf{1}}/8\\
			\textit{PCore}& 0.75/0/0.24 & 0.55/0.52/0.30& 0.62/0.58/0.11 & 0.72/0.09/0 & 0.94/0.03/0& 0.76/0.02/0.11& 0.76/0.06/0.04 & 0.60/0.08/0&7/4/5\\	
			\textit{DBS}&  0.66/0.18/0.21 & 0.72/\textbf{0.77}/0.18& 0.52/0.56/0.07  & 0.67/0.41/0 & 0.95/0.66/0&0.95/0.21/0.15& 0.92/0.33/0.09 & 0.70/0.43/0&8/\underline{2}/7\\
			\textit{MTIS}&  0.67/0.02/0.13 & 0.98/0.43/0.02& 0.98/0.27/0  & 1/0.32/0 & 1/0.26/0&1/0/0& 1/0/0 & 1/0/0&10/7/10\\
			\textit{MSCS}&0.53/0.08/0.38 &0.58/0.54/0.49   &0.31/0.29/0.54 &0.72/0.18/0 & 0.72/0.12/0&0.72/0/0.01& 0.59/0/0 & 0.60/0/0&6/6/6\\
			\textit{QTCS\_Baseline}  & 0.30/0.01/0.43 & 0.56/0.52/0.58& 0.45/0.17/0.46& 0.49/0.07/0 &0.68/0/0& 0.53/0.03/0.06& 0.54/0.20/0.04 & 0.55/0.05/0&\underline{3}/5/\underline{2}\\
			
			our model \ & \textbf{0.01}/0.18/\textbf{0.73} & \textbf{0.44}/0.73/\textbf{0.81}& \textbf{0.16}/0.56/\textbf{0.67} &\textbf{0.11}/0.46/\textbf{0.15} & \textbf{0.11}/0.57/\textbf{0.08}& \textbf{0.02}/0.20/\textbf{0.25}& 0.32/0.33/\textbf{0.26} & \textbf{0.03}/0.40/\textbf{0.15}& \underline{\textbf{1}}/\underline{3}/\underline{\textbf{1}}\\
			
			\bottomrule	
	\end{tabular}} \vspace{-0.5cm}
	\label{table:metric}
\end{table*}

\vspace{-0.1cm}
\stitle{Exp-4: Scalability testing on synthetic datasets.} To test the scalability of $\gr$ and $\ls$, we first artificially generate eight temporal subgraphs by selecting randomly  20\%, 40\%, 60\% and 80\% vertices or edges from DBLP.  Subsequently, we test the runtime of $\gr$ and $\ls$ on these temporal subgraphs. Fig. \ref{fig:sca} shows the results. As can be seen,  $\gr$ and $\ls$ scales near-linear w.r.t. the size of the temporal subgraphs. These
results indicate that our proposed algorithms can handle massive temporal networks.

\stitle{Exp-5: Memory overhead of $\gr$ and $\ls$.} From Table \ref{tab:memory}, we can see that the memory overhead of $\gr$ and $\ls$ is less than twice that of the original graph. Moreover, we can also see that the memory overhead of $\ls$ is less than $\gr$ in six of the eight datasets. This is because $\ls$ is a local search algorithm, thus fewer vertices may be visited (Exp-3 also confirms this), which further results in less space used to store reserve and residue for estimating the \textit{TPPR} values. But, $\gr$ is a global algorithm, which needs to store $D[u]$ for computing the exact \textit{TPPR} values.
These results show that $\gr$ and $\ls$ can achieve near-linear space cost, which is consistent with our analysis in Section \ref{sec:gr} and \ref{sec:ls}.

\vspace{-0.2cm}
\subsection{Effectiveness testing}
\vspace{-0.2cm}
\stitle{Exp-6: Effectiveness of different methods.} Table \ref{table:metric} reports our results. For the \textit{TC} metric, we have: (1) our model achieves the best scores on seven of the eight datasets. This is because our model can mitigate the \qd issue  (Section \ref{subsection:query_drift}), resulting in that it can keep good temporal separability by removing out many temporal irrelevant vertices to the  query vertex (i.e., \emph{query-drifted vertices}). (2) \textit{PPR\_NIBBLE} and \textit{QTCS\_Baseline} are the runner-up and third-place, respectively, which shows that these random walk methods can also obtain better temporal separability. (3) \textit{MPC}, \textit{PCore}, \textit{DBS}, \textit{MTIS}, and \textit{MSCS} have the worst performance. This is because they focus on internal temporal cohesiveness but ignore the separability from the outside. For the \textit{TD} metric,  we have:  (1) \textit{MPC} and \textit{DBS} outperform other methods (but they have poor \textit{TC}), and our model is the third-place and slightly worse than \textit{MPC} and \textit{DBS}. This is because \textit{MPC} and \textit{DBS} respectively adopt the clique and density as the criteria of the community, which has a strong density in itself. (2) \textit{CSM}, \textit{TCP} and \textit{PPR\_NIBBLE} have the worst performance. This is because they are static methods that ignore the temporal dimension of the graph. For the \textit{MD} metric, we have: (1) our model achieves the best scores on all datasets while other models are almost zero on large datasets. (2) The gap between other models and our model is smaller on small datasets (i.e., Rmin, Lyon, and Thiers) than on large datasets. Thus, these results indicate that existing models cannot optimize our proposed objective function well, and  our model is much denser and more separable in terms of temporal feature than existing models.

\stitle{Remark.} Optimizing \textit{TD} and \textit{TC} simultaneously is very challenging (or even impossible). So, our model is a trade-off between them. The reasons can be explained as follows. (1) Although the \textit{TD} score of our model is slightly worse than the baselines (i.e., \textit{MPC} and \textit{DBS}), our algorithm is at least three orders of magnitude faster than the baselines.  Thus, our solutions achieve better runtime by losing a small amount of quality, which is particularly important for processing massive datasets. (2) As we all know, a good community not only requires the vertices in the community to be internally cohesive (\textit{TD}) but also separates from the remainder of the network (\textit{TC}). In Table \ref{table:metric}, we  see that \textit{MPC} and \textit{DBS} rank ninth and eighth in terms of \textit{TC}, respectively, but our model is the best.

\stitle{Exp-7: Quality comparison between \textit{EGR} and \textit{ALS}.} Here, we compare the community identified by the approximate local search algorithm $\ls$ with that identified by the exact greedy removing algorithm $\gr$. Specifically, we use the community derived by $\gr$ as the ground-truth for evaluating the quality of $\ls$. Table \ref{tab:quality} reports the results. Here, $\epsilon$ is the theoretically  approximation ratio of $\ls$ (Algorithm \ref{algor:reducing}) and $\epsilon^{*}=\min\{\rho_{H_1}(u)| u \in H_{1}\}/\min\{\rho_{H_2}(u)| u \in H_{2}\}$ is the \textit{true}  approximation ratio, where $H_1$ and $H_2$ are the communities identified by $\gr$ and $\ls$, respectively.  We have the following observations. (1) $\ls$ obtains better results than the theoretical $\epsilon$-approximation ratio.  In particular, the \textit{true} approximate ratio of $\ls$ is between 1 and 4. (2) $\ls$ obtains a good recall value, which indicates the community found by $\ls$ covers almost all  members of the ground-truth. (3) $\ls$ obtains relatively high scores of precision and F1-Score, which implies the size of the community returned by $\ls$ is close to the ground-truth. In summary, the approximate algorithm $\ls$ can find high-quality communities in practice.

\stitle{Exp-8:The quality of \textit{ALS} with various $\alpha$.} Fig. \ref{fig:quality} shows the \textit{true} approximation ratio $\epsilon^{*}$ and the minimum query-biased temporal degree \textit{MD} with various $\alpha$. Due to the space limit, we only report the results on Rmin, Facebook, Enron, and DBLP. Other datasets can also obtain similar results. As shown in Fig. \ref{fig:quality(a)}, $\epsilon^{*}$ increases first and then decreases as $\alpha$ increases. The reasons are: (1) when $\alpha$ is small, the target community is closer to the query vertex and the locality of \textit{ALS} is stronger. As a result,  the community found by \textit{ALS} matches the target community. (2) When $\alpha$ is large, the target community may be very small. Thus, once the community identified by \textit{ALS} is slightly different from the target community, it will cause  $\epsilon^{*}$ to drop rapidly. From Fig. \ref{fig:quality(b)}, we can observe that \textit{MD} increases with increasing $\alpha$. This is because when $\alpha$ increases, the \textit{TPPR} value tends to be concentrated near the query vertex and these \textit{TPPR} values are large, which leads to a larger \textit{MD} by Definition \ref{def:degree}.

\begin{table}[t!]
	\centering
	\caption{\small Quality comparison between \textit{EGR} and \textit{ALS}} \vspace{-0.3cm}
	\footnotesize
	\scalebox{1}{
		\begin{tabular}{c|ccccc}
			\toprule
			&$\epsilon$ &$\epsilon^*$	& Precision& Recall& F1-Score \\
			\midrule
			Rmin &3.350 & 1.657 & 0.646  &0.984  & 0.780  \\
			Lyon &2.745 & 1.302 & 0.848  &1.000  & 0.918 \\
			Thiers &3.439 & 1.489 & 0.772  &1.000  & 0.871  \\ 	
			Facebook  &7.410 & 1.751 & 0.504  &0.977  & 0.665\\
			Twitter & 5.160 & 1.584 & 0.266 &0.983 & 0.419 \\
			Lkml &7.601 &1.937 & 0.477 & 0.995  &0.645 \\
			Enron  & 8.580 & 1.863 & 0.575 & 0.964   &0.720 \\
			DBLP  &13.024 & 3.279 & 0.224  &0.950&0.362\\
			\bottomrule		
	\end{tabular}} 	\vspace{-0.5cm}
	\label{tab:quality}
\end{table}

\begin{figure}[t!]
	\centering
	\subfigure[$\epsilon^{*}$ (vary $\alpha$)]{
		\includegraphics[width=0.22\textwidth]{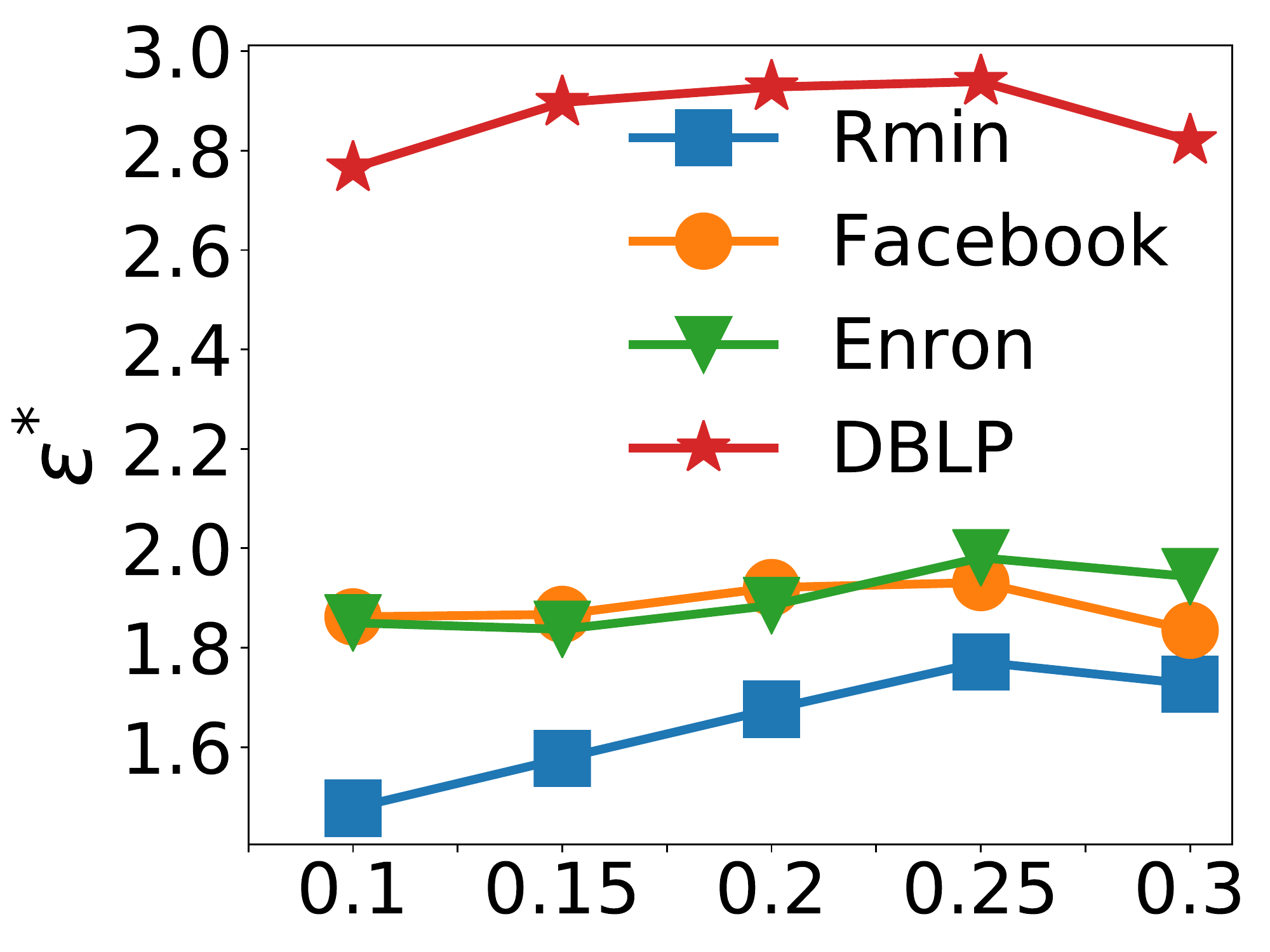}
		\label{fig:quality(a)}
	}
	\subfigure[\textit{MD} (vary $\alpha$)]{
		\includegraphics[width=0.22\textwidth]{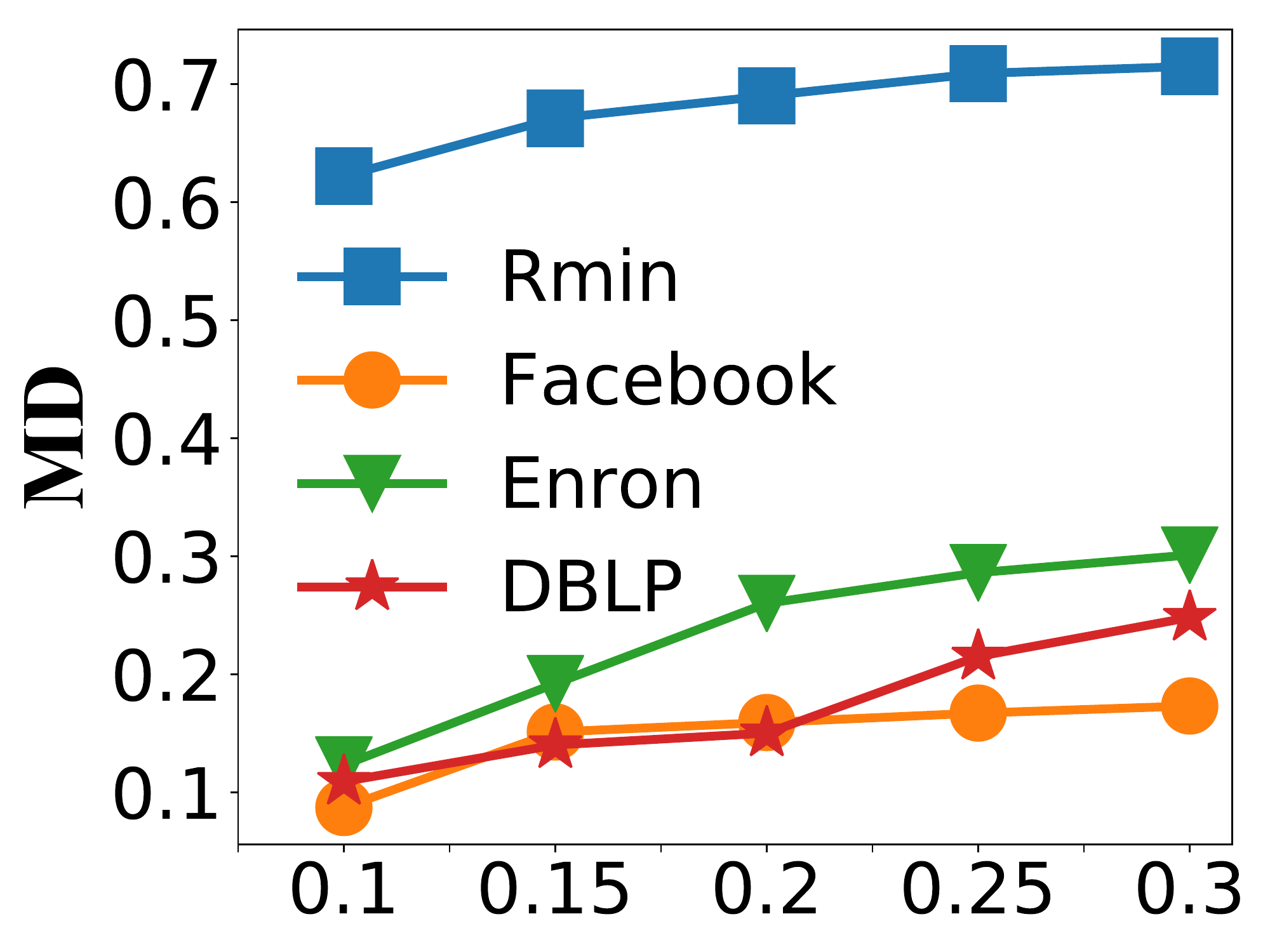}
		\label{fig:quality(b)}
	}\vspace*{-0.3cm}
	\caption{The quality of \textit{ALS} with various $\alpha$.}
	\label{fig:quality} \vspace{-0.5cm}
\end{figure}

\begin{figure*}[t!]
	\centering
	\subfigure[\textit{PCore}]{
		\includegraphics[height=0.14\textwidth,width=0.23\textwidth]{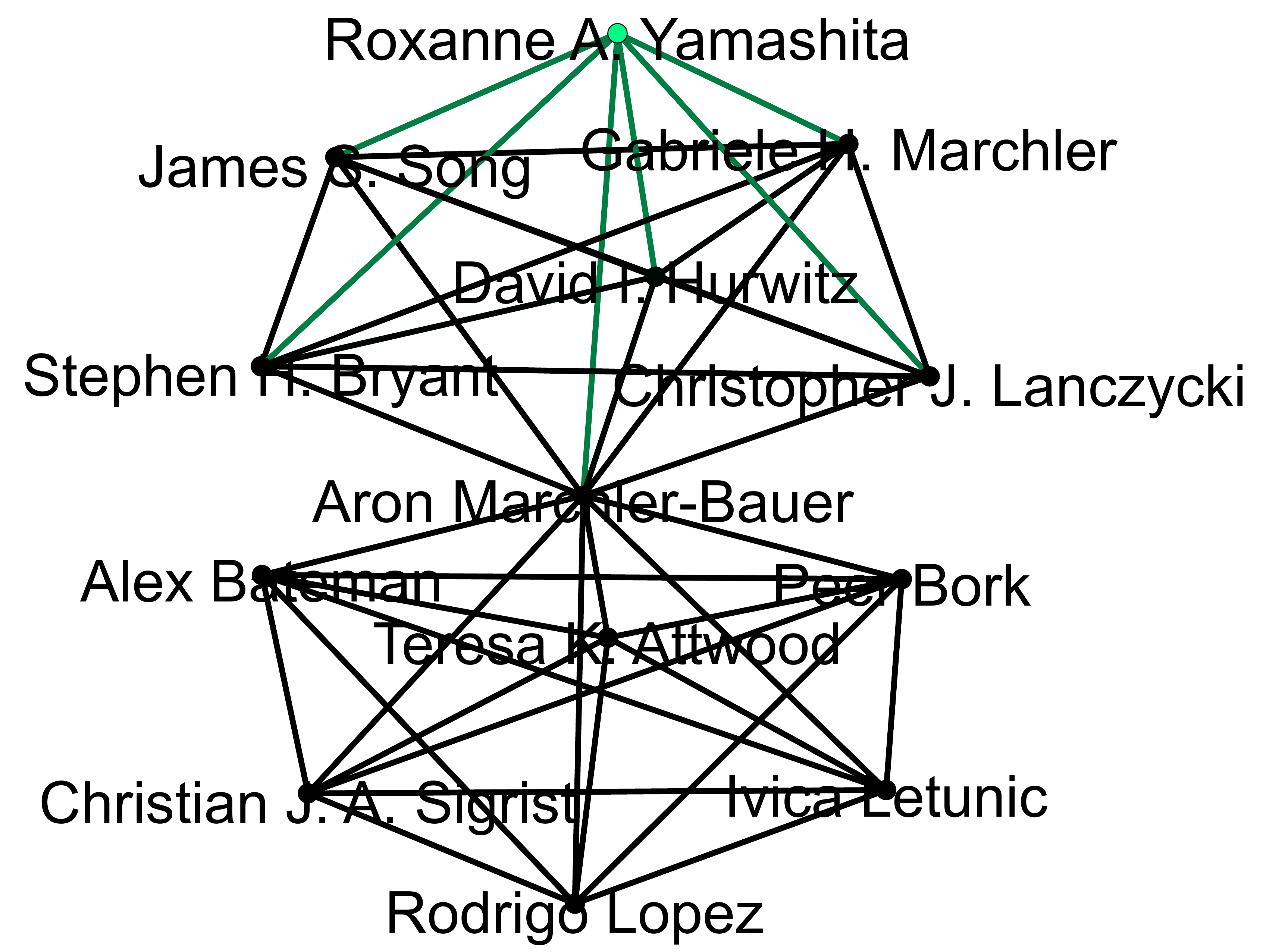}
	}
	\subfigure[\textit{MSCS}]{
		\includegraphics[height=0.14\textwidth,width=0.24\textwidth]{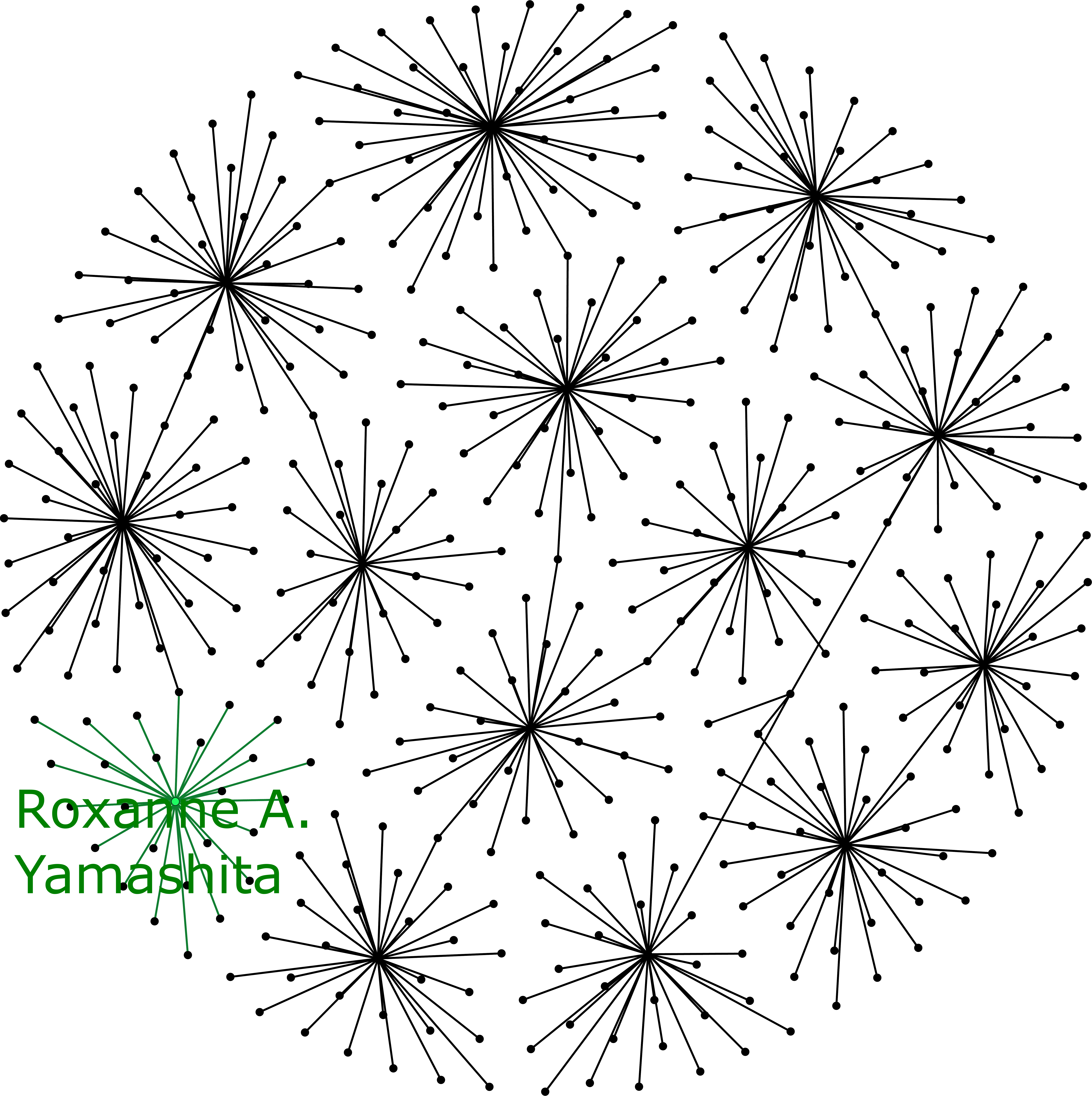}
	}
	\subfigure[\textit{Our model}]{
		\includegraphics[height=0.14\textwidth,width=0.22\textwidth]{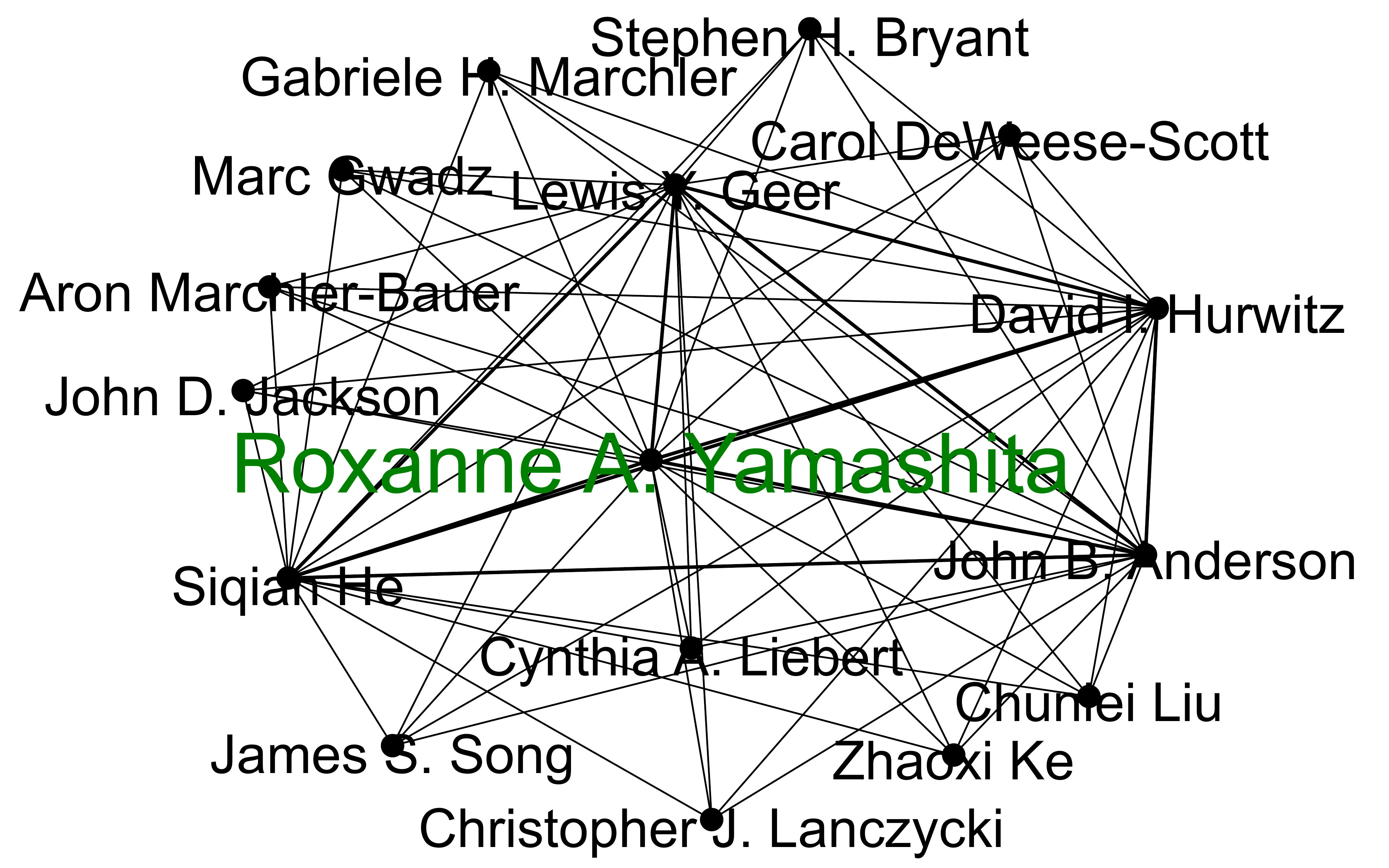}
	}
	\subfigure[\textit{Our model}]{
		\includegraphics[height=0.14\textwidth,width=0.23\textwidth]{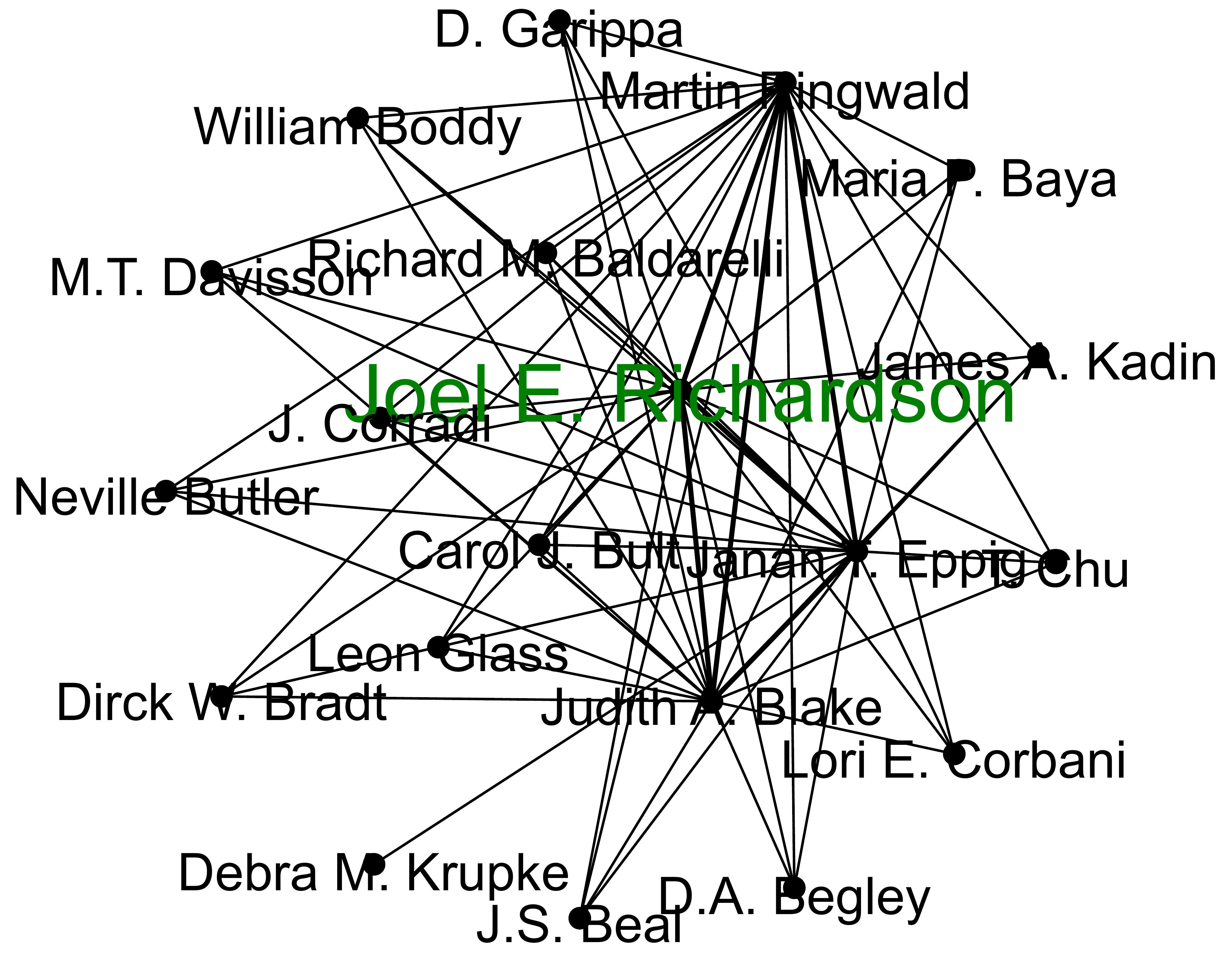}
	}
	\vspace{-0.3cm}
	\caption{Case studies on DBLP. (a-c) (resp. (d)) are the communities of Prof. Roxanne A. Yamashita (resp. Joel E. Richardson)}
	\label{fig:case} \vspace{-0.6cm}
\end{figure*}

 \stitle{Exp-9: Case studies on DBLP.} Here, we further show that our model can eliminate the \qd issue (Section \ref{subsection:query_drift}) while other models cannot eliminate it. Due to the space limit, we mainly report the results on \textit{PCore}, \textit{MSCS}, \textit{QTCS\_Baseline}, and our model. Similar results can also be obtained by the other models. Specifically, we choose Prof. Roxanne A. Yamashita or  Joel E. Richardson as the query vertex. Note that the community identified by \textit{QTCS\_Baseline} contains more than 1,000 authors (since it is too large to show in a figure, we do not visualize the community) that come from diverse research domains. This is because \textit{QTCS\_Baseline} considers structural cohesiveness and temporal proximity separately, which forces the result to include many vertices with poor temporal proximity to satisfy the structural cohesiveness. Thus, \textit{QTCS\_Baseline}  suffers from the \qd issue. On the other hand, as shown in  Fig. \ref{fig:case} (c), the community obtained by our model is a meaningful \emph{query-centered} temporal community and does not cause the \qd issue. This is because Roxanne A. Yamashita is \emph{centered} in the detected community and worked closely and frequently with other researchers. Besides, these researchers mainly investigate conserved sequence, amino acid sequence, and proteins,  which is consistent with Roxanne A. Yamashita. Thus, we can explain that this community is formed by their shared research interests and long-term cooperation with Roxanne A. Yamashita. However, from Fig. \ref{fig:case} (a), we can see that Roxanne A. Yamashita is marginalized, and the members on the upper and lower parts are connected by the hub vertex Aron Marchler-Bauer. Thus, the lower part is \emph{query-drifted vertices}. Additionally, by looking at the homepages of these researchers, we find that they come from different research backgrounds. Moreover, several important collaborators of Roxanne A. Yamashita in Fig. \ref{fig:case} (c) do not appear in Fig. \ref{fig:case} (a). Such as Stephen H. Bryant, Gabriele H. Marchler, and David I. Hurwitz (we can also see the importance of these three researchers to Roxanne A. Yamashita from https://www.aminer.cn/). By Fig. \ref{fig:case} (b), we can see that the community obtained by \textit{MSCS} is a connected subgraph composed of multiple stars. Furthermore, Fig. \ref{fig:case} (b) contains many \emph{query-drifted vertices}, which come from various backgrounds. Similar trends can also be observed in the community of Prof. Joel E. Richardson (due to the space limit, we only visualize the result of our model in Fig. \ref{fig:case} (d)). Since \textit{PCore} and \textit{MSCS} only consider the temporal cohesiveness but ignore the temporal proximity with the query vertex, they may find many temporal irrelevant vertices  to the query vertex for satisfying their cohesiveness, resulting in the query vertex being marginalized. Thus, \textit{PCore} and \textit{MSCS} also suffer from the \qd issue. In summary, these case studies further indicate that our model is indeed more effective than the other models to search \emph{query-centered} temporal communities.

%% file: related_work.tex
\vspace{-0.2cm}
\section{Related Work}\label{sec:relate}
\stitle{Community detection.} Existing studies mainly rely on structure-based approach to identify all communities from graphs, including  modularity optimization \cite{newman2004fast},  spectral analysis \cite{donetti2004detecting}, hierarchical clustering \cite{rokach2005clustering} and cohesive subgraph discovering \cite{DBLP:conf/icde/ChangQ19}. However, all these methods do not consider the temporal dimension of networks. Until recently, a few researches have been done on community detection over temporal networks \cite{9351686, DBLP:conf/kdd/YangYWCZL16, DBLP:journals/tsmc/LinYLWLJ22, DBLP:conf/icde/LiSQYD18, DBLP:conf/icde/MaHWLH17,DBLP:conf/icdm/RozenshteinBGST18, DBLP:conf/icde/QinLWQCY19, DBLP:journals/pvldb/ChuZYWP19, Chun}. For instance, Lin et al. \cite{DBLP:journals/tsmc/LinYLWLJ22} proposed the stable quasi-clique  to capture the stability of cohesive subgraphs. Ma et al. \cite{DBLP:conf/icde/MaHWLH17} studied the heavy subgraphs for detecting traffic hotspots. But, all these researches are query-independent, which are often costly to mine all communities. Thus, they cannot be  extended to perform online community search on temporal networks.

\stitle{Community search.} As a meaningful counterpart, community search has recently become a focal point of research in network analysis \cite{DBLP:conf/icde/HuangLX17, DBLP:journals/vldb/FangHQZZCL20}. For simple graphs, they aim to identify the subgraphs that contain the given query vertices and satisfy a specific community model such as $k$-core \cite{DBLP:conf/kdd/SozioG10,DBLP:conf/sigmod/CuiXWW14, DBLP:journals/datamine/BarbieriBGG15}, $k$-truss \cite{DBLP:conf/sigmod/HuangCQTY14, DBLP:conf/sigmod/LiuZ0XG20}, clique \cite{DBLP:conf/sigmod/CuiXWLW13, DBLP:journals/tkde/YuanQZCY18}, density \cite{DBLP:journals/pvldb/WuJLZ15}, connectivity \cite{DBLP:conf/kdd/TongF06, DBLP:conf/sigmod/RuchanskyBGGK15, DBLP:conf/cikm/RuchanskyBGGK17} and conductance \cite{DBLP:conf/focs/AndersenCL06, DBLP:conf/icdm/BianYCWLZ18, DBLP:conf/sigmod/YangXWBZL19}.  For instance, Sozio et al. \cite{DBLP:conf/kdd/SozioG10} introduced a framework of community search, which requires the target community is a connected subgraph containing query vertices and has a good score w.r.t. the proposed quality function. In particular, they used the  $k$-core as the quality function. Since the $k$-core is not necessarily dense, Huang et al. \cite{DBLP:conf/sigmod/HuangCQTY14} adopted a more cohesive subgraph model $k$-truss to model the community.  Recently, Wu et al. \cite{DBLP:journals/pvldb/WuJLZ15} observed the above approaches exist the free rider issue, that is, the returned community often contains many  redundant vertices. However,  our proposed \qd issue (Definition \ref{def:query_drift}) is more strict than the free rider issue. That is, if an objective function $f(.)$ suffers from the \qd issue, then $f(.)$ must  have the free rider issue, and vice versa is not necessarily true (see Section \ref{subsection:query_drift} for details). Besides, graph diffusion-based local clustering methods have also been considered. For example, Tong et al. \cite{DBLP:conf/kdd/TongF06} applied random walk with restart to measure the goodness score of any vertex w.r.t. the query vertices.   Andersen et al. \cite{DBLP:conf/focs/AndersenCL06} used Personalized PageRank to sort vertices and then executed a sweep cut procedure to obtain the local optimal conductance. However, the random walk used in these works is mainly tailored to static networks. Besides simple graphs, more complicated attribute information associated with vertices or edges also has been investigated, such as  keyword-based graphs \cite{DBLP:journals/pvldb/FangCLH16, DBLP:journals/pvldb/HuangL17, DBLP:conf/icde/LiuZZHXG20}, location-based social networks \cite{DBLP:journals/pvldb/FangCLLH17, DBLP:conf/kdd/ChenL0XY020}, multi-valued graphs \cite{DBLP:conf/sigmod/LiQYYXXZ18} and heterogeneous information networks \cite{DBLP:journals/pvldb/FangYZLC20, DBLP:journals/pvldb/JianWC20}. However, they ignore the temporal properties of networks that frequently appear in applications. Recently, two studies are done on temporal community search \cite{DBLP:conf/bigdataconf/TsalouchidouBB20, DBLP:journals/tkdd/GalimbertiCBBCG21}. But, they suffer from several defects (Section \ref{sec:introduction}, \ref{subsection:query_drift} and \ref{sec:experiments}).

 \stitle{Temporal proximity.} 
Node-to-node proximity is a fundamental concept in graph analysis, which captures the relevance between two nodes in a graph \cite{DBLP:journals/tkde/WuJZ16}. Perhaps, the most representative proximity model is the Personalized PageRank \cite{Page1999ThePC, DBLP:conf/wsdm/LofgrenBG16, DBLP:conf/sigmod/WeiHX0SW18} due to its effectiveness and solid theoretical foundation. However, this model only considers graph structural information and ignores the temporal properties. Recently, several studies were done on temporal proximity. For example,  \cite{DBLP:conf/wise/HuZG15,Rocha_2014} first converted the temporal graph into a weighted graph and then applied the traditional method over the weighted graph to define the temporal PageRank. These methods, however, only consider the temporal information of two directly-connected vertices, missing higher-order temporal and structural information. \cite{LV20191215} adopted the fourth-order tensor to represent the temporal network and calculated the eigenvector of the tensor to rank the vertices, which is inefficient for handling large graphs. The most related work to ours is \cite{DBLP:conf/pkdd/RozenshteinG16}. However,  \cite{DBLP:conf/pkdd/RozenshteinG16}  focuses on modeling the importance of vertices at a certain timestamp $t$. Thus, the method is to track the evolution of the importance of vertices. However, our \textit{TPPR} models the importance of vertices on the entire graph by non-trivially considering all timestamps. Thus, our \textit{TPPR} considers more structural and temporal information, which is more reasonable to capture temporal proximity. 